%% file: multiscalemms_arxiv.tex
\documentclass[final,11pt,onecolumn,a4]{IEEEtran}
\usepackage{url}
\usepackage{cite}
\newcommand{\cov}{\ensuremath{{\mathrm{cov}}}}
\usepackage{latexsym}
\usepackage{amsmath}
\usepackage{amstext,amsfonts,epsf,epsfig,pifont}
\usepackage{graphics,graphicx,bm,amsmath}
\usepackage{float}

\newtheorem{theorem}{Theorem}

\newtheorem{lemma}{Lemma}

\newtheorem{definition}{Definition}[section]

\newlength{\intwidth}

\input{texheader}

\begin{document}

\title{Multiscale Inference for High-Frequency Data}
\author{Adam Sykulski, Sofia C. Olhede and Grigorios A. Pavliotis
\thanks{A. Sykulski and G. A. Pavliotis are with the Department of Mathematics, Imperial College London, South Kensington Campus, SW7 2AZ London, UK (adam.sykulski03@imperial.ac.uk,
g.pavliotis@imperial.ac.uk),
tel: +44 20 7594 8564, fax +44 20 7594 8517.}
\thanks{S.~C.~Olhede is with the Department of Statistical Science, University College London, Gower Street,
London WC1 E6BT, UK (s.olhede@ucl.ac.uk), tel: +44 20 7679 8321, fax: +44 20 7383 4703.}}

\markboth{Stat. Sci. Report 290/Statistics Section Report TR-08-01}{Sykulski, Olhede \& Pavliotis: Multiscale Inference for High Frequency Data}

\maketitle

\begin{abstract}
This paper proposes a novel multiscale estimator for the integrated volatility of an
It\^{o} process, in the presence of market microstructure
noise (observation error). The multiscale structure of the observed process is represented frequency-by-frequency and the concept of the multiscale ratio is introduced to quantify the bias in
 the realized integrated volatility due to the observation error. The multiscale ratio is estimated from a single sample path, and a
frequency-by-frequency bias correction procedure is proposed, which simultaneously reduces variance. We extend the method to include correlated observation errors and provide the implied time domain form of the estimation procedure. The new method is
implemented to estimate the integrated volatility for the Heston and other models, and the improved
performance of our method over existing methods is illustrated by simulation studies.
\end{abstract}


\section{Introduction}\label{sec:intro}

Over the last few decades there has been an explosion of available data in
diverse areas such as econometrics, atmosphere/ocean science and molecular biology. It is essential to use this available data when developing and testing
mathematical models in physics, finance, biology and other disciplines. It
is imperative, therefore, to develop accurate and efficient methods for making
statistical inference in a parametric as well as a non-parametric setting.

Many interesting phenomena in the sciences are inherently multiscale
in the sense that there is an abundance of characteristic temporal and spatial scales. It is quite often the case that a simplified, coarse-grained model
is used to describe the essential features of the problem under investigation.
Available data is then used to estimate parameters in this reduced model \cite{MR2097022,Horenko2007,Horenko2008}.
This renders the problem of statistical inference quite subtle, since
the simplified models that are being used are compatible with the data only
at sufficiently large scales. In particular, it is not clear how and if the
high frequency data that is available should be used in the statistical inference
procedure.

On the other hand in many applications such as econometrics \cite{tsay} and oceanography \cite{Griffa_al_95} the observed data is contaminated by high frequency observation error. Statistical inference for data with a multiscale structure and for data contaminated by high frequency noise share common features. In particular the main difficulty in both problems is that the model that we wish to fit the data to is not compatible with the data at all scales. This is an example of a model misspecification problem \cite[p.~192]{Kut04}.

Parametric and non-parametric estimation for systems with multiple scales
and/or the usage of high frequency data has been studied quite extensively
in the last few years for the two different types of models. First, the problem
of estimating the integrated stochastic volatility in the presence of high
frequency observation noise has been considered by various
authors~\cite{AitMykZha05b,Zhang+2005}. Similar models and inference problems
have also been studied in the context of oceanic transport~\cite{Griffa_al_95}.
It was assumed in~\cite{AitMykZha05b,Zhang+2005} that the observed process
consists of two parts, an It\^{o} process $X_t$ (i.e. the solution of an SDE, which is a semimartingale) whose
integrated stochastic volatility (quadratic variation $\langle X_t , X_t \rangle$) we want to estimate, and a high
frequency noise component $\varepsilon_{t_j}$
\begin{equation}\label{e:sahalia}
Y_{t_j} = X_{t_j} + \varepsilon_{t_j}.
\end{equation}
$\{Y_{t_j} \}_{j=1}^{N+1}$ are the sampled observations.
The additional noise $\{\varepsilon_{t_j}\}_{j=1}^{N+1}$ was used to model market
microstructure.

It was shown for the model of~\eqref{e:sahalia} that
using high frequency data leads to asymptotically biased
estimators. In particular if all available data is used for the estimation of the quadratic variation of $X_t$ then {\em the realized integrated volatility}
$\left[ Y_t,Y_t\right]$ converges to the aggregated variance of the differenced observation noise.
Subsampling is therefore necessary for the accurate
estimation of the integrated volatility.
An algorithm for estimating the integrated volatility
which consists of subsampling at an optimal sampling rate combined with
averaging and an appropriate debiasing step was proposed
in~\cite{AitMykZha05b,Zhang+2005}. Various other estimators were suggested in \cite{Zhang+2005,tsay,Sircar03,Hansen06} for processes contaminated by high frequency nuisance structure.

Secondly, parameter estimation for fast/slow systems of SDEs for which a limiting SDE for the slow variable can be rigorously shown to exist was studied
in~\cite{PavlPokStu08, PapPavSt08, PavlSt06}. In these papers the problem of making inferences for the parameters of the limiting (coarse-grained) SDE for the slow variable from observed data generated by the fast/slow system was examined. It was
shown that  the maximum likelihood estimator is asymptotically biased. In order to correctly  estimate the
parameters in the drift and the diffusion coefficient of the coarse-grained
model from observations of the slow/fast system using maximum likelihood,
subsampling at an appropriate rate is necessary. The subsampling rate depends
on the ratio between the characteristic time scales of the fast and slow
variables. A similar problem, with no explicit scale separation, was studied
in \cite{CotterPavl09}.

All of the papers mentioned above propose inference methods in the time domain.
Yet, it would seem natural to analyse multiscale and high frequency properties of the data in the frequency domain. Most
of the time domain methods can be put in a unified framework as linear filtering
techniques,
i.e. as a convolution with a linear kernel, of some time-domain quadratic function of the data. The understanding of
these methods is enhanced by studying them directly in the frequency domain, as convolutions in time are multiplications in frequency.
Fourier domain estimators of the integrated
volatility have been proposed for observations devoid of microstructure features, see
\cite{Hansen05,Barucci02,Malliavin+2002}. Fourier domain estimators have also been used
for estimating noisy It\^{o} processes (i.e. processes of the form~\ref{e:sahalia}), see \cite{Mancino+tech,Reno2005,Reno2008}, based on smoothing the time domain quantities by using only a limited number of frequencies in the reconstruction.

The bias in the
realized integrated volatility of the observed process $Y_{t_j}$ due to the observation
noise $\varepsilon_{t_j}$  can be understood directly in the frequency domain, since the energy
associated with each frequency is contaminated by the microstructure noise process. This
bias is particularly damaging at high frequencies. In this article we propose a
frequency-by-frequency de-biasing procedure to improve the accuracy of the estimation of
the integrated volatility. The proposed estimation method can also be viewed in the time domain as
smoothing the estimated autocovariance of the increments of the process, but where the implied time domain smoothing kernel is itself estimated from the observed process.

In this paper we will consider a regularly sampled It\^{o} process with additive white noise $\varepsilon_{t_j}$
superimposed upon it at each observation point $t_j$, {\em cf} \eqref{e:sahalia}. The It\^{o} process satisfies an SDE of the form
\begin{equation}
\label{eq:ito}
dX_t=\mu_t dt+\sigma_t dB_t,\quad X_0=x_0.
\end{equation}
$B_t$ denotes a standard one dimensional Brownian motion and $\mu_t$, $\sigma_t$ are (in general) It\^{o} processes, see for example the Heston model which is studied in
Section~\ref{sec:mcs}. The Brownian motions driving the three  It\^{o} processes
can be correlated. The observations and the process
are related through
\begin{equation}
\label{e:model} Y_{t_j}=X_{t_j}+\varepsilon_{t_j},\quad j=1,2, \dots, N+1, \quad
t_j:=\frac{j-1}{N}T=(j-1)\Delta t.
\end{equation}
We assume the data is regularly spaced. The length of the path $T$ is fixed. The additive noise $\{ \varepsilon_{t_j} \}_{j=1}^{N+1}$ is initially taken to be a white noise process with variance $ \sigma^2_{\varepsilon}$, and it is assumed to be independent of the noise that drives the It\^{o} process $X_t$. Our main objective is to estimate the  integrated volatility,
$\langle X, X\rangle_T=\int_{0}^T \sigma^2_t\,dt$ of the It\^{o} process $\left\{X_{t}\right\}$,
from the set of observations $\left\{Y_{t_j}\right\}_{j=1}^{N+1}$.
In the absence of market microstructure noise (i.e.,
when $Y_{t_j} = X_{t_j},\;j=1,\dots,N+1$) the
integrated volatility can be estimated from the realized integrated volatility of the process $\{Y_t\}$ \cite{tsay}.
In the presence of market microstructure noise
this is no longer true, see also \cite{Zhang+2005}, and a different estimation procedure is necessary.

The proposed estimator can be described roughly as follows. Let $\{J_k^{(X)}\}$ denote the Discrete Fourier Transform (DFT) of the differenced sampled $X_t$ process, and similarly for $Y_{t_j}$ and $\varepsilon_{t_j}$.  The integrated volatility can be written in terms of the inverse DFT of the variance of $J_k^{(X)}$.
We calculate the bias in the variance of $J_k^{(Y)}$, when using its sample estimator to estimate the variance of $J_k^{(X)}.$ The high frequency coefficients are heavily contaminated by the microstructure noise. With a formula for the bias it is possible to debias the estimated variance of the Fourier transform at every frequency, with the unknown parameters of the bias estimated using the Whittle likelihood \cite{Whittle53,Whittle62}. This produces a debiased estimator of the integrated volatility via an aggregation of the estimated variance, and we show also that the variance of the proposed estimator is reduced by the debiasing.

Our estimator shows highly competitive mean square error performance; it
also has several advantages over existing estimators. First, it is robust with respect to the signal to noise ratio; furthermore,  it is easy
to formulate and to implement; in addition, it readily generalizes to the case of
correlated observation errors (in time). Finally, the properties of our estimator are
transparent using frequency domain analysis.

The rest of the paper is organized as follows. In Section~\ref{sec:estmeth}
we introduce our estimator and present some of its properties, stated in Theorems \ref{thm:thmA} and \ref{thm:thmB}. We also discuss the time-domain understanding of the proposed method and the extension of the method to the case where the observation noise is correlated. In Section~\ref{sec:mcs}
we present the results of Monte Carlo simulations for our estimator.
Section~\ref{sec:concl} is reserved for conclusions. Various technical results
are included in the appendices.

%
%
%
\section{Estimation Methods}\label{sec:estmeth}
Let $\left\{Y_{t_j}\right\}$ be given by \eqref{e:model}, where the noise $
\left\{\varepsilon_{t_j}\right\}$ is independent of $\left\{X_{t_j}\right\}$,
is zero-mean and its variance at any time is equal to $\sigma^2_{\varepsilon}$.
The simplest estimator of the integrated
volatility of $X_t$ would ignore the high frequency component
of the data and use the realized integrated volatility of the observed process. The realized integrated volatility is given by
\begin{equation}
\label{e:biased}
\widehat{\langle X, X\rangle}_T^{(b)}=\left[Y,Y\right]_T\equiv \sum_{j=1}^{N}
\left(Y_{t_{j+1}}-Y_{t_{j}} \right)^2={\mathcal{O}}\left(\frac{1}{\Delta t} \right).
\end{equation}
This estimator is both inconsistent and biased, see \cite{Hansen06}.  For comparative
purposes, we define also the  realized integrated volatility of the sampled process $\{X_{t_j}\}$:
\begin{equation}
\label{e:unbias}
\widehat{\langle X, X\rangle}_T^{(u)}=\left[X,X\right]_T\equiv \sum_{j=1}^{N}
\left(X_{t_{j+1}}-X_{t_{j}} \right)^2={\mathcal{O}}\left(N\Delta t \right)={\mathcal{O}}\left(1\right).
\end{equation}
This cannot be used in practice as $X_{t_j}$ is not directly observed. Both these are estimators of the integrated volatility (quadratic variation) of $X$.

\subsection{Fourier Domain Properties}
We shall start by deriving an alternative representation of \eqref{e:biased} to motivate further development. Firstly we define the increment process of a sample from a generic time series $U_{t_j},\;j=1,\dots N+1$ by $\Delta U_{t_j}=U_{t_{j+1}}-U_{t_{j}},\;j=1,\dots N,$ and then the discrete Fourier Transform of $\Delta U_{t_j}$ by $J_{k}^{(U)}$ as
by \cite{Percival}[p.~206]
\begin{eqnarray}
J^{(U)}_k=\sqrt{\frac{1}{N}}\sum_{j=1}^{N} \Delta U_{t_j} e^{-2\pi i t_j f_k },
\quad f_k=\frac{k}{T},\quad U=X,\;Y,\;\varepsilon .
\end{eqnarray}
Our proposed estimator will be based on examining the second order properties of 
$\{J^{(Y)}_k\}$. $\left|J^{(Y)}_k\right|^2$ is the {\em periodogram} \cite{Brillinger} defined for a time series and is an inefficient estimator of $\var\{J^{(Y)}_k\}={\mathcal{S}}^{(X)}_{k,k}$.
Firstly we examine the properties of $\{J^{(X)}_k\}$. We have, with 
$\overline{\mu}_j=\frac{1}{\Delta t}  \int_{(j-1)\Delta t}^{j\Delta t}\mu_s\,ds$ denoting the local average of $\mu_t$,  
\begin{eqnarray}
\nonumber
\Delta X_{t_j}&=&\int_{(j-1)\Delta t}^{j\Delta t}[\mu_{s} ds+\sigma_{s} dW_{s}]=\overline{\mu}_j\Delta t+\int_{(j-1)\Delta t}^{j\Delta t}\sigma_{s} dW_{s},\\
\nonumber
J^{(X)}_{k}&=&\sqrt{\frac{1}{N}}\sum_{j=1}^{N}\Delta X_{t_j} e^{-2i\pi \frac{k j}{N}}=\sqrt{\frac{1}{N}}
\sum_{j=1}^{N}\left[\Delta t\overline\mu_j+\int_{(j-1)\Delta t}^{j\Delta t}\sigma_{s} dW_{s}\right]e^{-2i\pi \frac{k j}{N}}\\
&=&{\mathcal{O}}\left(\Delta t^{1/2}\frac{T}{k}\right)+\sqrt{\frac{1}{N}}
\sum_{j=1}^{N}\int_{(j-1)\Delta t}^{j\Delta t}\sigma_{s} dW_{s}e^{-2i\pi \frac{k j}{N}}.
\label{eq:jkX}\end{eqnarray}
We define
\[\tilde{J}^{(X)}_{k}= \sqrt{\frac{1}{N}}
\sum_{j=1}^{N}\int_{(j-1)\Delta t}^{j\Delta t}\sigma_{s} dW_{s}e^{-2i\pi \frac{k j}{N}},\]
and this to leading order approximates $J^{(X)}_{k}$ as $\Delta t\rightarrow 0$ for all but a few frequencies.
We can also note that, since $\mu_s$ is an It\^{o} process, it has almost
surely continuous paths,  which implies that
\begin{eqnarray}
\label{biasink1}
&&\Delta t^2 \sum_{k=0}^{N-1} \left|\sqrt{\frac{1}{N}}
\sum_{j=1}^{N}\overline\mu_j e^{-2i\pi \frac{k j}{N}}\right|^2\sim \sum_k \left|\frac{1}{\sqrt{N}} \Delta t N/k\right|^2={\mathcal{O}}\left( \Delta t\right)\\
&&\Delta t\sum_{k=0}^{N-1}\left|\sqrt{\frac{1}{N}}
\sum_{j=1}^{N}\overline\mu_j e^{-2i\pi \frac{k j}{N}}\right|\sim \sum_k \left|\frac{1}{\sqrt{N}} \Delta t N/k\right|={\mathcal{O}}\left( \log(\Delta t)\sqrt{\Delta t}\right),
\label{biasink2}
\end{eqnarray}
as $\Delta t \sum_j \overline\mu_j e^{-2i\pi \frac{k j}{N}}={\mathcal{O}}\left(\frac{T}{k}\right).$
So we only need, to leading order, calculate $\sum_{k=0}^{N-1} \left|\tilde{J}_k^{(X)}\right|^2={\mathcal{O}}(1)$
when calculating the properties of $\sum_{k=0}^{N-1} \left|{J}_k^{(X)}\right|^2$
from \eqref{biasink1} and \eqref{biasink2}. More formally we note that
\[\sum_{k=0}^{N-1} \left|{J}_k^{(X)}\right|^2=\sum_{k=0}^{N-1} \left|\tilde{J}_k^{(X)}\right|^2+
{\mathcal{O}}\left(\log(\Delta t)\sqrt{\Delta t} \right).\]
We need to determine the first and second order structure of $\{\tilde{J}^{(X)}_{k}\}_k.$
In general $\{\tilde{J}^{(X)}_{k}\}_k$ is a complex-valued random vector, which may not be a sample from a multivariate Gaussian distribution. 
%
The covariance matrix of a complex random vector $\mathbf{Z}$ is given by $\cov\left\{
\mathbf{Z},\mathbf{Z}\right\}=\E\left\{\mathbf{Z}\mathbf{Z}^H\right\}-
\E\left\{\mathbf{Z}\right\}\E\left\{\mathbf{Z}\right\}^H$ \cite{Neeser93,Picinbono}.
We have
\begin{eqnarray*}
\E\left\{\tilde{J}^{(X)}_{k}\right\}&=&0,\;k=1,\dots,N-1.
\end{eqnarray*}
Furthermore, with $\tilde{\mathcal S}^{(X)}_{k_1,k_2}=
\E\left\{\tilde{J}^{(X)}_{k_1}\tilde{J}^{(X)*}_{k_2}\right\}$,
\begin{eqnarray*}
\tilde{\mathcal S}^{(X)}_{k_1,k_2}&=&\frac{1}{N}
\E\left\{\sum_{n=1}^{N}\int_{(n-1)\Delta t}^{n\Delta t}
\sum_{l=1}^{N}\int_{(l-1)\Delta t}^{l\Delta t}\sigma_{s} dW_{s}
\sigma_{t} dW_{t}
e^{-2i\pi \left(\frac{k_1 n}{N}-\frac{k_2 l}{N}\right)}\right\}\\
&=&\frac{1}{N}
\sum_{n=1}^{N}\int_{(n-1)\Delta t}^{n\Delta t}
\sum_{l=1}^{N}\int_{(l-1)\Delta t}^{l\Delta t}
\E\left\{\sigma_{s} dW_{s}\sigma_{t} dW_{t}\right\}
e^{-2i\pi \left(\frac{k_1 n}{N}-\frac{k_2 l}{N}\right)}.
\\ \nonumber
&=&\frac{1}{N}\sum_{n=1}^{N}\int_{(n-1)\Delta t}^{n\Delta t}
\sum_{l=1}^{N}\int_{(l-1)\Delta t}^{l\Delta t}
\E\left\{\sigma_{s}\sigma_{t}\right\}\delta_{nl}\delta(t-s) 
e^{-2i\pi \left(\frac{k_1 n}{N}-\frac{k_2 l}{N}\right)}ds dt.\end{eqnarray*}
In particular we have that
\begin{eqnarray}
\nonumber
\tilde{\mathcal S}^{(X)}_{k,k}&=&
\frac{1}{N}\int_{0}^{T}\E\left\{\sigma_{s}^2\right\}ds+{\mathcal{O}}\left(\Delta t^2\right)
:=\overline{\sigma}^2_X+{\mathcal{O}}\left(\Delta t^2\right)\\
&=&\frac{\langle X, X\rangle_T}{N}+{\mathcal{O}}\left(\Delta t^2\right),
\label{eq:varJk}
\end{eqnarray}
where the error terms are due to the Riemann approximation to an integral
and thus it follows that
\begin{eqnarray}
\sum_{k=0}^{N-1}\tilde{\mathcal S}^{(X)}_{k,k}&=&
\int_{0}^{T}\E\left\{\sigma_{s}^2\right\}ds+{\mathcal{O}}(\Delta t).
\label{blah1}
\end{eqnarray}
$\overline{\sigma}^2_X$ does not depend on the value of $k$ but is constant {\em irrespectively} of the value of $k$. Malliavin and Mancino \cite{Malliavin+2002} in contrast under very light assumptions show how the Fourier coefficients of $\{\sigma^2_t\}$ can be calculated from the Fourier coefficients of $dX_t$, using a Parseval-Rayleigh relationship, see also \cite{Reno2008,Mancino+tech}. We can from \eqref{eq:varJk} make a stronger link from the Fourier transform to the integrated volatility than that of the Parseval-Rayleigh relationship, and shall use this `uniformity of energy' to estimate the microstructure bias.

We note that the covariance between different frequencies is given by:
\begin{eqnarray*}
\tilde{\mathcal S}^{(X)}_{k_1,k_2}&=&\frac{1}{N}\sum_{n=1}^{N}\int_{(n-1)\Delta t}^{n\Delta t}
\E\left\{\sigma_{t}^2\right\}\;dt
e^{-2i\pi n\left(\frac{k_1 }{N}-\frac{k_2 }{N}\right)}\\
&=&\frac{1}{N}\int_0^T\E\left\{\sigma_{t}^2\right\}\;dt
e^{-2i\pi t\left(\frac{k_1 }{T}-\frac{k_2 }{T}\right)}+{\mathcal{O}}\left(\Delta t^2\right).
\end{eqnarray*}
Let $\Sigma(f) :=\int_0^T \E \sigma_t^2   e^{-2i\pi f t} \;dt$.
We can bound the size of $\Sigma(\frac{k_1 }{T}-\frac{k_2 }{T})$ as $|k_1-k_2|$ increases. As $\E\left\{\sigma_{t}^2\right\}$ is smooth in $t$ the modulus of the covariance can be bounded for increasing $|k_1-k_2|$, as the Fourier transform $\Sigma(\frac{k_1 }{T}-\frac{k_2 }{T})$ decays proportionally to $|k_1-k_2|^{-\alpha-1}$ where $\alpha$ is the number of smooth derivatives of $\E\left\{\sigma_{t}^2\right\}$. We can also directly note that the variance of the discrete Fourier transform of the noise is precisely (this is {\em not} a large sample result)
\begin{eqnarray}
\label{blah2}
{\mathcal S}^{(\varepsilon)}_{k_1,k_2}&=&\sigma^2_{\varepsilon}\left|2\sin
\left(\pi f_{k_1}\Delta t \right)\right|^2\delta_{k_1,k_2},
\end{eqnarray}
by virtue of being the first difference of white noise (see also \cite{Bloomfield}).
The naive estimator can therefore be rewritten as, with ${\mathcal S}_{k_1,k_2}^{(Y)}=
\cov\{J_{k_1}^{(Y)},J_{k_2}^{(Y)}\}$:
\begin{subequations}
\label{e:naive}
\begin{eqnarray}
\label{e:naive1}
\widehat{\langle X, X\rangle}_T^{(b)}&=&\sum_{j=1}^{N}\Delta Y_{t_j}^2
=\sum_{k=0}^{N-1}\left|J^{(Y)}_k\right|^2,\\
\widehat{\mathcal{S}}^{(Y)}_{k,k}&=&\left|J^{(Y)}_k\right|^2, \label{e:naive2}\\
\label{e:naive3}
\E\left\{\widehat{\langle X, X\rangle}_T^{(u)}\right\}&=&
\sum_{k=0}^{N-1}\tilde{\mathcal S}_{k,k}^{(X)}+{\mathcal{O}}\left(\log(\Delta t)\sqrt{\Delta t}\right)+{\mathcal{O}}\left(\Delta t\right)\\
&\equiv &\sum_{k=0}^{N-1}{\mathcal S}_{k,k}^{(X)}.
\end{eqnarray}
\end{subequations}
The Parseval-Rayleigh
relationship in \eqref{e:naive1} is discussed in 
\cite{Mancino+tech}, and is used in \cite{Malliavin+2002}. 
We shall now develop a frequency domain specification of the bias of the naive estimator.
\begin{lemma}{(Frequency Domain Bias of the Naive Estimator)}\label{lemma1}
Let $X_t$ be an It\^o process and assume that the covariance of $J_{k_1}^{(X)}$
and $J_{k_2}^{(X)}$ to be ${\mathcal{S}}^{(X)}_{k_1,k_2}$ with the chosen sampling.
Then the naive estimator of the integrated volatility given by \eqref{e:naive} has an expectation given by:
\begin{eqnarray}
\label{eq:freq1}
\E\left\{\widehat{\langle X, X\rangle}_T^{(b)}\right\}
&=&\sum_{k=0}^{N-1}\left(\tilde{\mathcal{S}}^{(X)}_{k,k}+\sigma^2_{\varepsilon}
\left|2\sin(\pi f_k\Delta t)\right|^2\right)+
{\mathcal{O}}\left(\log(\Delta t)\sqrt{\Delta t}\right)\\
\nonumber
&=&\E\left\{\widehat{\langle X, X\rangle}_T^{(u)}\right\}+\sum_{k=0}^{N-1}\sigma^2_{\varepsilon}
\left|2\sin(\pi f_k\Delta t)\right|^2+{\mathcal{O}}\left(\log(\Delta t)\sqrt{\Delta t}\right)\\
\nonumber
&=&{\mathcal{O}}(1)+{\mathcal{O}}(\Delta t^{-1})+{\mathcal{O}}\left(\log(\Delta t)\sqrt{\Delta t}\right).
\end{eqnarray}
\end{lemma}

\begin{proof}
This result follows from the independence of $\{\varepsilon_t\}$ and $\{X_t\}$, combined with \eqref{blah1} and \eqref{blah2}.
\end{proof}

We notice directly from \eqref{eq:freq1} that the relative frequency contribution of $\Delta X_t$ and $\varepsilon_t$, i.e. ${\mathcal{S}}_{k,k}^{(X)}$ compared to the noise contribution $\sigma^2_{\varepsilon}
\left|2\sin(\pi f_k\Delta t)\right|^2$ determines the inherent bias of $\widehat{\langle X,X\rangle}_T^{(b)}$. Estimator~\eqref{e:naive} is inconsistent and biased
since it is equivalent to estimator~\eqref{e:biased}, and such a procedure would give an
unbiased estimator of the integrated volatility only when $\sigma^2_{\varepsilon}=0$.
When the estimator is expressed in the time domain the microstructure cannot be
disentangled from the It\^{o} process. On the other hand in the frequency domain, from the
very nature of a multiscale process, the contributions to $\widehat{\mathcal{S}}^{(Y)}_{k,k}$
can be disentangled.

\subsection{Multiscale Modelling}\label{sec:multiscalemodelling}
To correct the biased estimator we need to correct the usage of the biased estimator of ${\mathcal{S}}_{k,k}^{(X)},$
$\widehat{\mathcal{S}}_{k,k}^{(Y)},$ at each frequency. We therefore define a new
shrinkage estimator \cite[p.~155]{Wasserman} of ${\mathcal{S}}_{k,k}^{(X)}$ by
\begin{equation}
\label{eq:newestimator}
\widehat{{\mathcal{S}}}_{k,k}^{(X)}(L_k)=L_k \widehat{\mathcal{S}}^{(Y)}_{k,k}.
\end{equation}
$0\le L_k\le 1$ is referred to as the {\em multiscale ratio} and its optimal form for perfect bias correction is for an arbitrary It\^{o} process given by
\begin{equation}
\label{e:multiscaler}
L_k=\frac{{\mathcal{S}}^{(X)}_{k,k}}{{\mathcal{S}}^{(X)}_{k,k}+
\sigma^2_{\varepsilon}\left|2\sin(\pi f_k \Delta t)\right|^2}.
\end{equation}
This quantity {\em cannot} be calculated without explicit knowledge of ${\mathcal{S}}^{(X)}_{k,k}$
and $\sigma^2_{\varepsilon}$. We can however use \eqref{eq:varJk} to simplify \eqref{e:multiscaler} to obtain
\begin{equation}
\label{e:multiscaler2}
L_k=\frac{\overline{\sigma}^2_X}{\overline{\sigma}^2_X+
\sigma^2_{\varepsilon}\left|2\sin(\pi f_k \Delta t)\right|^2}.
\end{equation}
For a fixed $0\le L_k\le 1$
\begin{eqnarray*}
\E\left\{\widehat{\mathcal{S}}^{(X)}_{kk}(L_k)\right\}&=&L_k \E\left\{\left|J^{(Y)}_k\right|^2\right\}\\
&=&\overline{\sigma}_X^2+{\mathcal{O}}\left(\frac{\Delta t^{3/2}}{k}\right),
\end{eqnarray*}
where the order terms follow from the continuity of $\mu_s$.
We can define a new estimator for the true $L_k$ via:
\begin{eqnarray*}
\widehat{\langle X, X\rangle}_T^{(m_3)}&=&
\sum_{k=0}^{N-1}\widehat{\mathcal{S}}^{(X)}_{kk}(L_k),\end{eqnarray*}
where
\begin{eqnarray*}
\E\left\{\widehat{\langle X, X\rangle}_T^{(m_3)}\right\}&=&\langle X, X\rangle_T+
{\mathcal{O}}\left(\log(\Delta t)\sqrt{\Delta t}\right).
\end{eqnarray*}
Recall that $\langle X, X\rangle_T=O(T)=O(1)$.
Consequently, to leading order we can remove the bias from the naive estimator if we know the multiscale ratio.
We shall now develop a multiscale
understanding of the process under observation and use this to construct
an estimator for the multiscale ratio.

\subsection{Estimation of the Multiscale Ratio}
We have a two-parameter description on how the energy should be adjusted at each
frequency. We only need to determine estimators of $\bm{\sigma}=
\left(\overline{\sigma}^2_{X},\, \sigma^2_{\varepsilon}\right)$. We propose to implement the
estimation using the Whittle likelihood methods (see \cite{Beran1994} or \cite{Whittle53,Whittle62}). For a time-domain sample
$\bm{\Delta }{\mathbf{Y}}=\left(\Delta Y_{t_1},\dots,\Delta Y_{t_{N}}\right)$ that is {\em stationary}, if suitable conditions are satisfied, see for example \cite{Dzham1983}, then the Whittle likelihood approximates the time domain likelihood, with improving approximation as the sample size increases. It is possible to show a number of suitable properties of estimators based on the
Whittle likelihood, see \cite{Whittle53,Whittle62}. For processes that are not stationary, such conditions are in general not met, and so the function can be used as an objective function to construct estimators, but not as a true likelihood.
The Whittle log-likelihood is defined \cite{Whittle53,Whittle62} by
{\small\begin{alignat*}{1}
l ({\bf {\mathcal{S}}})&\equiv \log\left[\prod_{k=1}^{N/2-1}\frac{1}{{\mathcal S}^{(Y)}_{kk}}
e^{-\frac{\widehat{\mathcal S}^{(Y)}_{kk}}{{\mathcal S}^{(Y)}_{kk}}}\right]\\
&=-\sum_{k=1}^{N/2-1} \log\left({\mathcal{S}}^{(X)}_{kk}+\sigma^2_{\varepsilon}\left|2 \sin(\pi f_k\Delta t)\right|^2\right)-\sum_{k=1}^{N/2-1}\frac{\widehat{{\mathcal{S}}}^{(Y)}_{kk}}
{{\mathcal{S}}^{(X)}_{kk}+\sigma^2_{\varepsilon}\left|2 \sin(\pi f_k\Delta t)\right|^2}.
\end{alignat*}}
If  $\left\{\Delta X_t\right\}$ is not stationary, then as long as the total contributions of the covariance of
the incremental process can be bounded, using this likelihood will asymptotically (in $\Delta t^{-1}$) produce
suitable estimators, as we shall discuss further. 

\begin{definition}{(Multiscale Energy Likelihood)}\\
The multiscale energy log-likelihood is then defined using \eqref{eq:varJk} as:
\begin{alignat}{1}
l(\bm{\sigma})&=-\sum_{k=1}^{N/2-1} \log\left(\overline{\sigma}^2_{X}+ \sigma^2_{\varepsilon}\left|2 \sin(\pi f_k\Delta t)\right|^2\right)-\sum_{k=1}^{N/2-1}\frac{\widehat{\mathcal S}^{(Y)}_{kk}}
{\overline{\sigma}^2_{X}+\sigma^2_{\varepsilon}\left|2 \sin(\pi f_k\Delta t)\right|^2}.
\label{e:multis}
\end{alignat}
\end{definition}
We stress that strictly speaking this may not be a (log-)likelihood, but merely
a device for determining the multiscale ratio.
We maximise this function in $\bm{\sigma}$ to obtain a set of  estimators $\widehat{\bm{\sigma}}$. 

\begin{theorem}
\label{thm:thmA}
{(The Estimated Multiscale Ratio)}\\
The estimated multiscale ratio is given by
\begin{eqnarray}
\label{e:dis}
\widehat{L}_k&=&\frac{\widehat{\sigma}^2_{X}}{\widehat{\sigma}^2_{X}+
 \widehat{\sigma}^2_{\varepsilon}\left|2 \sin(\pi f_k\Delta t)\right|^2},\end{eqnarray}
where $\widehat{\sigma}^2_{X}$ and $\widehat{\sigma}^2_{\varepsilon}$ maximise $\ell(\bm{\sigma})$ given in
\eqref{e:multis}. $\widehat{L}_k$ satisfies
\begin{eqnarray}
\label{e:rss}
\frac{\widehat{L}_k}{L_k}&=&1+{\mathcal{O}}\left(\Delta t^{1/4}\right).
\end{eqnarray}
\end{theorem}
\begin{proof}
See Appendix \ref{sec:apA}.
\end{proof}

Combining \eqref{eq:newestimator} with \eqref{e:dis} the proposed estimator of the spectral density of $\left\{\Delta X_t\right\}$ is:
\begin{equation}
\label{e:grr}
\widehat{\mathcal S}^{(X)}_{kk}(\widehat{L}_k)=\widehat{L}_k \widehat{\mathcal S}^{(Y)}_{kk},
\end{equation}
where $\widehat{L}_k$ is given by \eqref{e:dis}.
\begin{theorem}
\label{thm:thmB}
{(The Multiscale Estimator of the Integrated Volatility)}\\
Assume that $\Delta X_{t_j}$ satisfies the conditions of Lemma \ref{lemma1} and Theorem \ref{thm:thmA}.
The multiscale estimator of the integrated volatility defined by
\begin{equation}
\widehat{\langle X, X\rangle}_T^{(m_1)}=
\sum_{k=0}^{N-1}\widehat{\mathcal S}^{(X)}_{kk}(\widehat{L}_k),
\end{equation}
where $\widehat{\mathcal S}^{(X)}_{kk}(\widehat{L}_k)$ is defined by \eqref{e:grr} has a mean and variance given by:
 \begin{eqnarray*}
\nonumber
\E\left\{\widehat{\langle X, X\rangle}_T^{(m_1)} \right\}&=&\sum_{k=0}^{N-1}{\mathcal{S}}^{(X)}_{kk}
+{\mathcal{O}}\left(\log(\Delta t)\sqrt{\Delta t}\right)+
{\mathcal{O}}\left(\Delta t^{1/4}\right)\\
\nonumber
&=&\int_0^T \E\left\{\sigma^2_t\right\} +{\mathcal{O}}\left(\log(\Delta t)\sqrt{\Delta t}\right)+
{\mathcal{O}}\left(\Delta t^{1/4}\right) \end{eqnarray*}
and
\begin{eqnarray*}
\var\left\{\widehat{\langle X, X\rangle}_T^{(m_1)} \right\}&=&
\sum_{k=0}^{N-1}{ L}_{k}^2 \left|{\mathcal S}^{(Y)}_{k,k}\right|^2+{\mathcal{O}}(\Delta t^{1/2})={\mathcal{O}}(\Delta t^{1/2}).
\nonumber
\end{eqnarray*}
\end{theorem}
\begin{proof}
See Appendix \ref{sec:apB}.
\end{proof}

We also note that
\begin{eqnarray}
\nonumber
\var\left\{\widehat{\langle X, X\rangle}_T^{(m_1)} \right\}&=&
\sum_{k=0}^{N-1}{ L}_{k}^2
\left|{\mathcal S}^{(Y)}_{k,k} \right|^2+{\mathcal{O}}(\Delta t^{1/2})\\
&<&{\mathcal{O}}\left(\frac{1}{\Delta t}\right)=
\var\left\{\widehat{\langle X, X\rangle}_T^{(b)} \right\},
\label{eq:lessthan}
\end{eqnarray}
unless $\sigma_{\varepsilon}=0$. We note  that the multiscale estimator has {\em lower} variance than the naive
method of moments estimator $\widehat{\langle X, X\rangle}_T^{(b)}$ due to the fact that $0\le { L}_{k}\le 1$. We have thus removed bias {\em and} simultaneously decreased the variance, the latter effect usually being the main purpose of shrinkage estimators. Note that if we knew the true multiscale ratio $L_k$
and used it rather than $\widehat{L}_k$ (i.e. used 
$\widehat{\langle X, X\rangle}_T^{(m_3)}$) then we would expect an estimator from this quantity to recover the same variance as the estimator based on the noise-free observations. This loss of efficiency is inevitable, as we have to estimate $L_k$. Finally we can also construct a Whittle estimator for the integrated volatility by starting from \eqref{eq:varJk}
and taking
\begin{equation}
\label{whittle}
\widehat{\langle X, X\rangle}_T^{(w)}=N \widehat{\sigma}^2_X .
\end{equation}
The sampling properties of $\widehat{\langle X, X\rangle}_T^{(w)}$ are found in Appendix \ref{sec:apA}, and $\widehat{\sigma}_X^2$ is asymptotically unbiased. The results in Appendix \ref{sec:apA} imply that
\begin{eqnarray}
\var\left\{\widehat{\langle X, X\rangle}_T^{(w)}\right\}&=&T \frac{\sigma_{\varepsilon}}{\tau_X^{1/2}}16  \tau_X^{2}\sqrt{\Delta t}.
\label{varwhittle2}
\end{eqnarray}
We see that the variance depends on the length of the time course, the inverse of the signal to noise ratio, the square root of the sampling period and the fourth power of the ``average standard deviation'' of the $X_t$ process., We may compare the variance of \eqref{eq:lessthan} with the variance of \eqref{varwhittle2}, to determine which estimator of $\widehat{\langle X, X\rangle}_T^{(w)}$ and $\widehat{\langle X, X\rangle}_T^{(m_1)}$ is preferable. 
We shall return to this question of relative performance in the examples section, but intuitively argue that $\widehat{\langle X, X\rangle}_T^{(w)}$ and
 $\widehat{\langle X, X\rangle}_T^{(m_1)}$ are more or less the same estimator, with the latter estimator being more intuitive to explain.

\subsection{Time Domain Understanding of the Method}
We may write the frequency domain estimator of the spectral density of $\Delta X_t$
in the time domain to clarify some of its properties. We define 
\[\widehat{s}^{(X)}_{\tau}=\frac{1}{N}\sum_{k=0}^{N-1} \widehat{\mathcal{S}}^{(X)}_{kk} e^{2i\pi \frac{k\tau}{N}},\quad \tau \in {\mathbb{N}},\]
which when $\Delta X_t$ is a stationary process corresponds to the estimated autocovariance sequence of $\Delta X_t$ using the method of moments estimator \cite[Ch. 5]{Brillinger}. We then have 
\begin{eqnarray}
\widehat{\mathcal{S}}^{(X)}_{kk}&=&L_k \widehat{\mathcal{S}}^{(Y)}_{kk}, \quad \widehat{s}^{(X)}_{\tau}=\sum_u \ell_{\tau-u} \widehat{s}^{(Y)}_{u},
\label{timdomain33}
\end{eqnarray}
and so the estimated autocovariance of the $\Delta X_t$ process, namely $\widehat{s}^{(X)}_{\tau}$, is a smoothed version of $\widehat{s}^{(Y)}_{\tau}$.
We can therefore view $\widehat{\mathcal{S}}^{(X)}_{kk}$ as the Fourier transform of a smoothed version of the autocovariance sequence of $\Delta Y_t$. We let
\begin{equation}
L(f)=\frac{\overline{\sigma}^2_X}{\overline{\sigma}^2_X+
\sigma^2_{\varepsilon}\left|2\sin(\pi f \Delta t)\right|^2},
\end{equation}
be the continuous analogue of $L_k$.
To find the smoothing kernel we are using we need to calculate
\begin{eqnarray}
\nonumber
\ell_{\tau}&=&
\frac{1}{N}\sum_{k=0}^{N-1}L_k e^{2i\pi \frac{k\tau}{N}}\\
&=&\int_{-\frac{1}{2}}^{\frac{1}{2}}
\frac{\overline{\sigma}^2_X}{\overline{\sigma}^2_X+4\sigma^2_{\varepsilon}\sin^2(\pi f )}\;e^{2i\pi f \tau }\;df+{\mathcal{O}}\left( \Delta t\right).
\label{eqnsmooth}
\end{eqnarray}
Thus utilizing integration in the complex plane (see Appendix \ref{sec:time}) we obtain that
\begin{eqnarray}
\ell_{\tau}=\left\{\begin{array}{lcr}
\left(\frac{ \sigma^2_\varepsilon}{\overline{\sigma}^2_{X}}\right)^{\tau}+{\mathcal{O}}
\left(\left(\frac{ \sigma_\varepsilon}{\overline{\sigma}_{X}}\right)^{2\tau+2}\right) & {\mathrm{if}}
& \sigma^2_{\varepsilon}<\overline{\sigma}^2_X\\
\frac{\overline{\sigma}_X}{2\sigma_{\varepsilon}}\left(1- \frac{\overline{\sigma}_X}{\sigma_{\varepsilon}}\right)^{\tau}+
{\mathcal{O}}\left(\frac{\overline{\sigma}_X^2}{2\sigma_{\varepsilon}^2}
\left(1-\frac{\overline{\sigma}_X}{ \sigma_{\varepsilon}}\right)^{\tau}\right) & {\mathrm{if}}
& \sigma^2_{\varepsilon}>\overline{\sigma}_X^2
\end{array}\right.
\label{timedomain22}
\end{eqnarray}
These are both decreasing sequences in $\tau$. We write $r_{\tau}=\frac{\overline{\sigma}_X}{2\sigma_{\varepsilon}}\left(1- \frac{\overline{\sigma}_X}{\sigma_{\varepsilon}}\right)^{\tau}.$ If we can additionally assume that $L(f)$ decreases sufficiently rapidly to be near zero by $f=\frac{1}{\pi}$ then we find that
\begin{eqnarray*}
\ell_{\tau}\approx q_{\tau}=\frac{\overline{\sigma}_X}{2\sigma_{\varepsilon}}e^{-\frac{\overline{\sigma}_X}{\sigma_{\varepsilon}}|\tau|}.
\end{eqnarray*}
In the limit of no observation noise ($\frac{\overline{\sigma}_X}{\sigma_{\varepsilon}}\rightarrow \infty$) then this sequence becomes a delta function centered at $\tau=0$. Let us plot these functions, i.e. $\ell_{\tau}$, $r_{\tau}$ and $q_{\tau}$ for a chosen case of $\overline{\sigma}_X^2/\sigma^2_{\varepsilon}\approx0.0331 $ (the approximate SNR used in a later example), in Figure \ref{fig:fig-1} (left). We see that theory coincides very well with practise, and almost perfect agreement between the three functions. $\ell_{\tau}$ is however a strange choice of kernel, if dictated by the statistical inference problem: it has heavier tails than the common choice of the Gaussian kernel, and is extremely peaked around zero (a Gaussian kernel with the same variance has been overlaid in Figure \ref{fig:fig-1}). This is not strange, as we are trying to filter out correlations due to non-It\^{o} behaviour, but counter to our intuition about suitable kernel functions, as the differenced It\^{o} process exhibits very little covariance at any lag but zero, the sharp peak at zero is necessary.

\begin{figure}
\centerline{
\includegraphics[width=2.5in, height = 2.5in]{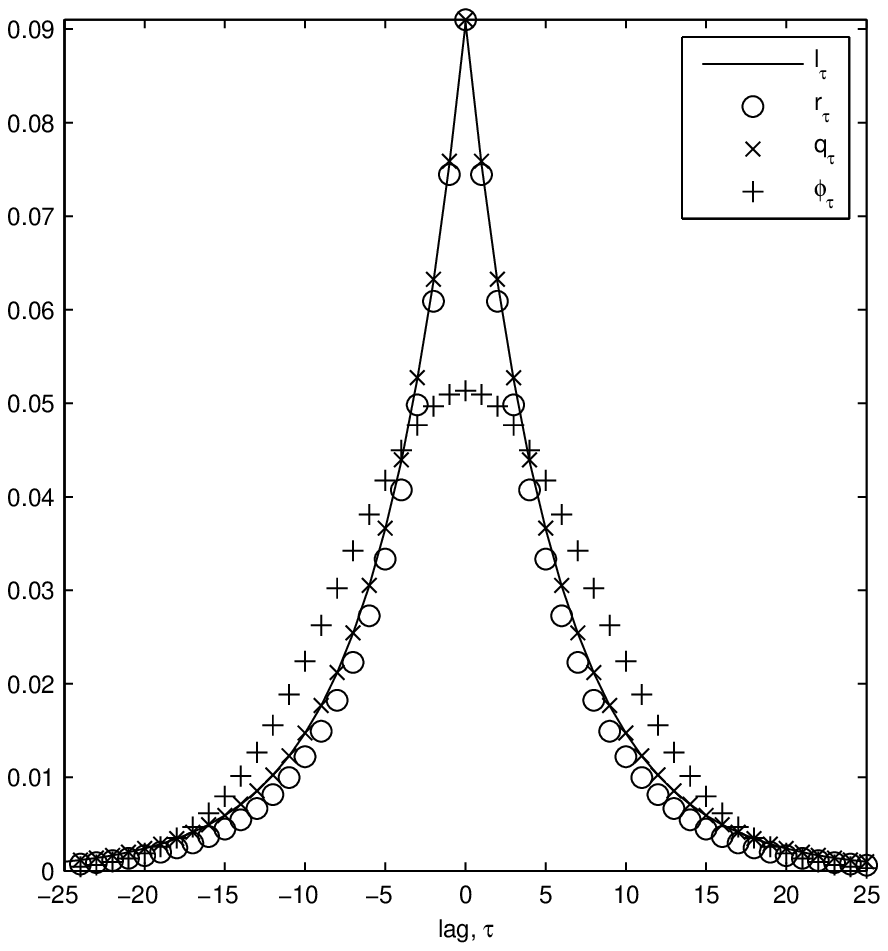} 
\includegraphics[width=2.5in, height = 2.5in]{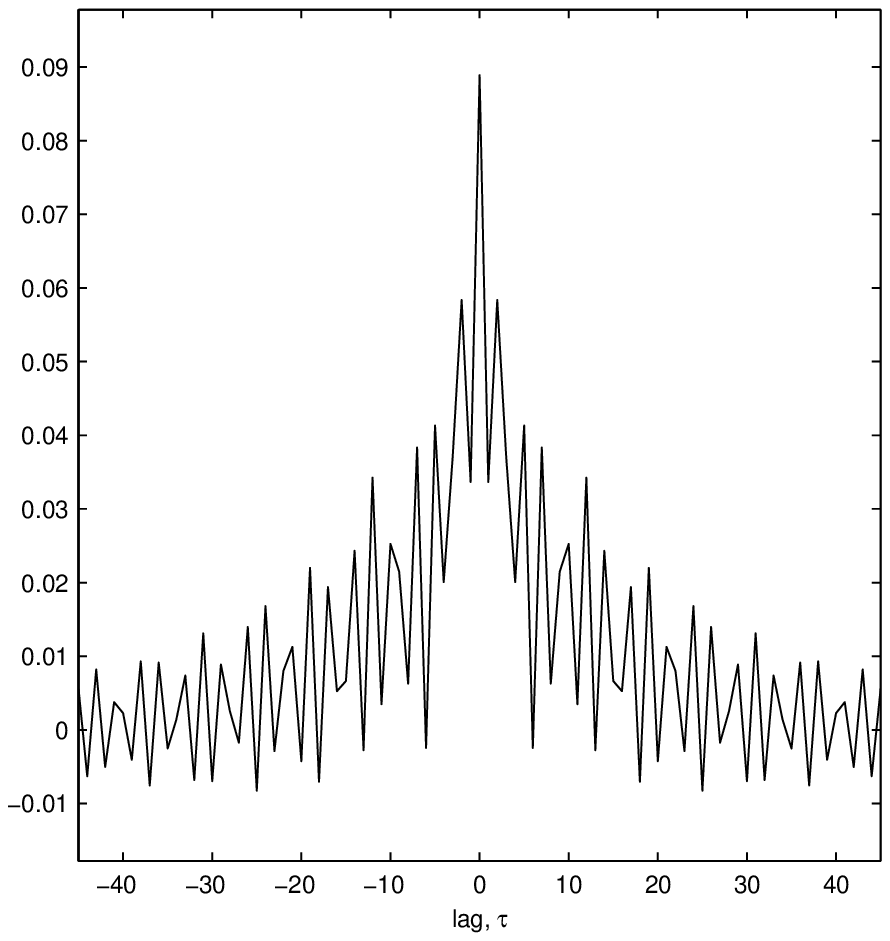}\\}
\begin{center}
\caption{$\ell_\tau$ as well as $r_{\tau}$ and $q_{\tau}$ for a chosen value of the SNR (left). The approximate weighting functions perfectly mirror the exact calculation. We overlay a Gaussian kernel with the same spread for comparison. $\ell_\tau$ estimated for the MA(6) case (right).\label{fig:fig-1}}
  \end{center}
\end{figure}

\subsection{Correlated Errors}\label{subsec:corr}
In many applications we need to consider correlated observation noise. We assume that despite being dependent the $\varepsilon_{t_j}$ is a stationary time series. Stationary processes
can be conveniently represented in terms of aggregations of uncorrelated white noise processes, using the Wold decomposition theorem \cite{Brockwell}[p.~187]. We may therefore write the zero-mean observation $\varepsilon_{t_j}$ as
\begin{equation}
\label{wolddecomp}
\varepsilon_{t_j}=\sum_{k=0}^{\infty} \theta_{t_k} \eta_{t_j-t_k},\end{equation}
where $\theta_{t_0}\equiv 1,$ $\sum_{j} \theta_{t_j}^2<\infty,$ and 
$\left\{\eta_{t_n}\right\}$ satisfies $\E\left\{\eta_{t_n} \right\}=
0$ and $\E\left\{\eta_{t_n} \eta_{t_m}\right\}=
\sigma^2_{\eta} \delta_{n,m}$, a model also used in \cite{Sykulskietal2008}.
Common practise would involve approximating the variable by a finite number of elements in the sum, and thus we truncate \eqref{wolddecomp} to some $q\in{\mathbb{Z}}$.
We therefore model the noise as a Moving Average (MA) process specified by
\begin{equation}
\varepsilon_{t_j}=\eta_{t_j}+\sum_{k=1}^{q}\theta_{t_k}\eta_{t_{j-k}},\end{equation}
and the covariance  of the DFT of the differenced $\varepsilon_{t_j}$ process takes the form:
\begin{equation}\label{eq:MA}
{\mathcal S}^{(\varepsilon)}_{k,k}=\sigma^2_{\eta}\left|1+\sum_{k=1}^{q} \theta_k e^{2i\pi f k}\right|^2|2\sin\left(\pi f \Delta t \right)|^2.
\end{equation}
This leads to defining a new multiscale ratio replacing $\sigma^2_{\varepsilon}|2\sin\left(\pi f \Delta t \right)|^2$ of \eqref{e:multiscaler2} with $\sigma^2_{\eta}\left|1+\sum_{k=1}^{q} \theta_k e^{2i\pi f k}\right|^2|2\sin\left(\pi f \Delta t \right)|^2$. We then obtain a new estimator of ${\mathcal{S}}^{(X)}_{kk}$. In general the value of $q$ is not known. To simultaneously implement model choice, we need to penalize the likelihood. We  define the corrected Aikake information criterion (AICC) by \cite[p. 303]{Brockwell} (refer to \eqref{e:multis} for $l\left(\bm{\sigma},\bm{\theta}\right)$
with $\sigma^2_{\varepsilon}|2\sin\left(\pi f \Delta t \right)|^2$ replaced by $\sigma^2_{\eta}\left|1+\sum_{k=1}^{q} \theta_k e^{2i\pi f k}\right|^2|2\sin\left(\pi f \Delta t \right)|^2$)
\begin{eqnarray}\label{e:aicc}
{\mathrm{AICC}}(\bm{\theta})&=&-2l\left(\bm{\sigma},\bm{\theta}\right)
+2\frac{(p+2)n}{n-p-3}.
\end{eqnarray}
By minimizing this function, in $\bm{\sigma}$, $\bm{\theta}$ and $q$, we obtain the best fitting model for the noise accounting for overfitting by using the penalty term. With this method we retrieve a new multiplier that is applied in the Fourier domain, which corresponds to a new smoother in the Fourier domain, where the smoothing window (and its smoothing width) have been automatically chosen by the data. See an example of such a smoothing window $\ell_{\tau}$ in  in Figure \ref{fig:fig-1} (right). Here $L_k$ has been estimated from an It\^{o} process immersed in an MA noise process. The spectrum of the MA has a trough at frequency 0.42. We therefore expect to reinforce oscillations at period $1/0.42\approx 2.5$, which is evident from the oscillations of the estimated kernel. For more details of this process see section \ref{excorr1}.

\section{Monte Carlo Studies}\label{sec:mcs}
In this section we demonstrate the performance of the multiscale estimator
through Monte Carlo simulations. We first describe the de-biasing procedure
of the estimator for the Heston Model using Fourier domain graphs. We then present bias, variance
and mean square error results of various estimators (including the multiscale
estimator, the naive estimator and the first-best estimator developed in \cite{Zhang+2005}), for the Heston Model as well as Brownian and Ornstein
Uhlenbeck processes. We then consider the case where the sample path in a
Heston Model is much
shorter and another case where the microstructure noise is greatly reduced.
Finally, we consider the case of correlated errors and show how a stationary noise process can be captured using model choice
methods and then the integrated volatility can be estimated using the adjusted
multiscale estimator.

\subsection{The Heston Model}
The Heston model is specified in \cite{Heston}:
\begin{equation}
\label{e:heston}
dX_t=\left(\mu-\nu_t/2\right)dt+\sigma_t dB_t,\quad
d\nu_t=\kappa\left(\alpha-\nu_t\right)dt+\gamma\nu_t^{1/2}dW_t,
\end{equation}
where $\nu_t = \sigma_t^2$, and $B_t$ and $W_t$ are
correlated 1-D Brownian motions. We will use the same
parameter values to the ones that were used in~\cite{Zhang+2005}, namely  $\mu=.05, \,
\kappa=5, \, \alpha=.04, \, \gamma =.5$ and the correlation coefficient between the two
Brownian motions B and W is $\rho = -.5$. We set $X_0 = 0$ and $\nu_0 = 0.04$, which is
the long time limit of the expectation of the process $\nu_t$.\footnote{$\lim_{t
\rightarrow +\infty}\mathbb{E} \nu_t = \alpha.$}

We calculate $\widehat{\mathcal{S}}^{(X)}_{kk}$
and $\widehat{\mathcal{S}}^{(\varepsilon)}_{kk}$ directly from simulated data and average across realizations, producing Figure \ref{fig:fig0}, where $k$ is indicated by its frequency $f_k=k/N$, and only plotted for $k=0,\dots,N/2-1$, as the spectrum (or ${\mathcal{S}}_{kk}^{(X)}$) is symmetric. We see directly from these plots that (on average as we showed) $\widehat{\mathcal{S}}^{(X)}_{kk}$ is constant whilst $\widehat{\mathcal{S}}^{(\varepsilon)}_{kk}$ is strongly increasing with $k$, completely dwarfing the other spectrum at large $k$. \eqref{blah1} implies that an equal weighting is given to all frequencies
for the differenced It\^o process. The noise process will in contrast have a spectrum
that is far from flat, and a suitable bias correction would shrink the estimator of ${\mathcal{S}}^{(X)}_{kk}$ at higher frequencies.
\begin{figure}
\centerline{
\begin{tabular}{c@{\hspace{3pc}}c}
\includegraphics[width=2.5in, height = 2.5in]{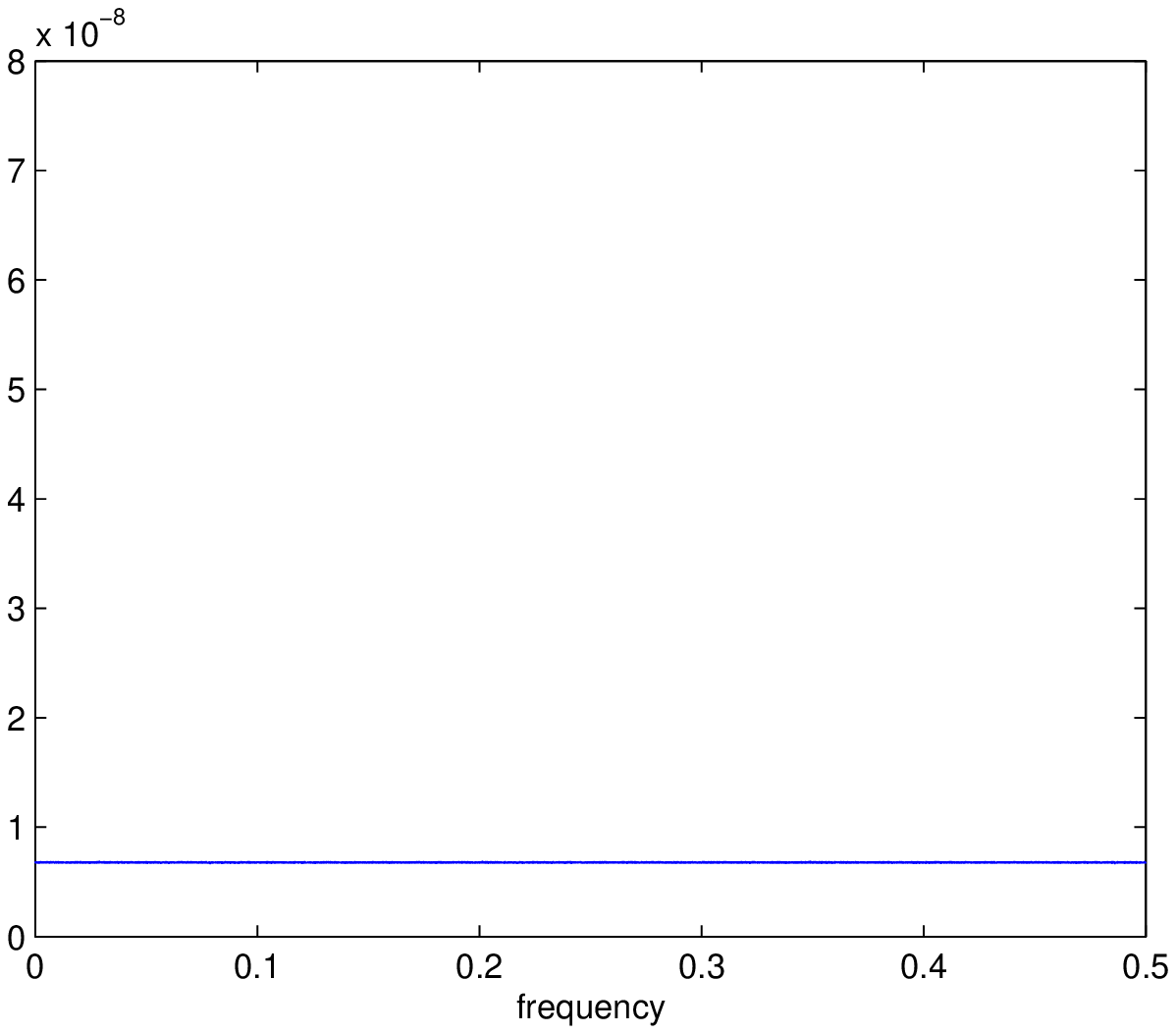} &
\includegraphics[width=2.5in, height = 2.5in]{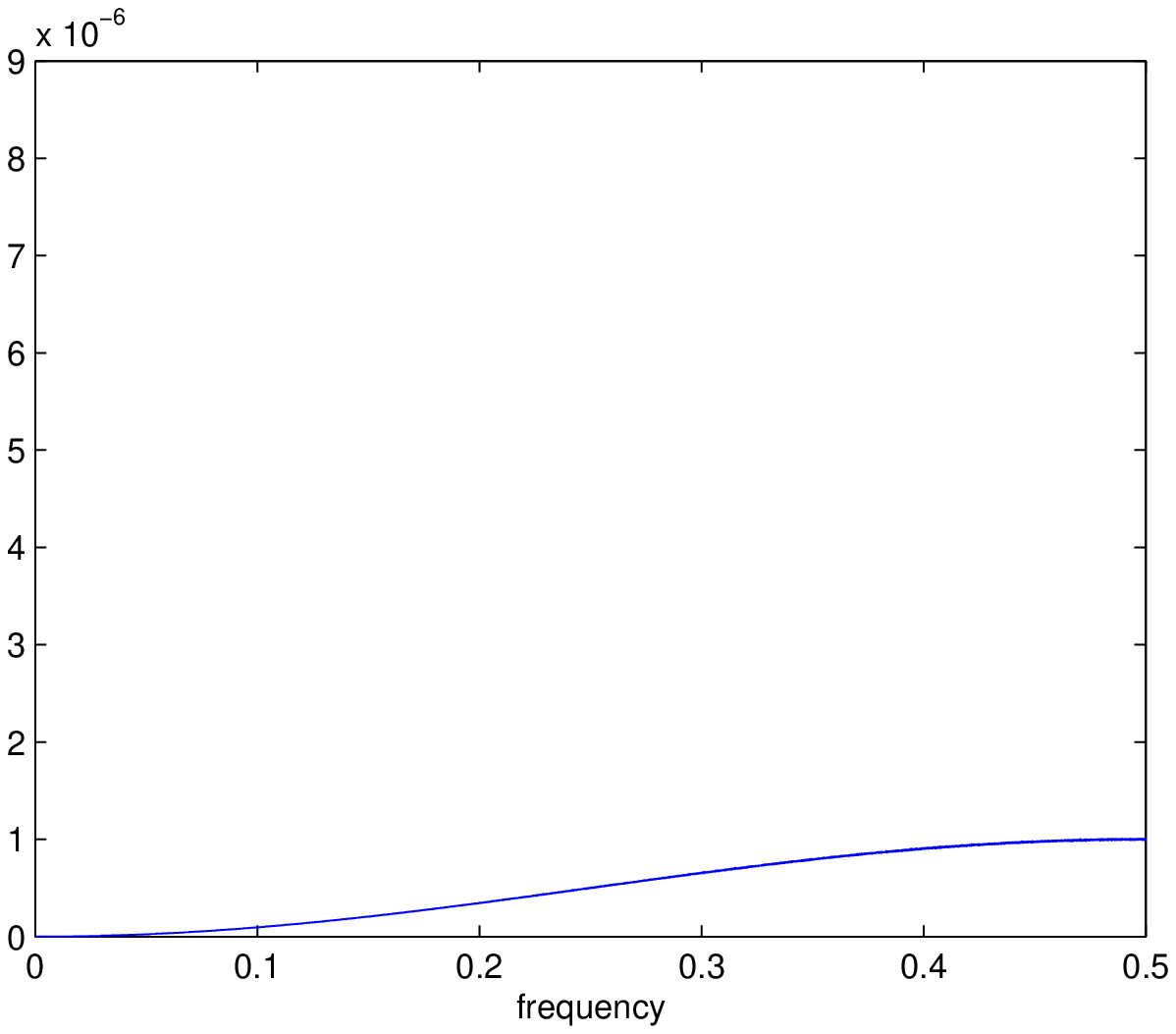} \\
\end{tabular}}
\begin{center}
\caption{$\widehat{\mathcal{S}}^{(X)}_{kk}$ (left)
and $\widehat{\mathcal{S}}^{(\varepsilon)}_{kk}$ (right) averaged over 100,000 realizations. Note the different scaling of the $y$
axis in the two figures. \label{fig:fig0}}
  \end{center}
\end{figure}

We also calculate $\widehat{\mathcal S}^{(X)}_{kk}$ and $\widehat{\mathcal S}^{(\varepsilon)}_{kk}$ for one simulated path, displayed in Figure \ref{fig:fig1}. Here
we have used the same sample length $T$ and noise intensity $\sigma^2_{\varepsilon}$ as in~\cite{Zhang+2005}: $T = 1$ day and $\sigma^2_{\varepsilon} = 0.0005^2$. The length of the sample path, $T = 1$ day or $23,400s$ with $\Delta t = 1 s$, corresponds to one trading day, since we take one trading day to be $6.5h$ long. Notice the different shape of the two periodograms. $\widehat{\mathcal S}^{(Y)}_{kk}$ will not be distinguishable
from $\widehat{\mathcal S}^{(\varepsilon)}_{kk}$ at higher frequencies, despite the moderate to low intensity of the market microstructure noise. If we observed the two components $X_t$ and $\varepsilon_t$ separately, then the multiscale ratio $L_k$ could be estimated from $\widehat{\mathcal S}^{(X)}_{kk}$
and $\widehat{\mathcal S}^{(\varepsilon)}_{kk}$ using the method of moments formula. In this case, we would estimate $L_k$ by the sample Fourier Transform variances
\begin{equation}
\label{e:multiscale_L_est}
\widetilde{L}_k=\frac{\widehat{\mathcal S}^{(X)}_{kk}}{\widehat{\mathcal S}^{(X)}_{kk}+\widehat{\mathcal S}^{(\varepsilon)}_{kk}}.
\end{equation}
The corresponding estimator of the integrated volatility becomes
\begin{equation}
\label{e:m2}
\widehat{\langle X, X\rangle}_T^{(m_2)}=\sum_{k=0}^{N-1}\widetilde{L}_k \widehat{\mathcal S}^{(Y)}_{kk}.
\end{equation}
The estimated multiscale ratio $\widetilde{L}_k$, for the Heston model with the specified parameters, is plotted in Figure \ref{fig:fig2}.

\begin{figure}
\centerline{
\begin{tabular}{c@{\hspace{1pc}}c}
\includegraphics[width=2.5in, height = 2.5in]{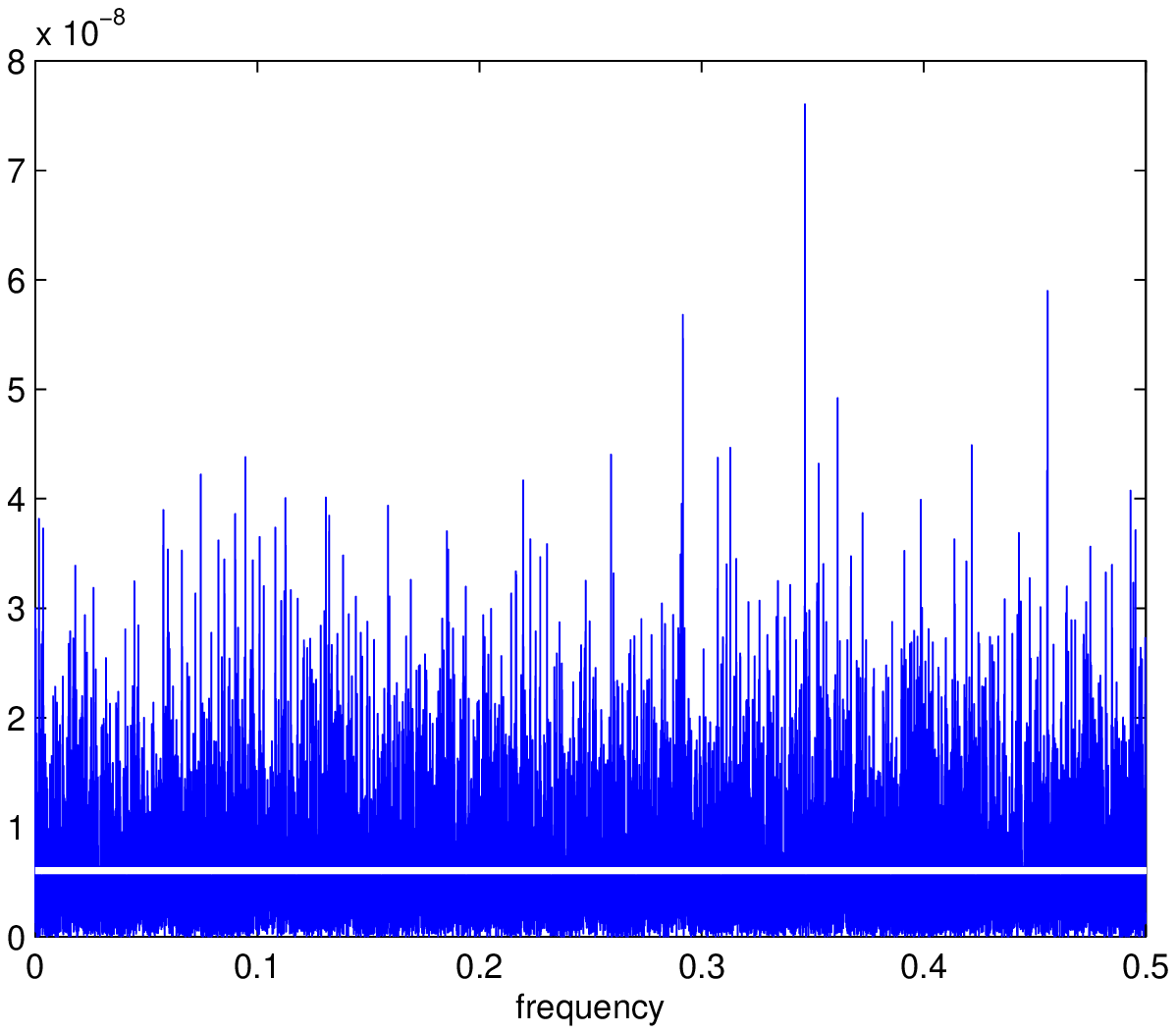} &
\includegraphics[width=2.5in, height = 2.5in]{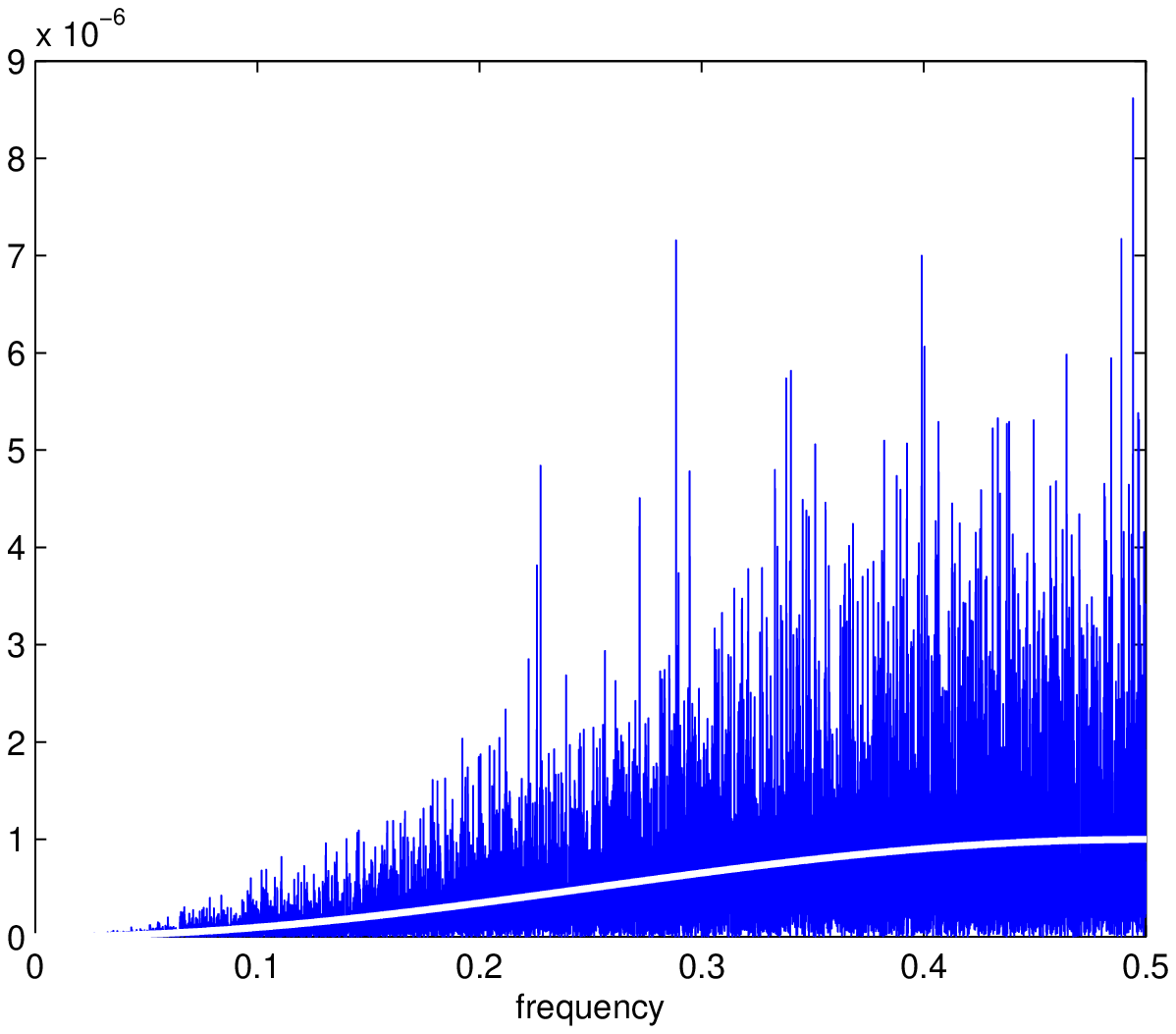} \\
\includegraphics[width=2.5in, height = 2.5in]{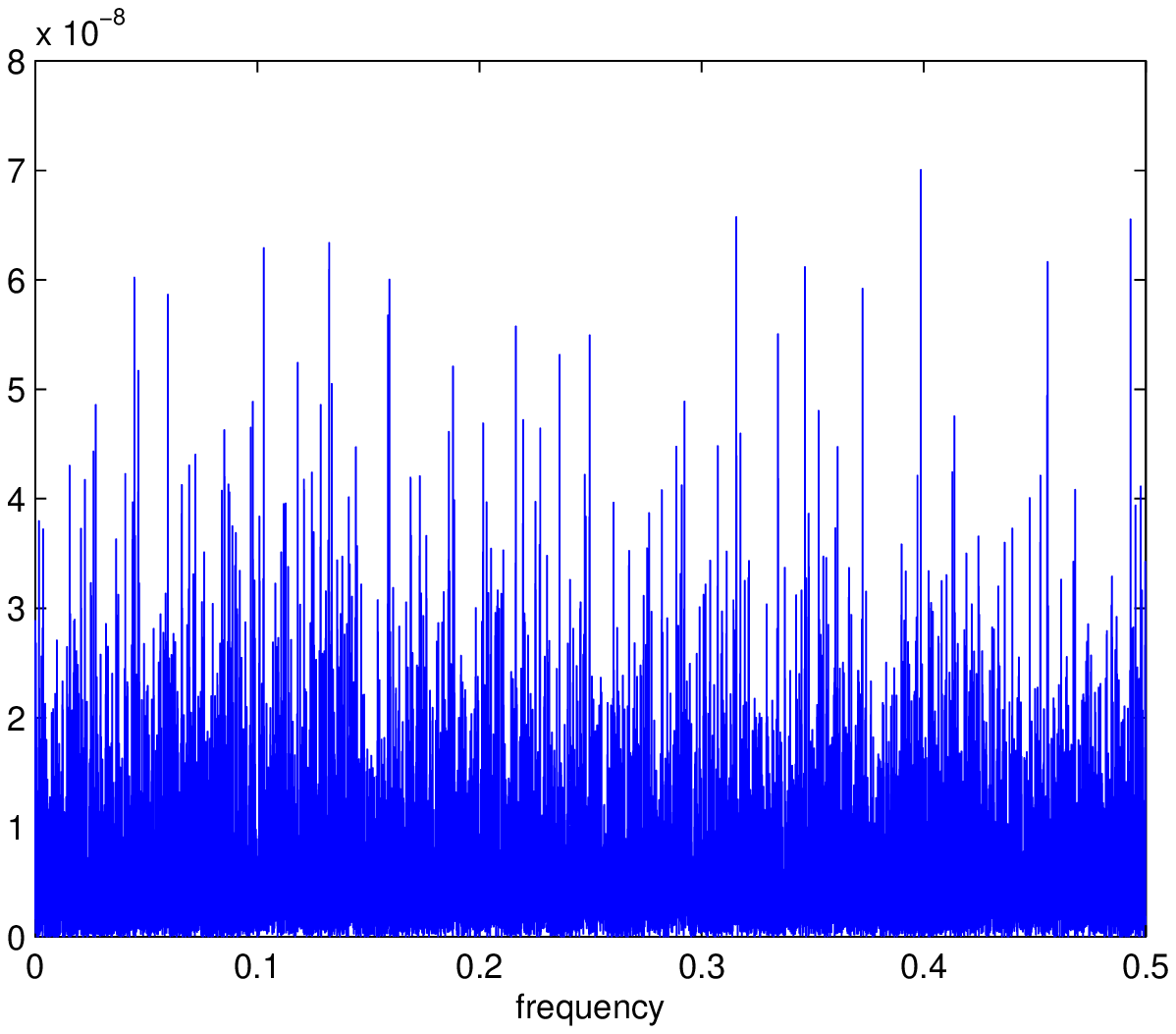} &
\includegraphics[width=2.5in, height = 2.5in]{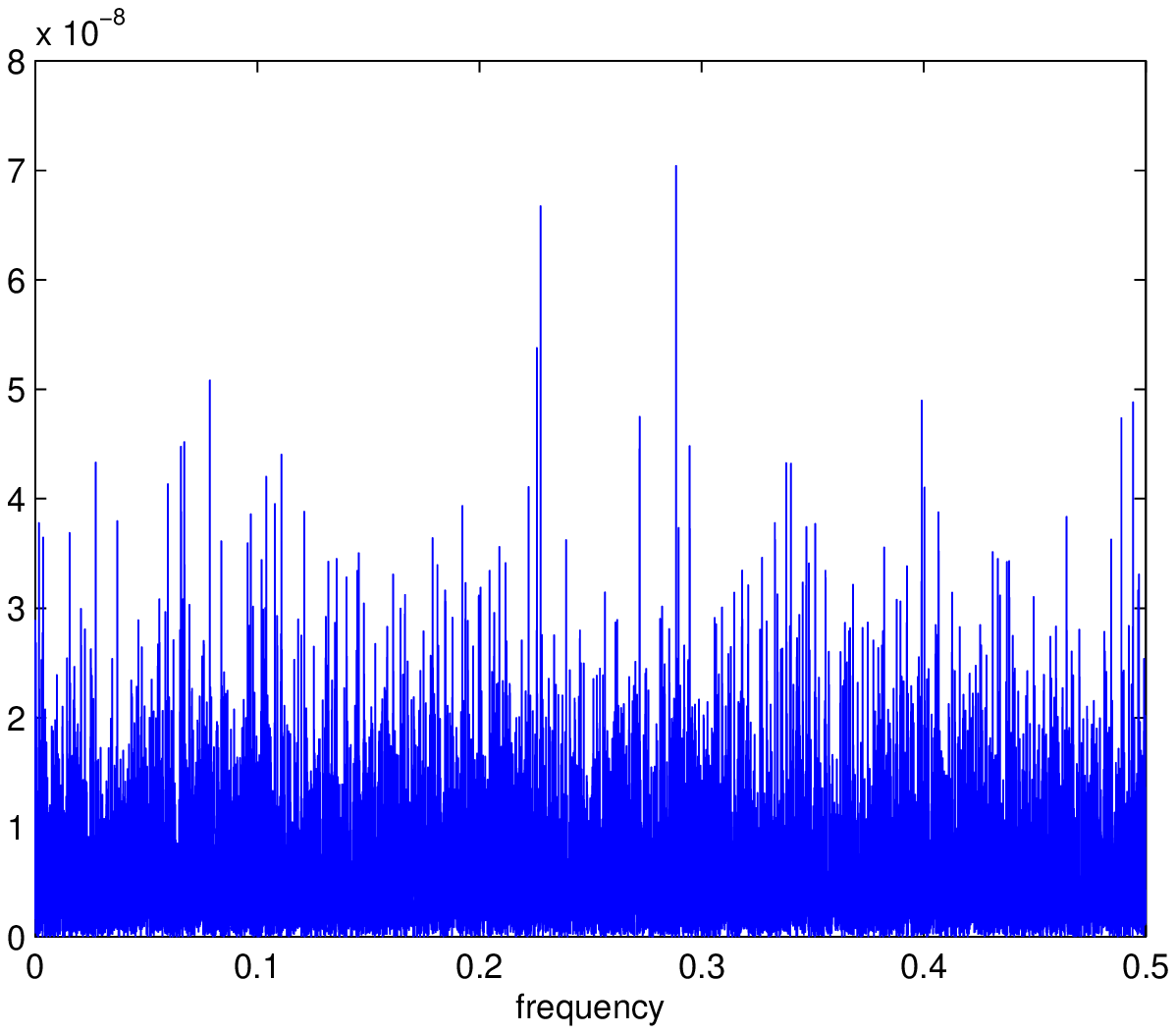} \\
\end{tabular}}
\begin{center}
\caption{A realisation of $\widehat{\mathcal S}^{(X)}_{kk}$ (top left),
 a realisation of $\widehat{\mathcal S}^{(\varepsilon)}_{kk}$ (top right) with the Whittle estimates superimposed
 and of two biased corrected estimators
 of ${\mathcal S}^{(X)}_{kk}$, using $\widetilde{L}_k\widehat{\mathcal S}^{(X)}_{kk}$ (bottom left) and
 $\widehat{L}_k\widehat{\mathcal S}^{(Y)}_{kk}$ (bottom right). Notice the different scales in the
four figures\label{fig:fig1}. Estimated spectra are plotted on a linear scale for ease of comparison to the effect of applying $\widehat{L}_k$.}
\end{center}
\end{figure}
\begin{figure}
\centerline{
\includegraphics[height=7.6cm]{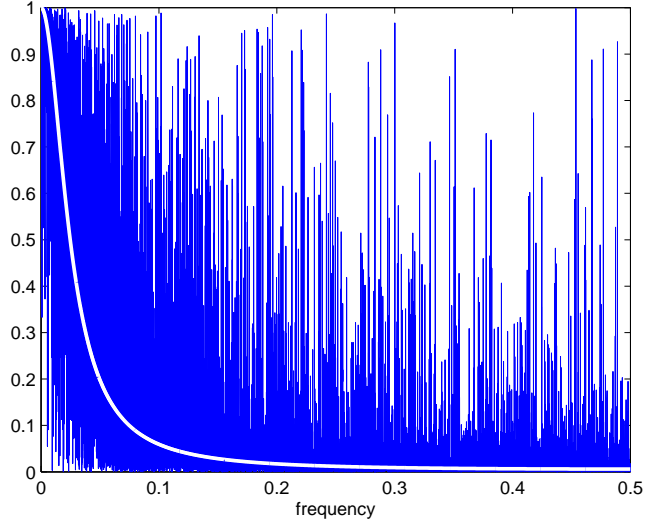}}
\begin{center}
\caption{The method of moments estimate $\widetilde{L}_k$ from a single realisation, with the Whittle estimate (white line)  of $L_k$ superimposed.\label{fig:fig2}}
\end{center}
\end{figure}
The multiscale ratio cannot be estimated using the method of moments in realistic scenarios, as we only observe the
aggregated process $Y_t$ and not the two processes $X_t$ and $\varepsilon_t$
separately. Figure
\ref{fig:fig1} displays the estimated multiscale ratio $\widetilde{L}_k$ applied to
$\widehat{\mathcal S}^{(Y)}_{kk}$ over one path realisation.
This plot suggests that the energy over the high frequencies has been shrunk and that $\widetilde{L}_k\widehat{\mathcal S}^{(Y)}_{kk}$ is a good approximation to $\widehat{\mathcal S}^{(X)}_{kk}$. It therefore seems not unreasonable that the summation of this function across frequencies
should make a good approximation to the
integrated volatility.

The parameters ($\widehat\sigma^2_X$
and $\widehat\sigma^2_\varepsilon$) are
found separately for each path using the MATLAB function \textsf{fmincon}
on \eqref{e:multis}. Figure \ref{fig:fig1} shows $\widehat\sigma^2_X$
and $\widehat\sigma^2_\varepsilon \left|2\sin(\pi f_k \Delta t)\right|^2$ (in white) plotted over the periodograms $\widehat{\mathcal S}^{(X)}_{kk}$
and $\widehat{\mathcal S}^{(\varepsilon)}_{kk}$ for one simulated path.
The approximated values of $\overline{\sigma}^2_X$
and $\sigma^2_\varepsilon$ are quite similar to the averaged periodograms
of Figure \ref{fig:fig0}; in fact the accuracy of the new estimator depends on how consistently these parameters are estimated in the presence of limited information
from the sampled process $Y_t$. Figure \ref{fig:fig2} shows
the corresponding estimated multiscale ratio $\widehat{L}_k$ (in white) from this simulated path, as defined in \eqref{e:dis}. The function decays, as expected, so that it will remove the high-frequency
microstructure noise in the spectrum of $Y_t$; the ratio is also a good approximation of $\widetilde{L}_k$. Figure \ref{fig:fig1} shows $\widehat{L}_k\widehat{\mathcal S}^{(Y)}_{kk}$, which
is again similar to $\widehat{\mathcal S}^{(X)}_{kk}$. It would appear that the new estimator has successfully removed the microstructure effect from each frequency.

It is worth noting that the ratios $L_k$ and $\widehat{L}_k$
quantify the effect of the multiscale structure of the process. If $\sigma_\varepsilon^2$ is zero (ie. there is no microstructure noise), then no correction will be made to the spectral density function (the ratio will equal 1 at all frequencies).
So in the case of zero microstructure noise, the estimate would recover $\widehat{\mathcal S}^{(X)}_{kk}$ and from \eqref{e:naive} the estimate of the integrated
volatility would simply be the realized integrated volatility of the observable process.

We investigate the performance of the multiscale estimator using Monte Carlo simulations. In this study 50,000 simulated paths are generated. Table \ref{tbl:1} displays the results of our simulation, where biases, variances and errors are calculated using a Riemann sum approximation of the integral
\begin{equation}
\frac{T}{N} \sum_{i=1}^N \sigma^2_i=\int_0^T \sigma^2_t\,dt.
\end{equation}
The two estimators
$\widehat{\langle X, X\rangle}_T^{(u)}$ and $\widehat{\langle X, X\rangle}_T^{(m_2)}$
(see \eqref{e:unbias} and \eqref{e:m2} respectively) are both included for comparison, even though these
require use of the unobservable $X_t$ process. The performance of the first-best
estimator in \cite{Zhang+2005} (denoted by $\widehat{\langle X, X\rangle}_T^{(s_1)}$) is also included as a well-performing and
tested estimator using only the $Y_t$ process, as is the naive estimator of
the realized volatility on $Y_t$ at the highest frequency, $\widehat{\langle X, X\rangle}_T^{(b)}$, given in \eqref{e:biased} (the fifth-best estimator in \cite{Zhang+2005}). We also include the performance of $\widehat{\langle X, X\rangle}_T^{(w)}$, defined in \eqref{whittle}.

Table \ref{tbl:1} shows that the new estimator, $\widehat{\langle X, X\rangle}_T^{(m_1)}$, is competitive with the first-best
approach in \cite{Zhang+2005}, $\widehat{\langle X, X\rangle}_T^{(s_1)}$, as an estimator of the integrated volatility
for the Heston model with the stated parameters. For this simulation the new method
performed marginally better. The similar performance of the two estimators
is quite remarkable, given their different approach; both estimators involve a bias-correction,
\cite{Zhang+2005} perform this globally by weighting different sampling
frequencies, whilst we correct locally at each frequency. The realized integrated
volatility of $Y_t$ at the highest frequency, $\widehat{\langle X, X\rangle}_T^{(b)}$, produces disastrous results,
as expected.

\begin{table}
\begin{center}
\begin{footnotesize}
\begin{tabular}{|l|c|c|c|}\hline
& Sample bias & Sample variance & Sample RMSE
\\ \hline
$\widehat{\langle X, X\rangle}_T^{(b)}$ & $1.17\times10^{-2}$ & $1.80\times10^{-8}$
& $1.17\times10^{-2}$ \\
$\widehat{\langle X, X\rangle}_T^{(s_1)}$ & $6.44\times10^{-7}$ & $2.76\times10^{-10}$ & $1.66\times10^{-5}$ \\
$\widehat{\langle X, X\rangle}_T^{(m_1)}$ & $2.90\times10^{-7}$ & $2.59\times10^{-10}$ & $1.61\times10^{-5}$ \\
$\widehat{\langle X, X\rangle}_T^{(w)}$ & $2.63\times10^{-7}$ & $2.59\times10^{-10}$ & $1.61\times10^{-5}$ \\
$\widehat{\langle X, X\rangle}_T^{(m_2)}$ & $1.39\times10^{-8}$ & $2.07\times10^{-10}$ & $1.44\times10^{-5}$ \\
$\widehat{\langle X, X\rangle}_T^{(u)}$ & $1.20\times10^{-8}$ & $2.06\times10^{-10}$ & $1.44\times10^{-5}$ \\\hline
\end{tabular}
\end{footnotesize} 
\caption{Simulation study comparing the new estimator with the
best estimator of \cite{Zhang+2005}. {\label{tbl:1}}}
\end{center}
\end{table}

We also note that $\widehat{\langle X, X\rangle}_T^{(m_1)}$ performs more or less identically to $\widehat{\langle X, X\rangle}_T^{(w)}.$ These two estimators can almost be used interchangeably due to the invariance property of a maximum likelihood estimator. This observation is born out by our simulation studies, and we henceforth only report results for $\widehat{\langle X, X\rangle}_T^{(m_1)}$. Note that the variance of  $\widehat{\langle X, X\rangle}_T^{(w)}$ can be found from \eqref{varwhittle2}. To compare theory with simulations we note that the average estimated standard deviation is $  1.6093\times 10^{-5}$ whilst the expression for the variance to leading order gives an expression for the standard deviation of
$\left[\var\left\{\widehat{\langle X, X\rangle}_T^{(w)}\right\}\right]^{1/2}=1.0246\times 10^{-5}$, using the parameter values of $\overline{\sigma}^2_X \approx 6.8 \times 10^{-9}$ and $\sigma^2_\varepsilon \approx 2.5\times 10^{-7}$.

A histogram of the observed bias of the new estimator is plotted in Figure \ref{fig:fig3} along with a histogram of the observed bias of the first-best estimator in \cite{Zhang+2005}. The observed
bias of our estimator follows a Gaussian distribution centred at zero,
suggesting that this estimator is unbiased, as out results claim to be true. Comparing our estimator to
the first-best estimator, it can be seen that the new estimator
has similar magnitudes of error also (hence the similar Root Mean Square Error (RMSE)).
\begin{figure}
\centerline{
\begin{tabular}{c@{\hspace{3pc}}c}
\includegraphics[width=2.5in, height = 2.5in]{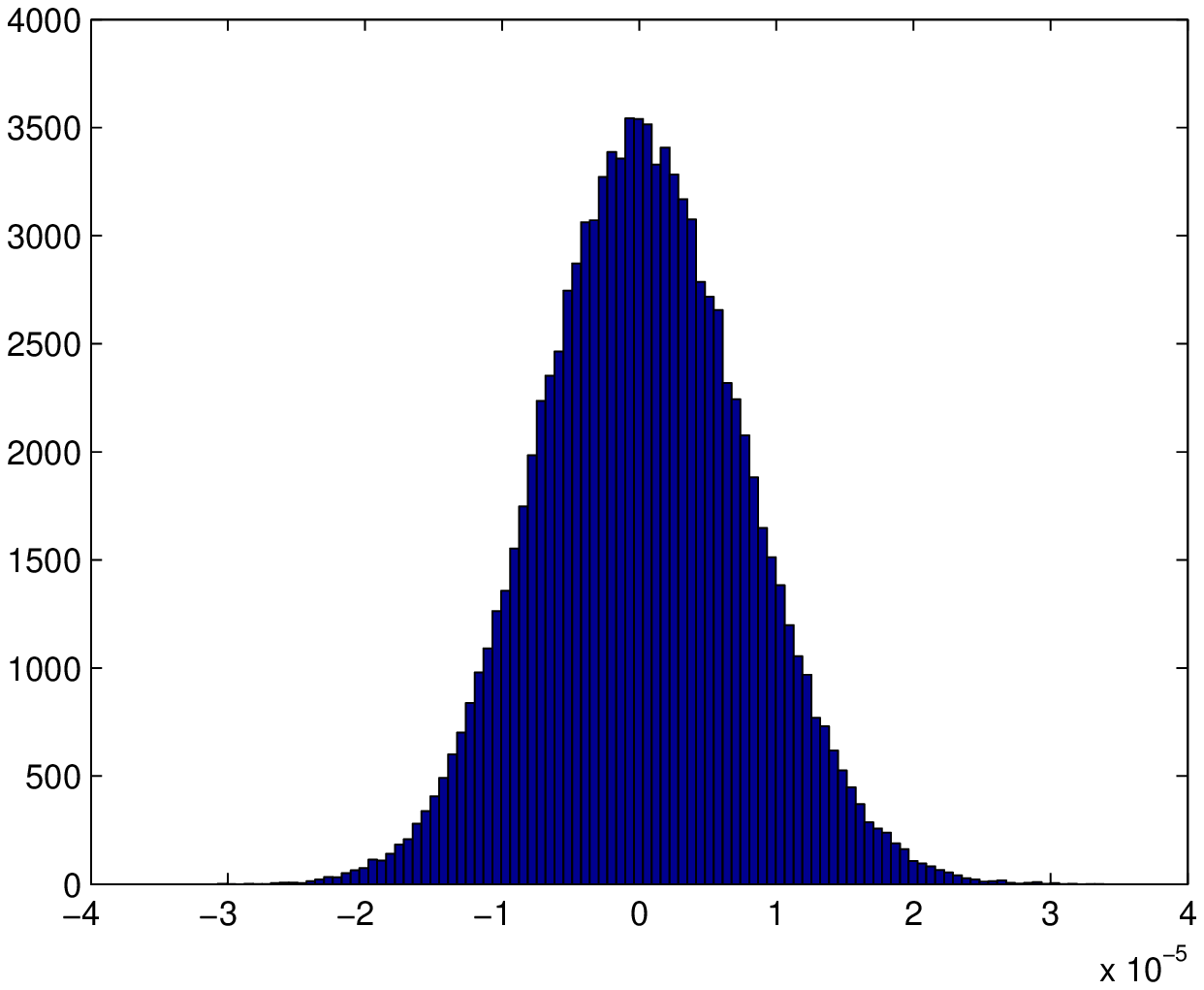} &
\includegraphics[width=2.5in, height = 2.5in]{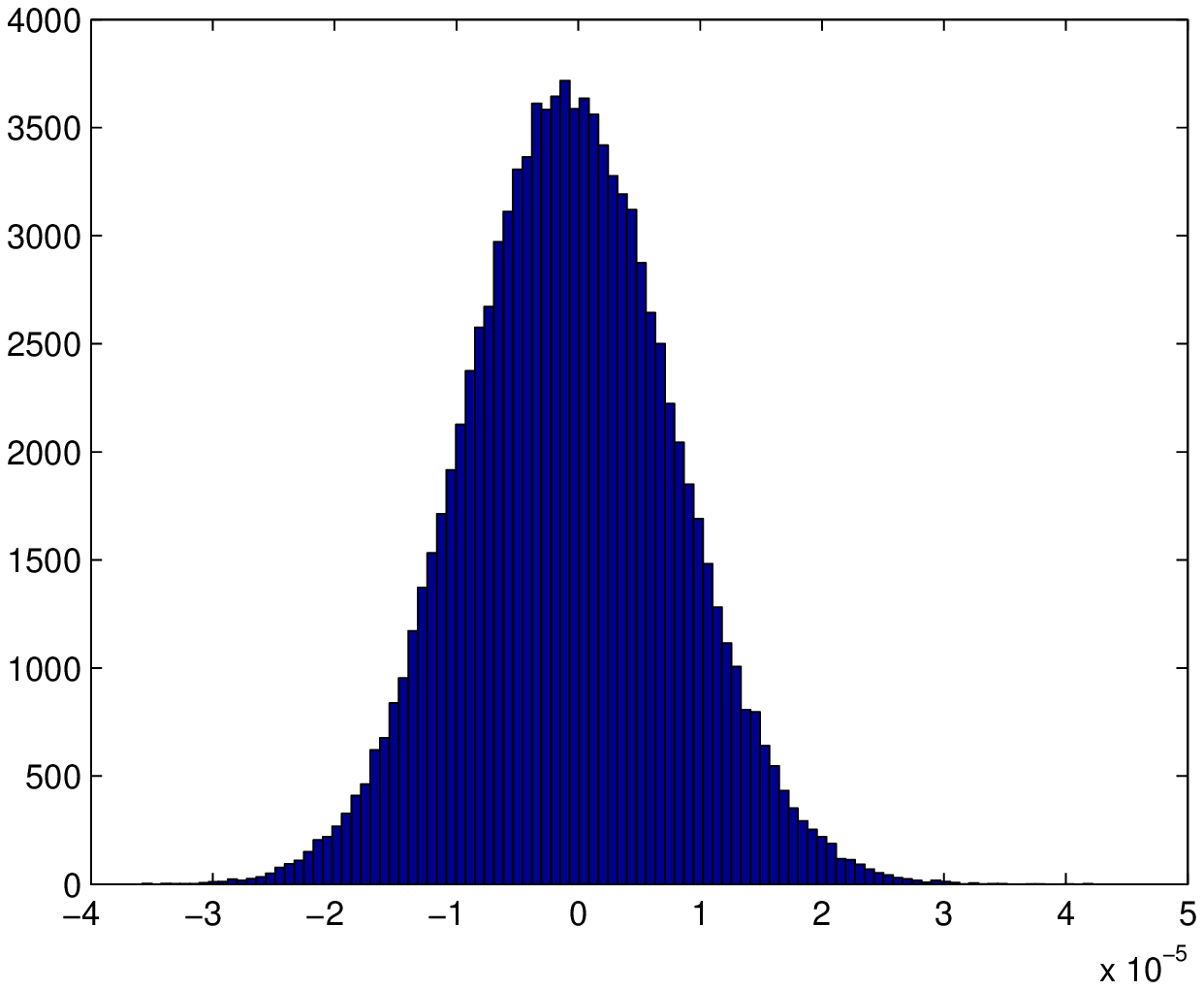} \\
\end{tabular}}
\begin{center}
\caption{The histograms of the observed bias of the proposed estimator (a), and the first-best estimator (b), over 100,000 sample paths.\label{fig:fig3}}
  \end{center}
\end{figure}

The new estimator requires calculation of $\widehat\sigma^2_X$ and $\widehat\sigma^2_\varepsilon$ which will vary over each process due to
the limited information given from the $Y_t$ process. The stability of this
estimation is of great importance if the estimator is to perform well.
Figure \ref{fig:fig4} shows the distribution of the parameters $\widehat\sigma^2_X$ and $\widehat\sigma^2_\varepsilon$ over the simulated paths. The parameter
estimation is quite consistent, with all values estimated within a narrow
range. Figure \ref{fig:fig0} suggests that these estimates are roughly unbiased; as $\overline{\sigma}^2_X \approx 6.8 \times 10^{-9}$ and $\sigma^2_\varepsilon \approx 2.5 \times 10^{-7}$ (as $\sigma^2_\varepsilon\left|2 \sin(\pi f_k )\right|^2 \approx 1 \times 10^{-6}$,
at $f_k=0.5$).
\begin{figure}
\centerline{
\begin{tabular}{c@{\hspace{3pc}}c}
\includegraphics[width=2.5in, height = 2.5in]{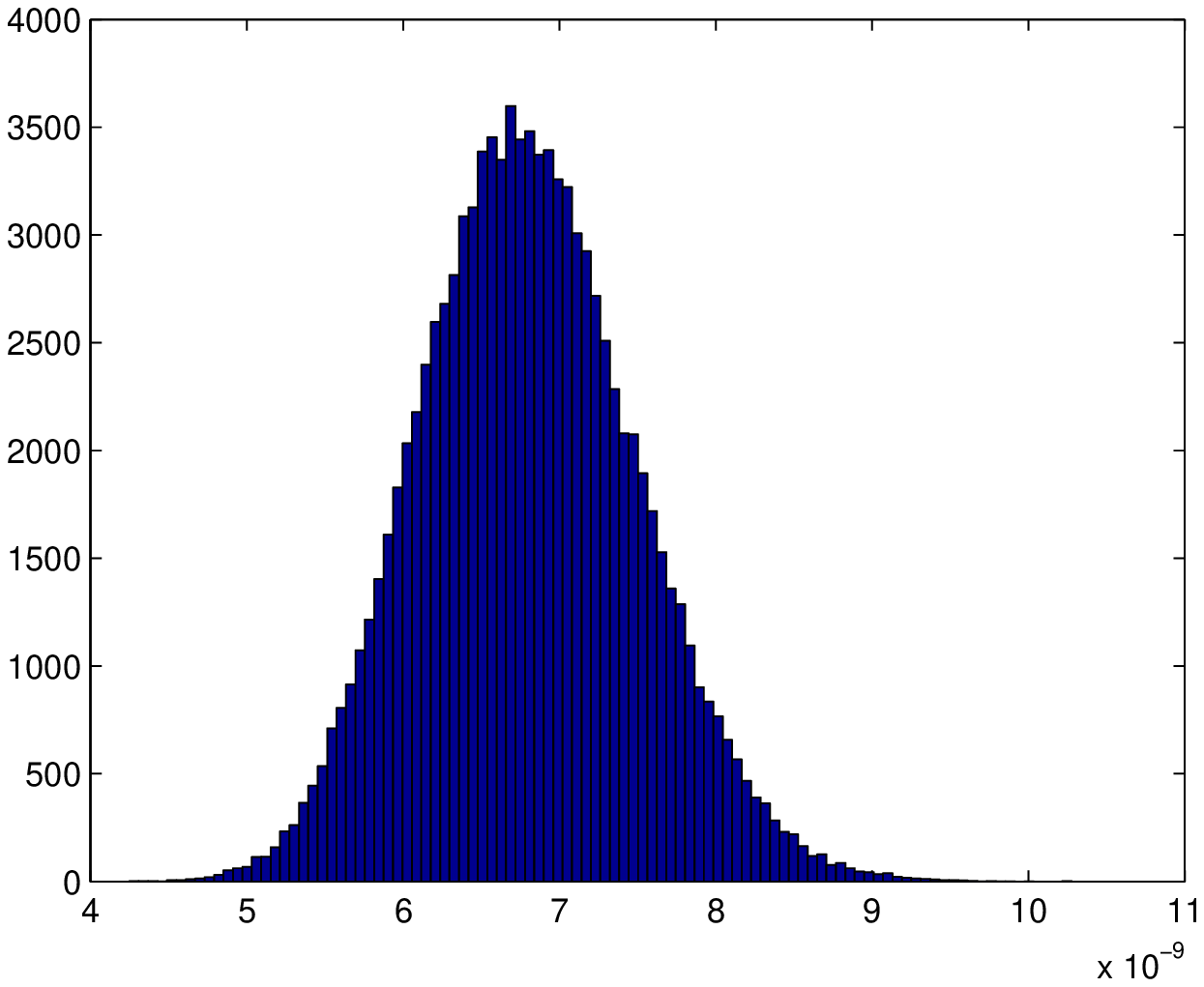} &
\includegraphics[width=2.5in, height = 2.5in]{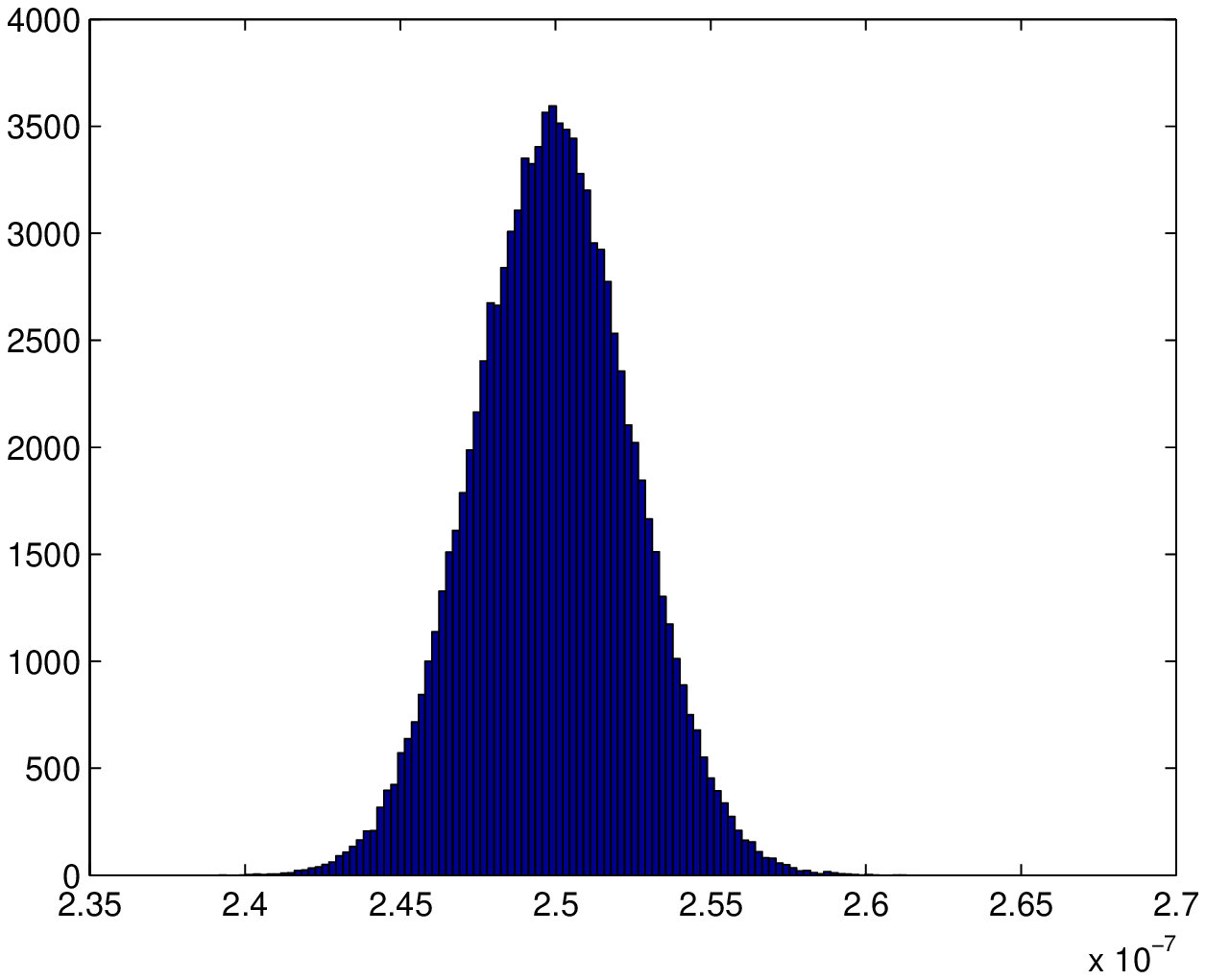} \\
\end{tabular}}
\begin{center}
\caption{The histograms of the estimated $\overline{\sigma}_X^2$ (a) and $\sigma^2_{\varepsilon}$
(b).\label{fig:fig4}}
  \end{center}
\end{figure}

\subsection{Brownian Process and Ornstein Uhlenbeck Process}
We repeated our simulations for a Brownian Process given by:
\begin{equation}
dX_t=\sqrt{2\sigma^2_t} dB_t,
\end{equation}
where $\sigma_t^2=0.01$. We otherwise keep the same simulation setup as before
with 50,000 simulated paths of length 23,400. The results are displayed in
Table \ref{tbl:1a}. The new estimator, $\widehat{\langle X, X\rangle}_T^{(m_1)}$,
again delivers a marked improvement
on the naive estimator, $\widehat{\langle X, X\rangle}_T^{(b)}$, and performs marginally better than the first-best
estimator in \cite{Zhang+2005}, $\widehat{\langle X, X\rangle}_T^{(s_1)}$.

\begin{table}
\begin{center}
\begin{footnotesize}
\begin{tabular}{|l|c|c|c|}\hline
& Sample bias & Sample variance & Sample RMSE
\\ \hline
$\widehat{\langle X, X\rangle}_T^{(b)}$ & $1.17\times10^{-2}$ & $1.77\times10^{-8}$
& $1.17\times10^{-2}$ \\
$\widehat{\langle X, X\rangle}_T^{(s_1)}$ & $6.52\times10^{-7}$ & $2.68\times10^{-11}$ & $5.22\times10^{-6}$ \\
$\widehat{\langle X, X\rangle}_T^{(m_1)}$ & $3.02\times10^{-7}$ & $1.98\times10^{-11}$ & $4.46\times10^{-6}$ \\
$\widehat{\langle X, X\rangle}_T^{(m_2)}$ & $1.96\times10^{-9}$ & $6.93\times10^{-13}$ & $8.32\times10^{-7}$ \\
$\widehat{\langle X, X\rangle}_T^{(u)}$ & $3.79\times10^{-9}$ & $5.44\times10^{-13}$ & $7.38\times10^{-7}$ \\\hline
\end{tabular}
\end{footnotesize} 
\caption{Simulation study for the Brownian process. {\label{tbl:1a}}}
\end{center}
\end{table}

We also performed a Monte Carlo simulation for the Ornstein Uhlenbeck process given by:
\begin{equation}
dX_t=X_tdt+\sqrt{2\sigma_t} dB_t,
\end{equation}
where also $\sigma_t^2=0.01$. Again we retain the same simulation setup and the
results are displayed in Table \ref{tbl:1b}. The results are almost identical
to that of the Brownian process, with the new estimator again outperforming
other time-domain estimators.

\begin{table}
\begin{center}
\begin{footnotesize}
\begin{tabular}{|l|c|c|c|}\hline
& Sample bias & Sample variance & Sample RMSE
\\ \hline
$\widehat{\langle X, X\rangle}_T^{(b)}$ & $1.17\times10^{-2}$ & $1.78\times10^{-8}$
& $1.17\times10^{-2}$ \\
$\widehat{\langle X, X\rangle}_T^{(s_1)}$ & $6.69\times10^{-7}$ & $2.66\times10^{-11}$ & $5.20\times10^{-6}$ \\
$\widehat{\langle X, X\rangle}_T^{(m_1)}$ & $2.95\times10^{-7}$ & $1.97\times10^{-11}$ & $4.44\times10^{-6}$ \\
$\widehat{\langle X, X\rangle}_T^{(m_2)}$ & $5.09\times10^{-9}$ & $6.76\times10^{-13}$ & $8.22\times10^{-7}$ \\
$\widehat{\langle X, X\rangle}_T^{(u)}$ & $6.29\times10^{-9}$ & $5.33\times10^{-13}$ & $7.30\times10^{-7}$ \\\hline
\end{tabular}
\end{footnotesize} 
\caption{Simulation study for the Ornstein Uhlenbeck process. {\label{tbl:1b}}}
\end{center}
\end{table}

\subsection{Comparing estimators over shorter sample lengths}
This section compares estimators for a shorter sample length which will reduce the
benefit of subsampling due to the variance issues of small-length data but will also affect the variance of the multiscale ratio ({\em cf} Theorem \ref{thm:thmA}).

The simulation setup is exactly the same as before (using the Heston model with the same parameters) except that $T$, the simulation length, is reduced by a factor of 10 to 0.1 days or $2340s$. Before
the results of the simulation are reported, it is of interest to see whether
the frequency domain methods developed still model each process accurately.
Figure \ref{fig:fig5} shows the calculated $\widehat\sigma^2_X$
and $\widehat\sigma^2_\varepsilon \left|\sin(\pi \Delta t f_k) \right|^2$ (in white) together with the periodograms $\widehat{\mathcal S}^{(X)}_{kk}$
and $\widehat{\mathcal S}^{(\varepsilon)}_{kk}$ for one simulated path. The estimator still approximates the energy structure of the processes accurately. Figure \ref{fig:fig5} also shows the corresponding estimate of the multiscale ratio $\widehat{L}_k$ (in white) from this simulated path (together with $\widetilde{L}_k$) and the corresponding plot
of $\widehat{L}_k$$\widehat{\mathcal S}^{(Y)}_{kk}$. The new estimator has removed the microstructure noise effect and has
formed a good approximation of $\widehat{\mathcal S}^{(X)}_{kk}$. The approximation
of the periodograms is still accurate despite the shortening of available
data.

\begin{figure}
\centerline{
\begin{tabular}{c@{\hspace{1pc}}c}
\includegraphics[width=2.5in, height = 2.5in]{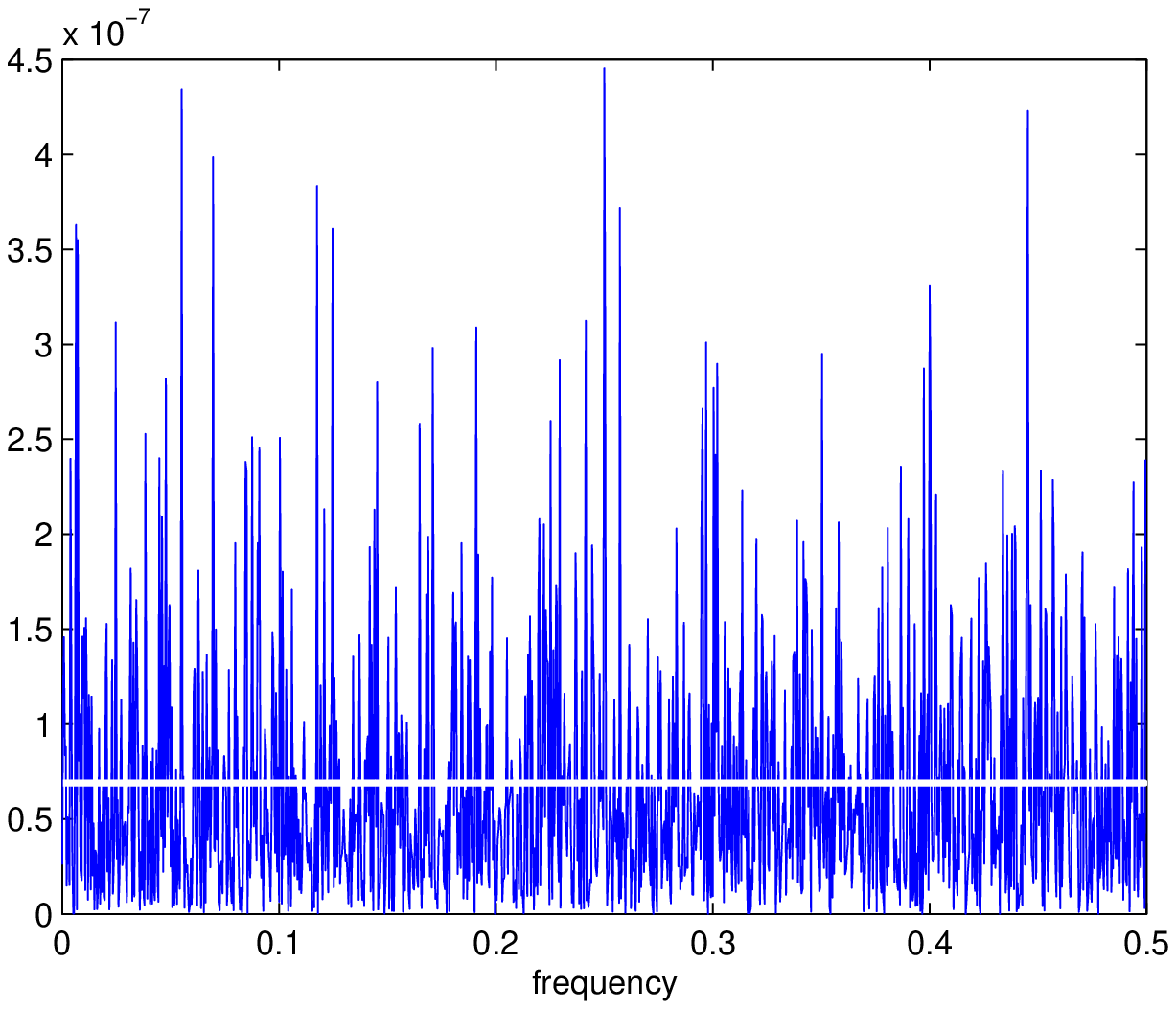} &
\includegraphics[width=2.5in, height = 2.5in]{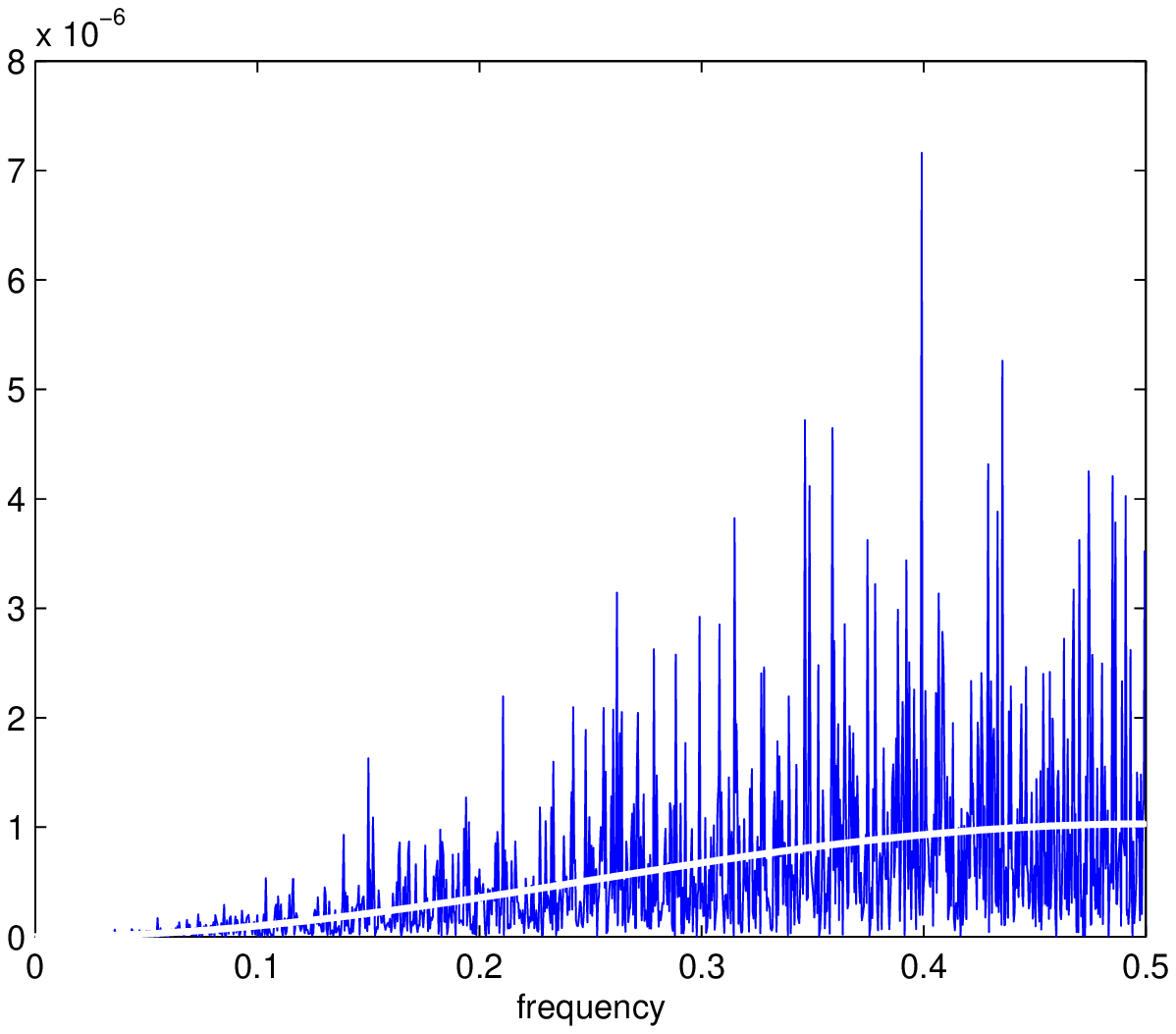} \\
\includegraphics[width=2.5in, height = 2.5in]{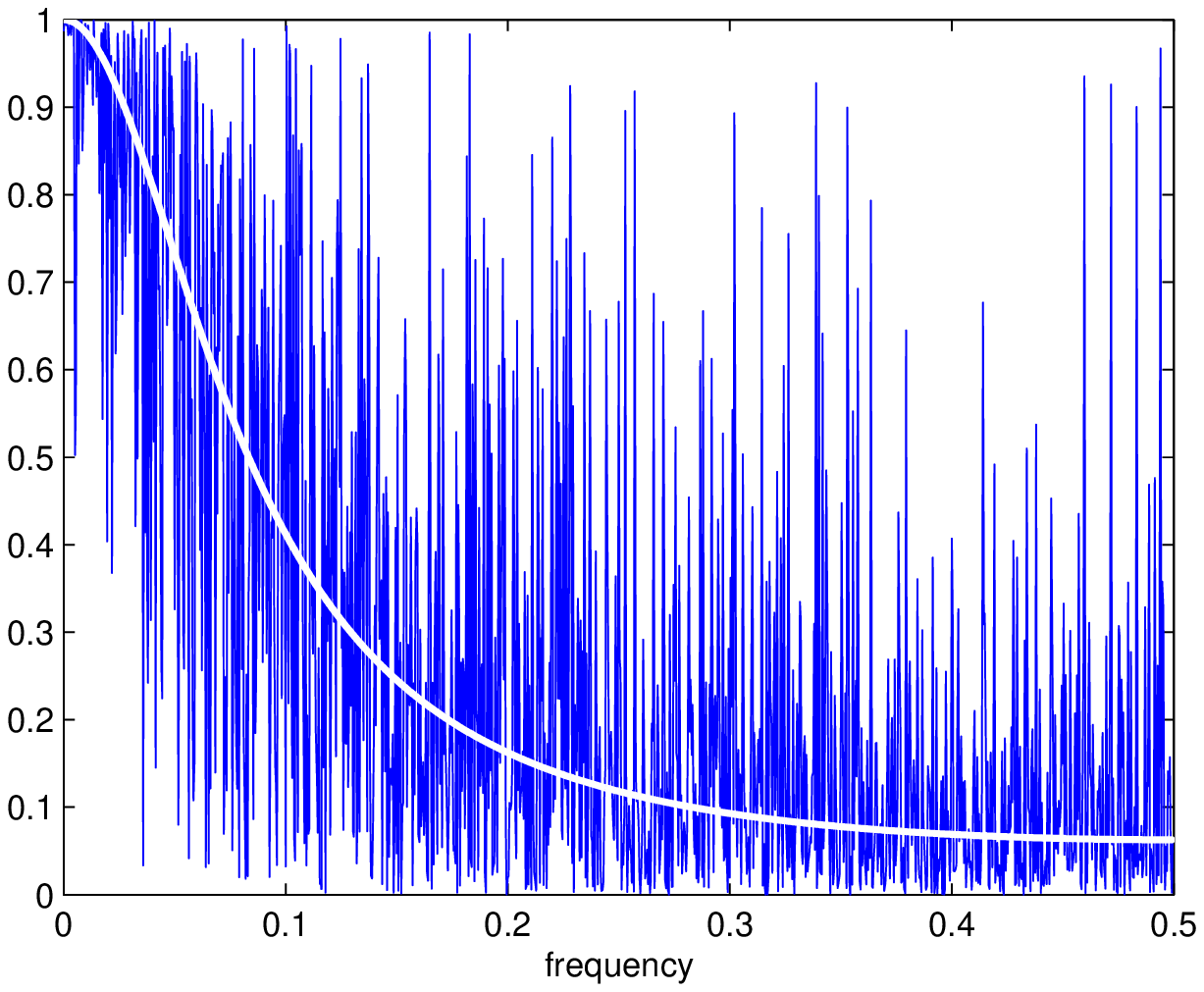} &
\includegraphics[width=2.5in, height = 2.5in]{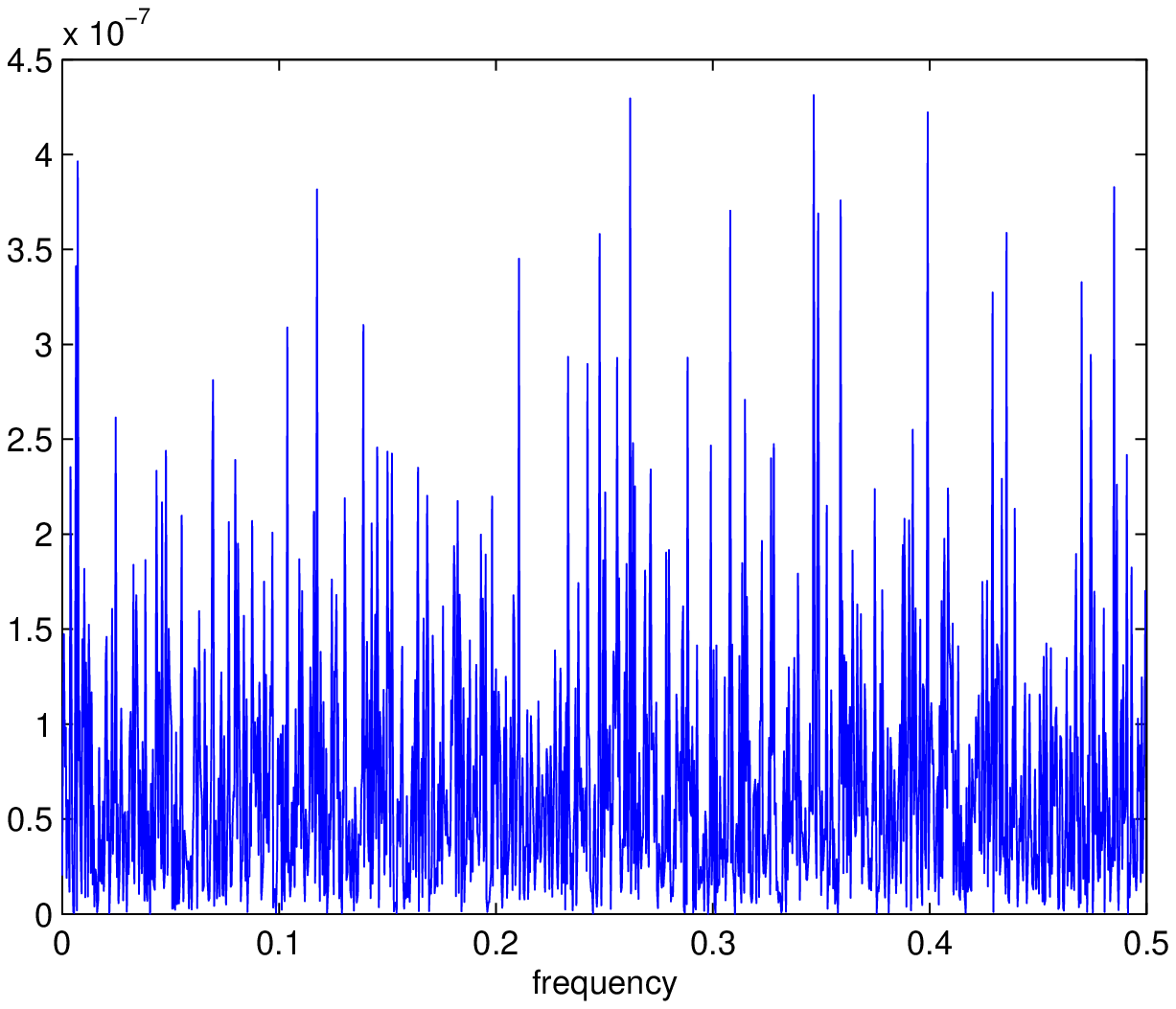} \\
\end{tabular}}
\begin{center}
\caption{A realisation of $\widehat{\mathcal S}^{(X)}_{kk}$ (top left),
 a realisation of $\widehat{\mathcal S}^{(\varepsilon)}_{kk}$ (top right) with the Whittle estimates superimposed, the estimate of $L_k$ (bottom left) with the Whittle estimate of $L_k$ superimposed
 and the biased corrected estimator
 of ${\mathcal S}^{(X)}_{kk}$ using $\widehat{L}_k\widehat{\mathcal S}^{(Y)}_{kk}$ (bottom right). Notice the different scales in the
four figures.\label{fig:fig5}}
\end{center}
\end{figure}

Table \ref{tbl:2} displays the accuracy of the estimators over the 50,000 simulated paths. The first-best estimator in \cite{Zhang+2005}, $\widehat{\langle X, X\rangle}_T^{(s_1)}$, and the new estimator, $\widehat{\langle X, X\rangle}_T^{(m_1)}$,
are once again comparable in performance and both estimates are close to the best attainable RMSE given by, $\widehat{\langle X, X\rangle}_T^{(u)}$, the realized integrated volatility on $X_t$.

\begin{table}
\begin{center}
\begin{footnotesize}
\begin{tabular}{|l|c|c|c|}\hline
& Sample bias & Sample variance & Sample RMSE
\\ \hline
$\widehat{\langle X, X\rangle}_T^{(b)}$ & $1.17\times10^{-3}$ & $2.29\times10^{-9}$
& $1.17\times10^{-3}$ \\
$\widehat{\langle X, X\rangle}_T^{(s_1)}$ & $1.00\times10^{-6}$ & $4.51\times10^{-10}$ & $2.13\times10^{-5}$ \\
$\widehat{\langle X, X\rangle}_T^{(m_1)}$ & $1.84\times10^{-7}$ & $4.23\times10^{-10}$ & $2.06\times10^{-5}$ \\
$\widehat{\langle X, X\rangle}_T^{(m_2)}$ & $4.80\times10^{-8}$ & $2.42\times10^{-10}$ & $1.55\times10^{-5}$ \\
$\widehat{\langle X, X\rangle}_T^{(u)}$ & $5.27\times10^{-8}$ & $2.28\times10^{-10}$ & $1.51\times10^{-5}$ \\\hline
\end{tabular}
\end{footnotesize} 
\caption{Simulation study for shorter sampler length. {\label{tbl:2}}}
\end{center}
\end{table}

\subsection{Comparing estimators with a low-noise process}
This section compares estimators for smaller levels of microstructure noise. Reducing the microstructure
noise will
reduce the need to subsample. The first-best estimator in \cite{Zhang+2005},
$\widehat{\langle X, X\rangle}_T^{(s_1)}$, will have a higher sampling
frequency and the new estimator
will reduce its estimate of $\widehat{\sigma}^2_{\varepsilon}$
accordingly. For very small levels of noise, however,
the first-best estimator will become zero, as the optimal number of samples
becomes $n$ (the highest available). This possibility is now examined, using the Heston model as
before, with all parameters unchanged except the noise is reduced by a factor
of 10, ie. $\sigma^2_{\varepsilon} = 0.00005^2$. Note that the path length is
kept at its original length of $T=1$ day.

Figure \ref{fig:fig6} shows the estimates of $\widehat{\sigma}_X^2$
and $\widehat\sigma^2_\varepsilon\left|2\sin(\pi \Delta t f_k) \right|^2$ (in white) along with the periodograms $\widehat{\mathcal S}^{(Y)}_{kk}$
and $\widehat{\mathcal S}^{(\varepsilon)}_{kk}$ for one simulated path along with the corresponding estimate of the multiscale ratio $\widehat{L}_k$ (in white) (plotted over the approximated
$\widetilde{L}_k$) and the corresponding plot
of $\widehat{L}_k$$\widehat{\mathcal S}^{(Y)}_{kk}$. The estimation method
works well again; notice how the magnitude of the microstructure
noise has been greatly reduced (the scale is now of order $10^{-8}$ rather
than $10^{-6}$) causing the multiscale ratio $L_k$ to be more tempered across
the high frequencies than it was before, due to the
smaller microstructure noise.
Nonetheless, the new estimator has still detected the smaller levels of noise in
the data. 

\begin{figure}
\centerline{
\begin{tabular}{c@{\hspace{1pc}}c}
\includegraphics[width=2.5in, height = 2.5in]{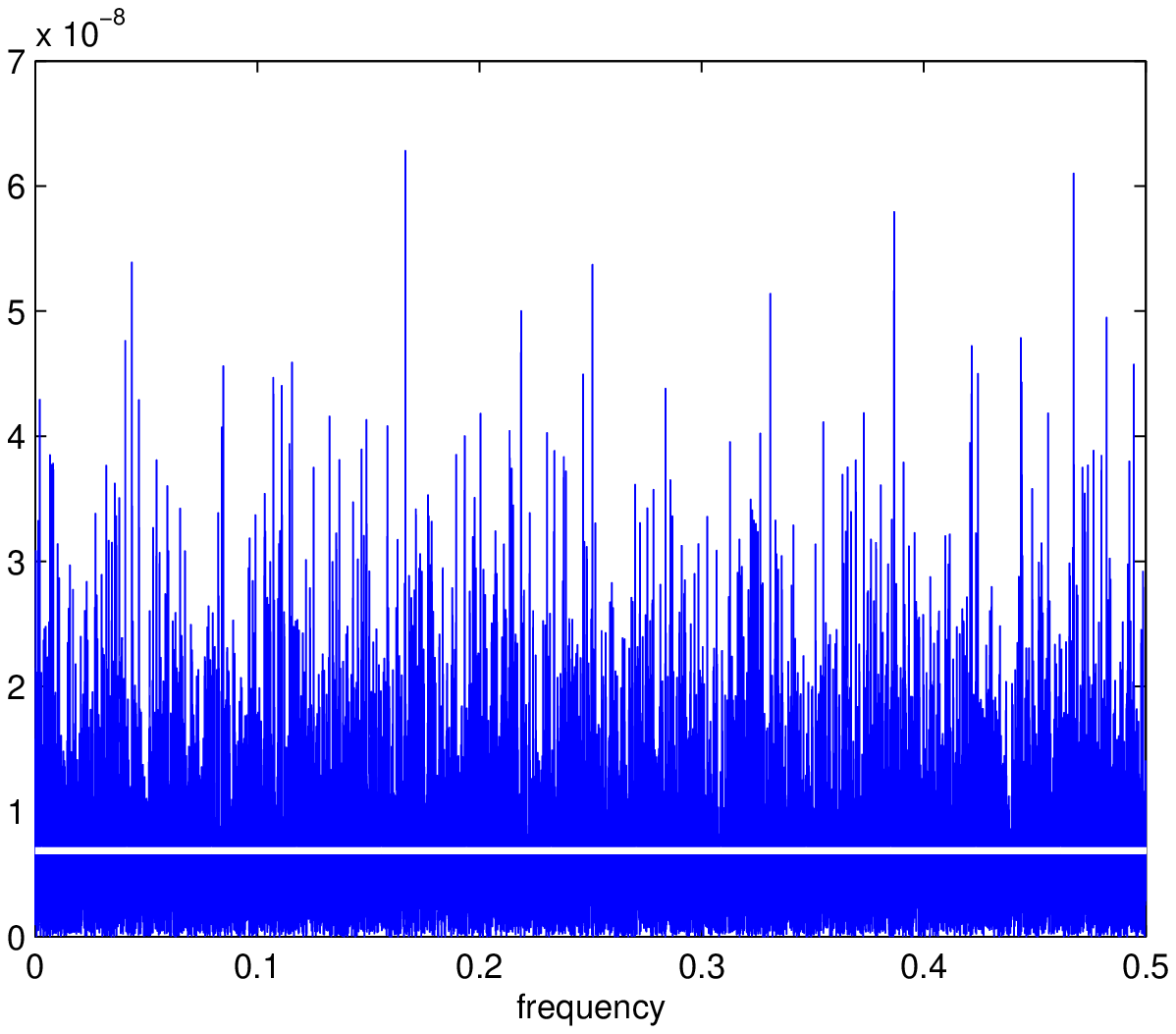} &
\includegraphics[width=2.5in, height = 2.5in]{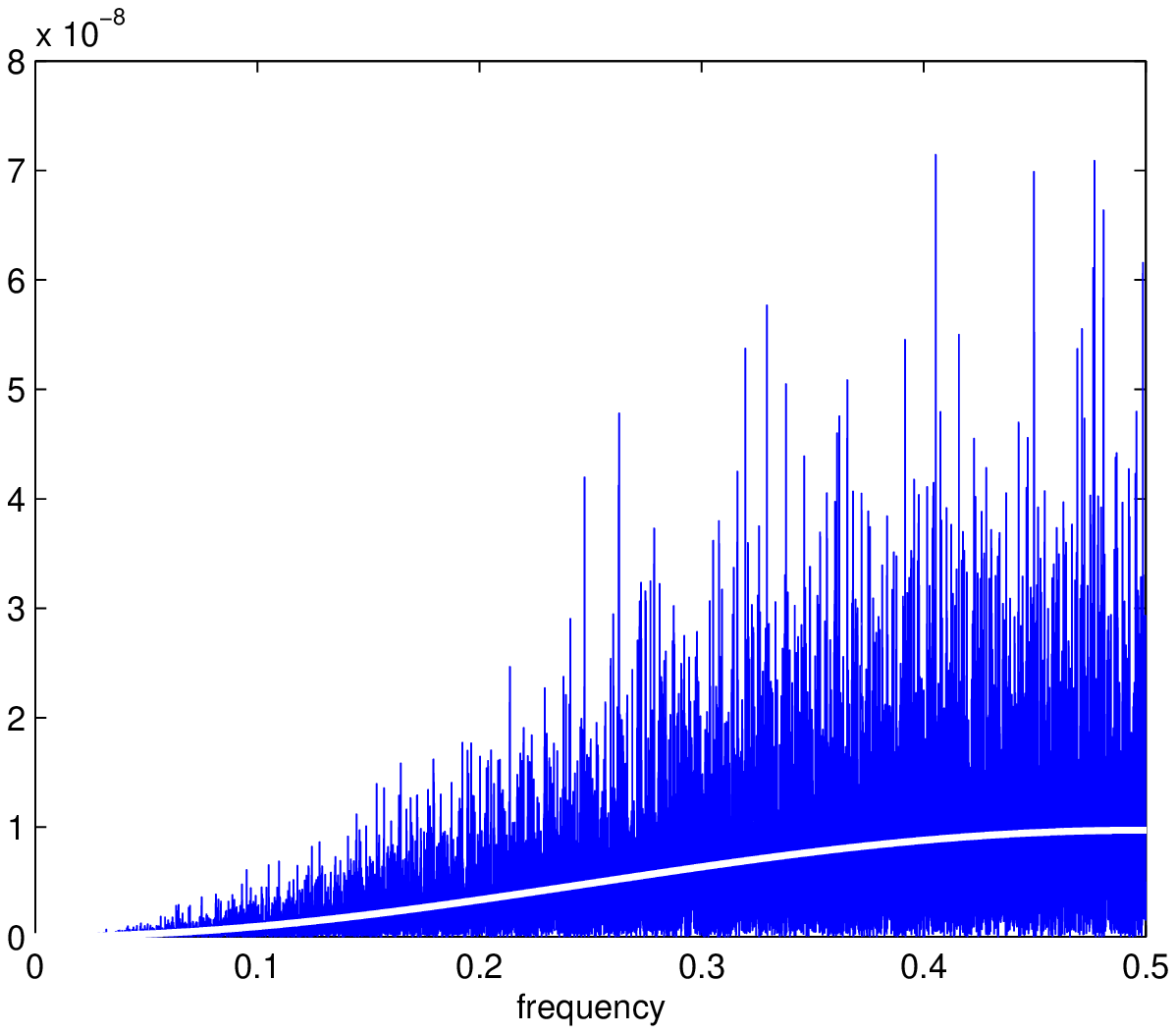} \\
\includegraphics[width=2.5in, height = 2.5in]{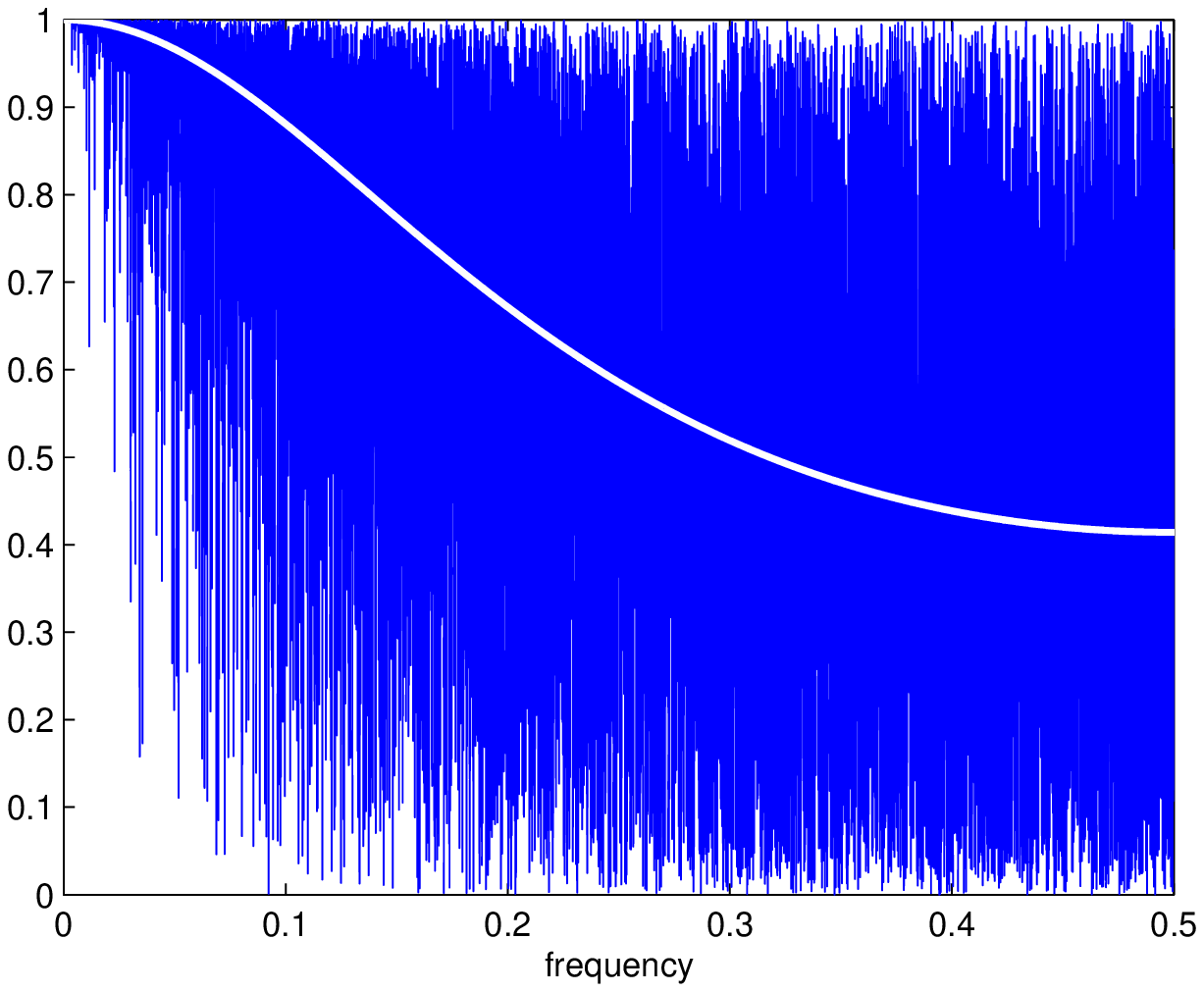} &
\includegraphics[width=2.5in, height = 2.5in]{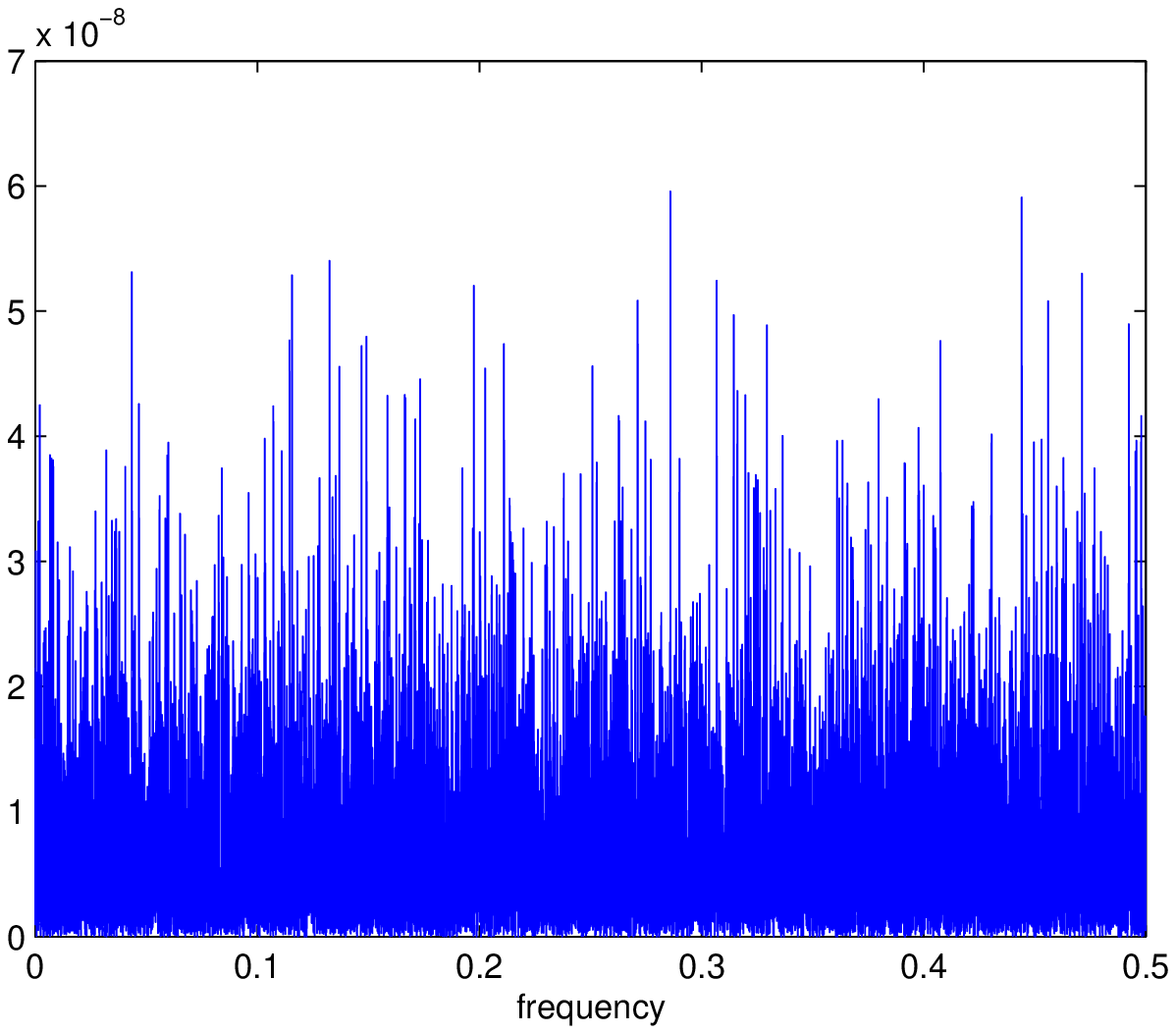} \\
\end{tabular}}
\begin{center}
\caption{A realisation of $\widehat{\mathcal S}^{(X)}_{kk}$ (top left),
 a realisation of $\widehat{\mathcal S}^{(\varepsilon)}_{kk}$ (top right) with the Whittle estimates superimposed, the estimate of $L_k$ (bottom left) with the Whittle estimate of $L_k$ superimposed
 and the biased corrected estimator
 of ${\mathcal S}^{(X)}_{kk}$ using $\widehat{L}_k\widehat{\mathcal S}^{(Y)}_{kk}$ (bottom right). Notice the different scales in the
four figures.\label{fig:fig6}}
\end{center}
\end{figure}

Table \ref{tbl:3} reports on the results of 50,000 simulations performed
as before.
The first-best estimator of \cite{Zhang+2005}, $\widehat{\langle X, X\rangle}_T^{(s_1)}$, categorically failed for
this model. This is due to the fact that the optimal number of samples was
always equal to $n$, the total number of samples available. Therefore,
the first-best estimator was always zero. The second-best estimator in \cite{Zhang+2005},
denoted
by $\widehat{\langle X, X\rangle}_T^{(s_2)}$, was reasonably effective. This is simply an estimator that averages estimates
calculated from sub-sampled paths at different starting points and is therefore
asymptotically biased. The new estimator, $\widehat{\langle X, X\rangle}_T^{(m_1)}$,
was remarkably robust, with RMSE very
close to the RMSE of estimators based on the $X_t$ process. The
difference in performance between estimators using $Y_t$ and estimators using $X_t$
is expected to become smaller with less microstructure noise and this can
be seen by the similar order RMSE errors between all estimators. Nevertheless,
the new estimator was much closer in performance to the realized integrated volatility
on $X_t$ than it was to any other estimator on $Y_t$, a result that demonstrates
the precision and robustness of this new estimator of integrated volatility.

\begin{table}
\begin{center}
\begin{footnotesize}
\begin{tabular}{|l|c|c|c|}\hline
& Sample bias & Sample variance & Sample RMSE
\\ \hline
$\widehat{\langle X, X\rangle}_T^{(b)}$ & $1.17\times10^{-4}$ & $2.11\times10^{-10}$
& $1.18\times10^{-4}$ \\
$\widehat{\langle X, X\rangle}_T^{(s_2)}$ & $3.53\times10^{-6}$ & $1.00\times10^{-9}$ & $3.19\times10^{-5}$ \\
$\widehat{\langle X, X\rangle}_T^{(m_1)}$ & $7.63\times10^{-9}$ & $2.12\times10^{-10}$ & $1.46\times10^{-5}$ \\
$\widehat{\langle X, X\rangle}_T^{(m_2)}$ & $7.91\times10^{-9}$ & $2.06\times10^{-10}$ & $1.44\times10^{-5}$ \\
$\widehat{\langle X, X\rangle}_T^{(u)}$ & $9.83\times10^{-9}$ & $2.05\times10^{-10}$ & $1.43\times10^{-5}$ \\\hline
\end{tabular}
\end{footnotesize} 
\caption{Simulation study for lower market microstructure noise. {\label{tbl:3}}}
\end{center}
\end{table}

\subsection{Correlated Noise\label{excorr1}}
In this section we consider microstructure noise that is correlated. If this
process is stationary, the
noise process can be modelled as an MA process (as described in Section \ref{subsec:corr}),
and the corresponding parameters
can be estimated by maximising the multiscale Whittle likelihood using \eqref{e:multiscaler2}
and \eqref{eq:MA}. Figure \ref{fig1d} shows the multiscale estimator applied
to the Heston Model (with the same parameters as before) with a microstructure noise that follows an MA(6) process (parameters given in the caption). The Whittle estimates (in white) form
a good approximation of $\widehat{\mathcal S}^{(X)}_{kk}$
and $\widehat{\mathcal S}^{(\varepsilon)}_{kk}$ despite the more complicated
nuisance structure. The corresponding estimate of the multiscale ratio $\widehat{L}_k$ (in white) therefore removes energy from the correct frequencies and the corresponding plot
of $\widehat{L}_k$$\widehat{\mathcal S}^{(Y)}_{kk}$ is a good approximation of $\widehat{\mathcal S}^{(X)}_{kk}$. This is the same noise process and It\^{o} process for which we calculated the optimal smoothing window in section \ref{subsec:corr}, and the trough in the noise at about $f=0.42$ corresponds to the oscillations in the kernel plotted in Figure \ref{fig:fig-1}.

\begin{figure}[!htbp]
\centerline{
\begin{tabular}{c@{\hspace{1pc}}c}
\includegraphics[width=2.5in, height = 2.5in]{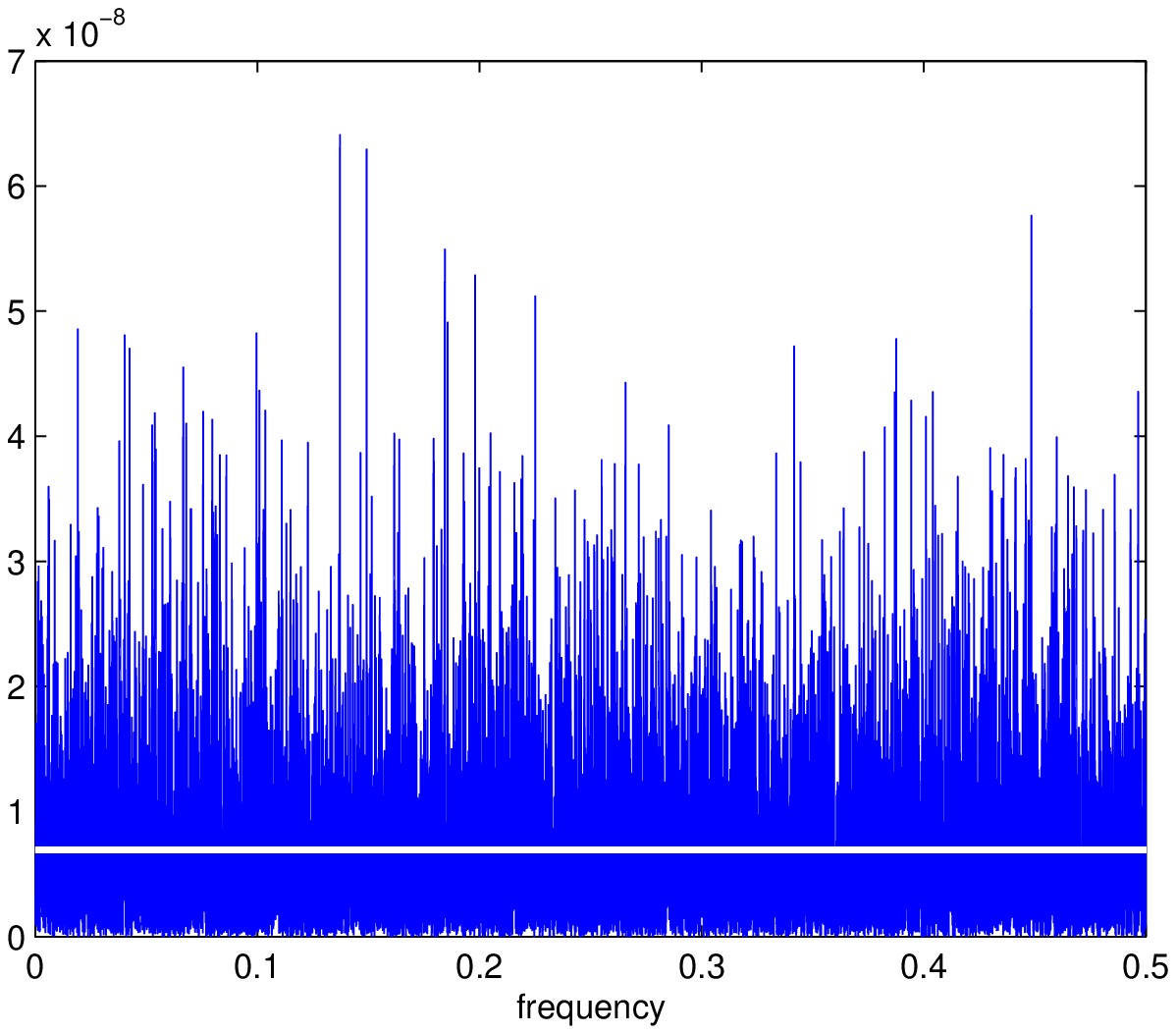} &
\includegraphics[width=2.5in, height = 2.5in]{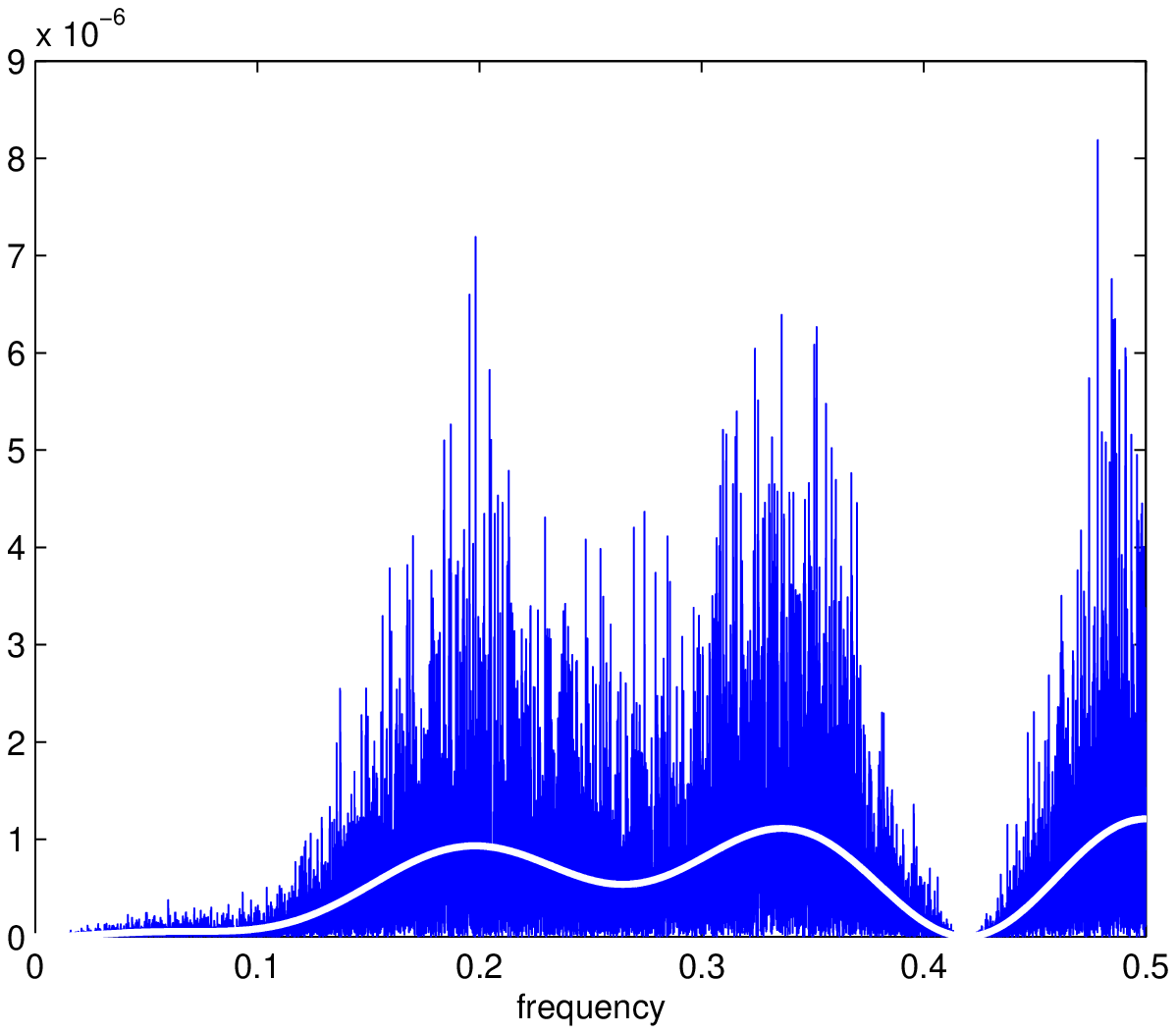} \\
\includegraphics[width=2.5in, height = 2.5in]{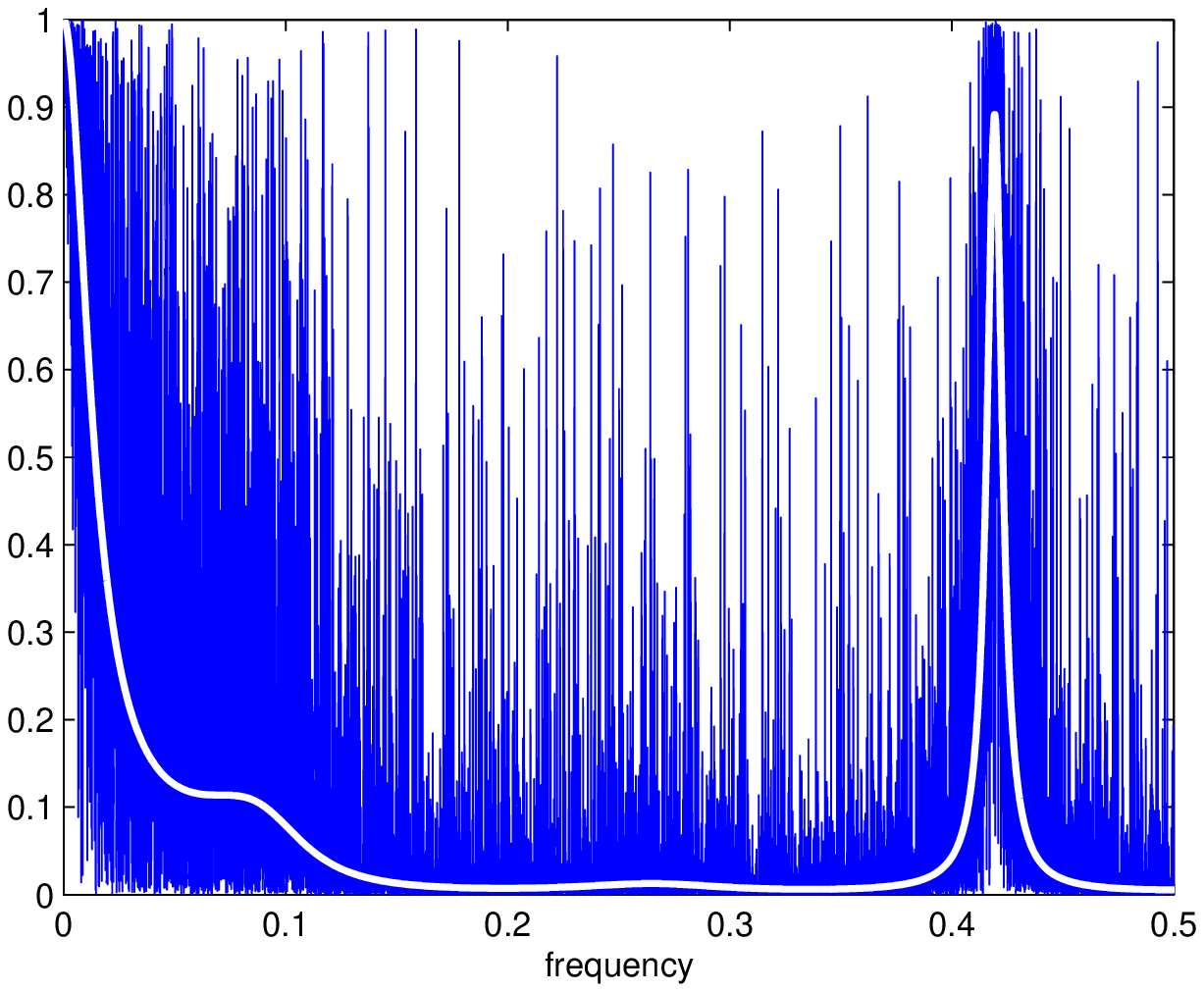} &
\includegraphics[width=2.5in, height = 2.5in]{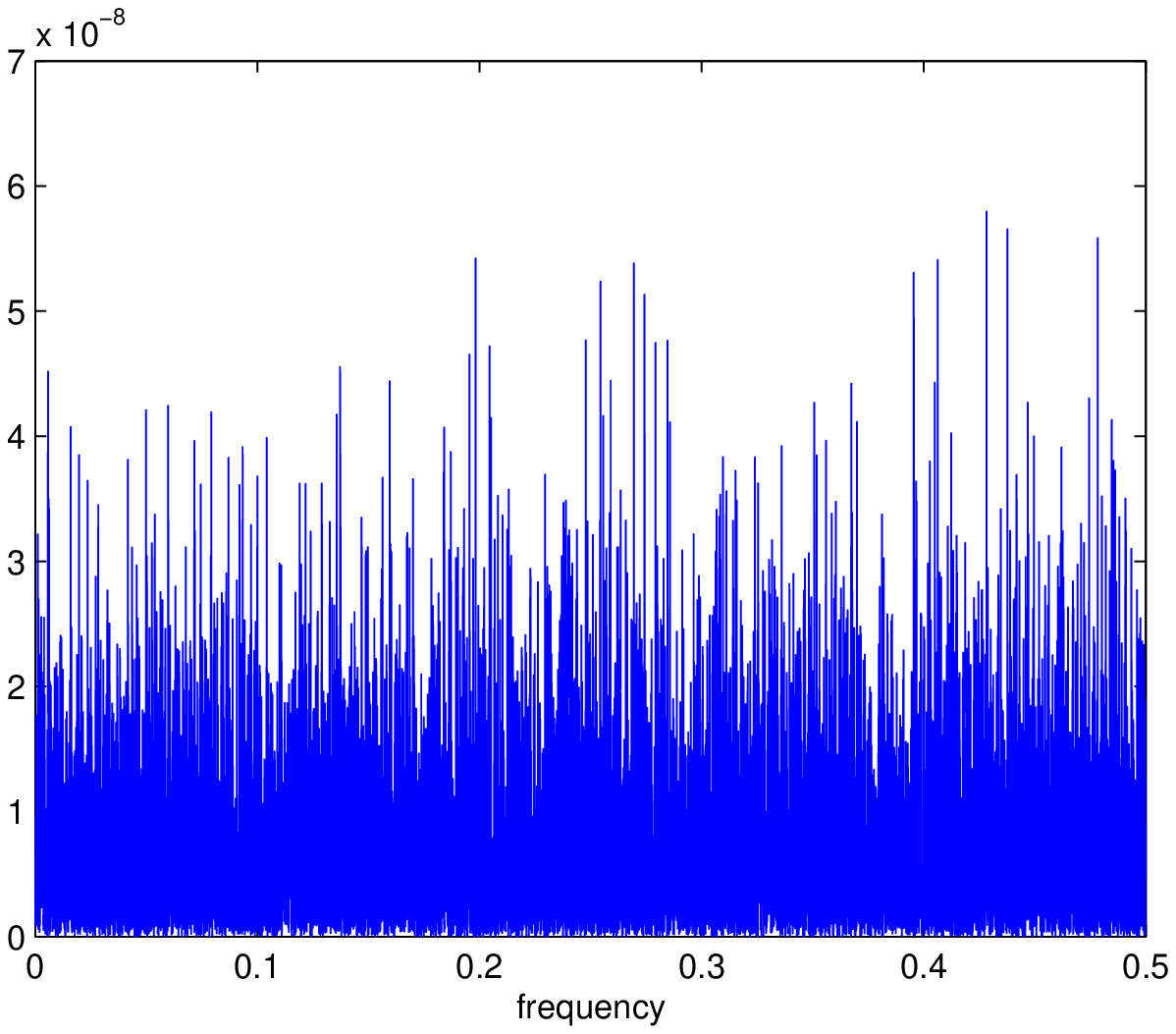} \\
\end{tabular}}
\begin{center}
\caption{A realisation of $\widehat{\mathcal S}^{(X)}_{kk}$ (top left),
 a realisation of $\widehat{\mathcal S}^{(\varepsilon)}_{kk}$ (top right) with the Whittle estimates superimposed, the estimate of $L_k$ (bottom left) with the Whittle estimate of $L_k$ superimposed
 and the biased corrected estimator
 of ${\mathcal S}^{(X)}_{kk}$ using $\widehat{L}_k\widehat{\mathcal S}^{(Y)}_{kk}$ (bottom right). In this example we use an MA(6) with $\theta_1=0.5$, $\theta_2=-0.1$, $\theta_3=-0.1$, $\theta_4=0.2$, $\theta_5=0$ and $\theta_6=0.4$.
Notice the different scales in the
four figures.\label{fig1d}}
\end{center}
\end{figure}

If the length of the MA$(p)$ process is unknown, then $p$ can be determined
using \eqref{e:aicc}. In Table \ref{tbl:4} we show an example with $p=4$ with paramaters $\theta_1=0.8$, $\theta_2=-0.6$, $\theta_3=0.1$, $\theta_4=0.4$, Clearly $p=4$ is identified as the best fitting model yielding near to perfect estimates of the noise parameters. The estimator is therefore robust to removing
the effect of microstructure noise when this process is correlated (and stationary),
even if the length of the MA$(p)$ process is not explicitly known. 

\begin{table}
\begin{center}
\begin{footnotesize}
\begin{tabular}{|l|c|c|c|c|c|c|c|c|c|}\hline
MA($p$) & $\theta_1$ & $\theta_2$ & $\theta_3$ & $\theta_4$ & $\theta_5$ & $\theta_6$ & $\theta_7$ & $\theta_8$ & AICC
\\ \hline
$p=1$ & 0.935 & & & & & & & & $-3.208490\times10^{5}$\\
$p=2$ & 0.624 & -0.445 & & & & & & & $-3.239947\times10^{5}$\\ 
$p=3$ & 0.658 & -0.459 & -0.046 & & & & & & $-3.240000\times10^{5}$\\ \hline
$p=4$ & 0.806 & -0.603 & -0.101 & 0.410 & & & & & $-3.262427\times10^{5}$\\ \hline
$p=5$ & 0.813 & -0.606 & -0.101 & 0.411 & -0.008 & & & & $-3.262416\times10^{5}$\\
$p=6$ & 0.815 & -0.604 & -0.097 & 0.420 & -0.003 & 0.000 & & & $-3.262409\times10^{5}$\\
$p=7$ & 0.807 & -0.613 & -0.114 & 0.413 & 0.002 & -0.002 & -0.005 & & $-3.262402\times10^{5}$\\
$p=8$ & 0.817 & -0.614 & -0.128 & 0.427 & 0.005 & 0.011 & -0.009 & -0.017 & $-3.262384\times10^{5}$\\\hline
\end{tabular}
\end{footnotesize} 
\caption{ Values of $\theta$ found by modelling the noise process as an MA(p) process for $p=1,\ldots,8$.  Model choice methods (AICC) are used to select which process to model the noise by, in this case the AICC is minimised by selecting an MA(4) with the given parameters. The true noise is indeed an MA(4) process (with paramaters $\theta_1=0.8$, $\theta_2=-0.6$, $\theta_3=0.1$, $\theta_4=0.4$).{\label{tbl:4}}}
\end{center}
\end{table}

We also tested our estimator using Monte Carlo simulations in
\cite{Sykulskietal2008} for a variety of MA(1) processes and the results showed a
significant reduction in error compared with not only the naive estimator, but also
the estimators based on a white-noise assumption. Furthermore, the adjusted
multiscale estimator performed almost identically to our multiscale estimator when
we set $\theta_1=0$ and recovered a white-noise process, meaning the loss in precision
from searching for a parameter unnecessarily was negligible (as to be expected for $q\ll N$). Notice also that in
Table \ref{tbl:4} there appears to be little loss in precision from estimating more
parameters in the MA$(4)$ process then is required as $\theta_p$ for $p>4$ is always
estimated to be very close to zero. This further demonstrates the robustness and precision of our
estimation technique.

\section{Conclusions}\label{sec:concl}
The problem of estimating the integrated stochastic volatility of an It\^{o} process from noisy
observations was studied in this paper. Unlike most
previous works on this problem, see~\cite{Zhang+2005, PavlSt06}, the method for estimating
the integrated volatility developed in this paper is based on the frequency domain
representation of both the It\^{o}  process and the noisy observations. The integrated
volatility can be represented as a summation of variation in the process of interest
over all frequencies (or scales). In our estimator we  adjust the raw sample
variance at each frequency. Such an estimator is truly multiscale, as it corrects
the estimated energy directly at every scale. In other words, the estimator is
debiased {\em locally at each frequency}, rather than globally.

To estimate the degree of scale separation in the data we used the Whittle likelihood, and
quantified the noise contribution by the multiscale ratio. Various properties of the
multiscale estimator were determined, see Theorems~\ref{thm:thmA} and~\ref{thm:thmB}.
As was illustrated by the set of examples, our estimator performs extremely well on
data simulated from the Heston model, and is competitive with the methods proposed
by \cite{Zhang+2005}, under varying signal-to-noise and sampling scenarios. The
proposed estimator is truly multiscale in nature and adapts automatically to
the degree of noise contamination of the data, a clear strength. It is also
easily implemented and computationally efficient.

The new estimator for the integrated stochastic volatility can be written as
\begin{eqnarray*}
\widehat{\langle X,X\rangle}&=&\sum_u \ell_{-u}
\sum_k \left(X_{t_{k-u}}-X_{t_{k-u-1}}\right)\left(X_{t_{k}}-X_{t_{k-1}}\right),
\end{eqnarray*}
where the kernel $\ell_u$ is given by \eqref{eqnsmooth}.
We can compare this estimator with kernel estimators, see \cite{Florens93}. There the estimated increment square $\Delta X_t^2$ is locally smoothed to estimate the diffusion coefficient using a kernel function, $K(\cdot)$.
Contrary to this approach we estimate the integrated volatility by smoothing the estimated autocovariance of $\Delta X_{t_j}$. In particular, we use a data-dependent choice of smoothing window. We show that, from a minimum bias perspective, using a Laplace window to smooth is optimal. This data-dependent choice of smoothing window becomes more interesting after relaxing the assumptions on the noise process, and treating correlated observation error.

Inference procedures implemented in the frequency domain are still very underdeveloped for problems with a multiscale structure. The modern data deluge has caused an excess of high frequency observations in a number of application areas, for example finance and molecular dynamics. More flexible models could also be used for the high frequency nuisance structure. In this paper we have introduced a new frequency domain based estimator and applied it to a relatively simple problem, namely the estimation of the integrated stochastic volatility, for data contaminated by high frequency noise. There are many extensions and potential applications of the new estimator. Here we list a few which seem interesting to us and which are currently under investigation.
\begin{itemize}
\item Study parameter estimation for noisily observed SDEs which are driven by more general noise processes, for example L\'{e}vy processes.
\item Application of the new estimator to the problem of statistical inference for fast/slow systems of SDEs, of the type studied in~\cite{PavlSt06, PapPavSt08}.
\item Study the combined effects of high-frequency and multiscale structure in the data. A first step in this direction was taken in~\cite{CotterPavl09}.
\end{itemize}

\bibliographystyle{siam}

\setcounter{section}{0}
\renewcommand{\thesection}{{\Alph{section}}}
\renewcommand{\thesubsection}{\Alph{section}.\arabic{subsection}}
\renewcommand{\thesubsubsection}{\Alph{section}.\arabic{subsection}.\arabic{subsubsection}}
\renewcommand{\theequation}{\Alph{section}-\arabic{equation}}
\section{Proof of Theorem \ref{thm:thmA} \label{sec:apA}}
Let the true value of $\bm{\sigma}$ be denoted $\bm{\sigma}^{\star}.$
We differentiate the multiscale energy likelihood function~\eqref{e:multis}
with respect to $\bm{\sigma}$ to obtain
\begin{alignat*}{1}
\dot{\ell}_X(\bm{\sigma})&=\frac{\partial \ell(\bm{\sigma})}{\partial \overline{\sigma}^2_{X}}
=-\sum_{k=1}^{N/2-1} \frac{1}{\overline{\sigma}^2_{X}+\sigma^2_{\varepsilon}\left|2 \sin(\pi
f_k\Delta t)\right|^2}+
\sum_{k=1}^{N/2-1}\frac{\widehat{\mathcal{S}}^{(Y)}_{kk}}{\left(\overline{\sigma}^2_{X}+\sigma^2_{\varepsilon}\left|2
\sin(\pi f_k\Delta t)\right|^2\right)^2}\\
\dot{\ell}_{\varepsilon}(\bm{\sigma})&=
\frac{\partial \ell(\bm{\sigma})}{\partial \sigma^2_{\varepsilon}}=-\sum_{k=1}^{N/2-1} \frac{\left|2 \sin(\pi f_k \Delta
t)\right|^2}{\overline{\sigma}^2_{X}+\sigma^2_{\varepsilon}\left|2 \sin(\pi f_k \Delta
t)\right|^2}+\sum_{k=1}^{N/2-1}\frac{\left|2 \sin(\pi f_k \Delta
t)\right|^2\widehat{\mathcal{S}}^{(Y)}_{kk}} {\left(\overline{\sigma}^2_{X}+ \sigma^2_{\varepsilon}\left|2
\sin(\pi f_k \Delta t)\right|^2\right)^2}.
\end{alignat*}
To remove implicit $\Delta t$ dependence we let $\tau_X=\overline{\sigma}^2_X/\Delta t,$ and denote derivatives with respect to $\tau_X$ by subscript $\tau$. Then $\dot{\ell}_{\tau}(\widehat{\bm{\sigma}})=\Delta t\dot{\ell}_X(\widehat{\bm{\sigma}})$, and so on.
We calculate the expectation and variance of the score functions evaluated at $\bm{\sigma}^{\star},$ and find that the bias of $\widehat{\tau}_X$ is order ${\mathcal{O}}\left(\Delta t^{1/2}\log(\Delta t)\right)$ and the bias of $\widehat{\sigma}^2_{\varepsilon}$ is order ${\mathcal{O}}\left(\Delta t^{2}\log(\Delta t)\right).$ These contributions become negligible, and are of lesser importance compared to the variance.

To show large sample properties we Taylor expand the multiscale likelihood with $\widehat{\bm{\sigma}}$ corresponding to the estimated maximum likelihood, and $\bm{\sigma}'$ is lying between 
$\widehat{\bm{\sigma}}$ and ${\bm{\sigma}}^{\star}$. Then
\begin{alignat}{1}
\nonumber
\dot{\ell}_\tau(\widehat{\bm{\sigma}})&=\dot{\ell}_{\tau}(\bm{\sigma}^{\star})+\ddot{\ell}_{\tau \tau}(\bm{\sigma}')
\left[\widehat{\sigma}^2_{X}-\sigma_{X}^{\star 2}\right]/\Delta t+\ddot{\ell}_{\tau\varepsilon}(\bm{\sigma}')
\left[\widehat{\sigma}^2_{\varepsilon}-\sigma_{\varepsilon}^{\star 2}\right]\\
\dot{\ell}_{\varepsilon}(\widehat{\bm{\sigma}})&=\dot{\ell}_{\varepsilon}(\bm{\sigma}^{\star})+\ddot{\ell}_{\varepsilon \tau}(\bm{\sigma}')
\left[\widehat{\sigma}^2_{X}-\sigma_{X}^{\star 2}\right]/\Delta t+\ddot{\ell}_{\varepsilon\varepsilon}(\bm{\sigma}')
\left[\widehat{\sigma}^2_{\varepsilon}-\sigma_{\varepsilon}^{\star 2}\right].
\nonumber
\end{alignat}
We note with the observed Fisher information 
\[\bm{F}=\left[\ddot{\ell}_{\tau \tau}(\bm{\sigma}')\quad\ddot{\ell}_{\tau\varepsilon}(\bm{\sigma}');
\ddot{\ell}_{\varepsilon \tau}(\bm{\sigma}')\quad\ddot{\ell}_{\varepsilon
\varepsilon}(\bm{\sigma}')\right]\] that
\begin{eqnarray}
\label{scoredist}
\begin{pmatrix}
(\widehat{\sigma}^2_{X}-\sigma_{X}^{\star 2})/\Delta t\\
\widehat{\sigma}^2_{\varepsilon}-\sigma_{\varepsilon}^{\star 2}
\end{pmatrix}=\bm{F}^{-1}\begin{pmatrix}\dot{\ell}_{\tau}(\widehat{\bm{\sigma}})-\dot{\ell}_{\tau}(\bm{\sigma}^{\star})\\
\dot{\ell}_{\varepsilon}(\widehat{\bm{\sigma}})-\dot{\ell}_{\varepsilon}(\bm{\sigma}^{\star})
\end{pmatrix}.
\end{eqnarray}
We henceforth ignore the term $J^{(\mu)}_k=J^{(X)}_k-\tilde{J}^{(X)}_k$ as this will not contribute to leading order, and write $J^{(X)}_k$ where formally we would write $\tilde{J}^{(X)}_k$ or ${J}^{(X)}_k$. We can observe the suitability of this directly from \eqref{e:multis}
and use bounds for $J^{(\mu)}_k$, where we could formally apply these to get bounds on each derivative of $l(\bm{\sigma})$ (note that we cannot differentiate bounds). To avoid needless technicalities, the details of this approach will not be reported.
To leading order
{\small
\begin{eqnarray*}
\var\left(\dot{\ell}_\tau({\bm{\sigma}}) \right)&=&\sum_{l=1}^{N/2-1}
\sum_{k=1}^{N/2-1}\frac{\Delta t^2\cov\left(
\widehat{\mathcal S}^{(Y)}_{kk},\widehat{\mathcal S}^{(Y)}_{ll}\right)}
{\left(\overline{\sigma}^2_{X}+ \sigma^2_{\varepsilon}\left|2 \sin(\pi f_k\Delta t)\right|^2\right)^2
\left(\overline{\sigma}^2_{X}+\sigma^2_{\varepsilon}\left|2 \sin(\pi f_l\Delta t)\right|^2\right)^2}\\
\var\left(\dot{\ell}_{\varepsilon}({\bm{\sigma}}) \right)&=&\sum_{l=1}^{N/2-1}
\sum_{k=1}^{N/2-1}
\frac{\left|2 \sin(\pi f_k\Delta t)\right|^2
\left|2 \sin(\pi f_l\Delta t)\right|^2\cov\left(\widehat{\mathcal S}^{(Y)}_{kk},
\widehat{\mathcal S}^{(Y)}_{ll}\right)}{\left(\overline{\sigma}^2_{X}+ \sigma^2_{\varepsilon}\left|2 \sin(\pi f_k\Delta t)\right|^2\right)^2\left(\overline{\sigma}^2_{X}+ \sigma^2_{\varepsilon}\left|2 \sin(\pi f_l\Delta t)\right|^2\right)^2}\\
\\
\cov\left(\dot{\ell}_\tau({\bm{\sigma}}),\dot{\ell}_{\varepsilon}({\bm{\sigma}}) \right)&=&\sum_{l=1}^{N/2-1}
\sum_{k=1}^{N/2-1}\frac{\Delta t
\left|2 \sin(\pi f_l\Delta t)\right|^2
\cov\left(\widehat{\mathcal S}^{(Y)}_{kk},\widehat{\mathcal S}^{(Y)}_{ll}\right)}
{\left(\overline{\sigma}^2_{X}+ \sigma^2_{\varepsilon}\left|2 \sin(\pi f_k\Delta t)\right|^2\right)^2\left(\overline{\sigma}^2_{X}+\sigma^2_{\varepsilon}
\left|2 \sin(\pi f_l\Delta t)\right|^2\right)^2}.
\end{eqnarray*}}
We now need to calculate $\cov\left(\widehat{\mathcal S}^{(Y)}_{kk},\widehat{\mathcal S}^{(Y)}_{ll}\right)$ which is
\begin{eqnarray}
\nonumber
\cov\left(\widehat{\mathcal S}^{(Y)}_{kk},\widehat{\mathcal S}^{(Y)}_{ll}\right)&=&
\E\left\{J_k^{(Y)} [J^{(Y)}_k]^* [J_l^{(Y)}]^* J_l^{(Y)}\right\}-\E\left\{\widehat{\mathcal S}^{(Y)}_{kk}\right\}
\E\left\{\widehat{\mathcal S}^{(Y)}_{ll}\right\}\\
&=&\rho_{kl}^{(Y)}{\mathcal S}^{(Y)}_{kk}{\mathcal S}^{(Y)}_{ll}.
\label{covarss}
\end{eqnarray}
Furthermore
\begin{eqnarray*}
&&\E\left\{J_k^{(Y)} [J^{(Y)}_k]^* [J_l^{(Y)}]^* J_l^{(Y)}\right\}-
\E\left\{J_k^{(Y)} [J^{(Y)}_k]^*\right\}\E\left\{ [J_l^{(Y)}]^* J_l^{(Y)}\right\}
\\
&=&
\E\left\{(J_k^{(X)}+J_k^{(\varepsilon)}) [(J_k^{(X)}+J_k^{(\varepsilon)})]^* [(J_l^{(X)}+J_l^{(\varepsilon)})]^* (J_l^{(X)}+J_l^{(\varepsilon)})\right\}\\
&&-
\E\left\{J_k^{(Y)} [J^{(Y)}_k]^*\right\}\E\left\{ [J_l^{(Y)}]^* J_l^{(Y)}\right\}
\\
&=&\cov\left\{\widehat{\mathcal S}^{(X)}_{kk}, \widehat{\mathcal S}^{(X)}_{ll} \right\}+
\cov\left\{\widehat{\mathcal S}^{(\varepsilon)}_{kk}, \widehat{\mathcal S}^{(\varepsilon)}_{ll} \right\}+{\mathcal S}^{(X)}_{kl}{\mathcal S}^{(\varepsilon)*}_{kl}+
{\mathcal S}^{(X)*}_{kl}{\mathcal S}^{(\varepsilon)}_{kl}.
\end{eqnarray*}
We therefore need to calculate the individual terms of this expression. We note
\[\cov\left\{\widehat{\mathcal S}^{(\varepsilon)}_{kk}, \widehat{\mathcal S}^{(\varepsilon)}_{ll} \right\}=\delta_{kl}[{\mathcal S}^{(\varepsilon)}_{kk}]^2,\quad
{\mathcal S}^{(X)}_{kl}{\mathcal S}^{(\varepsilon)*}_{kl}+
{\mathcal S}^{(X)*}_{kl}{\mathcal S}^{(\varepsilon)}_{kl}=2\delta_{kl}
{\mathcal S}^{(X)}_{kk}{\mathcal S}^{(\varepsilon)}_{kk}.\]
Then it follows
\begin{equation}\label{e:cov_per}
\cov\left\{\widehat{\mathcal S}^{(Y)}_{kk},\widehat{\mathcal S}^{(Y)}_{ll}\right\}=
\cov\left\{\widehat{\mathcal S}^{(X)}_{kk},\widehat{\mathcal S}^{(X)}_{ll}\right\}+
\delta_{kl}[{\mathcal S}^{(\varepsilon)}_{kk}]^2+2\delta_{kl}
{\mathcal S}^{(X)}_{kk}{\mathcal S}^{(\varepsilon)}_{kk}.
\end{equation}
We therefore only need to worry about $\cov\left\{\widehat{\mathcal S}^{(X)}_{kk}, \widehat{\mathcal S}^{(X)}_{ll} \right\}.$
We need
\begin{eqnarray*}
&&\E\left\{ J_k^{(X)} [J^{(X)}_k]^* [J_l^{(X)}]^* J_l^{(X)}\right\}=\frac{1}{N^2}
\E\left\{\sum_{n=1}^{N}\int_{(n-1)\Delta t}^{n\Delta t}\sigma_{s} dW_{s}e^{-2i\pi \frac{k n}{N}}\right.\\
&&\left.\times
\sum_{p=1}^{N}\int_{(p-1)\Delta t}^{p\Delta t}\sigma_{t} dW_{t}e^{2i\pi \frac{k p}{N}}
\sum_{m=1}^{N}\int_{(m-1)\Delta t}^{m\Delta t}\sigma_{u} dW_{u}e^{-2i\pi \frac{l m}{N}}
\sum_{w=1}^{N}\int_{(w-1)\Delta t}^{w\Delta t}\sigma_{v} dW_{v}e^{2i\pi \frac{l w}{N}}\right\}  \\ && =: \frac{1}{N^2} \sum_{n =1}^N \sum_{p =1}^N  \sum_{m =1}^N  \sum_{\rho =1}^N \left(e_{k n} e^*_{k p} e_{\ell m}^* e_{\ell \rho}
\E \Big( M_n M_p M_m M_{\rho} \Big) \right),
\end{eqnarray*}
where $M_n := \int_{(n-1) \Delta t}^{n \Delta t} \sigma_s \, dW_s$ and $e_{k
n} := e^{- \frac{2 i \pi k n}{N}}$. Since Brownian motion has independent
increments, we have that $\E \Big( M_n M_p M_m M_{\rho} \Big) = \E M_n^4  \; {\mathrm{if}} \;n=p=m=\rho$,  $\E \Big( M_n M_k M_m M_{\rho} \Big) = \E M_n^2 \E M_k^2 $ if $n=k, \; m=\rho$ and $\E \Big( M_n M_p M_m M_{\rho} \Big) = 0 $, otherwise.
Consequently,
\begin{eqnarray*}
\E\left\{ J_k^{(X)} [J^{(X)}_k]^* [J_l^{(X)}]^* J_l^{(X)}\right\}  &=& \frac{1}{N^2}
\sum_{n=1}^N \ E M_n^4 + \frac{1}{N^2}  \left( \sum_{n=1}^N \ E M_n^2 \right)^2
\\ && +  \frac{1}{N^2} \sum_{n=1}^N \sum_{p=1}^N e_{kn} e_{\ell n}^* e_{k
p}^* e_{\ell p} \ E M_n^2  \ E M_p^2  \\
&&+  \frac{1}{N^2} \sum_{n=1}^N \sum_{p=1}^N e_{kn} e_{\ell p}^* e_{k p}^* e_{\ell n} \ E M_n^2  \ E M_p^2 
 \end{eqnarray*}
We use standard bounds on moments of stochastic integrals~\cite{KSh91} to
obtain the bound
$$
\frac{1}{N^2} \sum_{n=1}^N \ E M_n^4 \leq \frac{1}{N^2} \sum_{n=1}^N 36 \Delta t \int_{(n-1) \Delta t}^{n \Delta t} \E \sigma_s^4 \, ds \leq C (\Delta t)^3,
$$
since, by assumption, $\E \sigma^4_s = \mathcal{O}(1)$\footnote{$C$ in this paper denotes a generic constant, rather than the same constant.}. 
We
have:
\begin{eqnarray*}
\nonumber
\rho_{kl}^{(X)}{\mathcal S}^{(X)}_{kk}{\mathcal S}^{(X)}_{ll}
 & = & \E\left\{ J_k^{(X)} [J^{(X)}_k]^* [J_l^{(X)}]^* J_l^{(X)}\right\}
 - \E \left| J_k^{(X)} \right|^2 \E \left| J_l^{(X)} \right|^2
\\ &=&
\frac{1}{N^2}\int_0^T \int_0^T
\left(\cos(2\pi (k+l)(\frac{s-u}{T}))+\cos(2\pi (k-l)(\frac{s-u}{T})) \right) \\\nonumber
&&\times \E\left\{\sigma_{s}^2\right\} \E\left\{\sigma_{u}^2\right\} ds du  + \mathcal{O}((\Delta t)^3) \\
 &=&\frac{1}{2N^2}\int_0^T \int_0^T
\E\left\{\sigma_{s}^2\right\} \E\left\{\sigma_{u}^2\right\}
\left(e^{2i\pi (k+l)(\frac{s-u}{T})}+e^{-2i\pi (k+l)(\frac{s-u}{T})}\right.\\
&&\left.
+e^{2i\pi (k-l)(\frac{s-u}{T})}+e^{-2i\pi (k-l)(\frac{s-u}{T})} \right)  ds du + \mathcal{O}((\Delta t)^3)\\
&=&\frac{1}{2N^2}\left(\Sigma(-\frac{k+l}{T})\Sigma(\frac{k+l}{T})+
\Sigma(\frac{k+l}{T})\Sigma(-\frac{k+l}{T})\right.\\
&&\left.+\Sigma(-\frac{k-l}{T})\Sigma(\frac{k-l}{T})+
\Sigma(\frac{k-l}{T})\Sigma(-\frac{k-l}{T}) \right) + \mathcal{O}((\Delta t)^3).
\end{eqnarray*}
Since $\E \sigma^2_t$ is a smooth function of time we can bound the decay of $\Sigma(f)\propto \frac{1}{f}$ so that:
\begin{eqnarray}
\label{ordercorr}
\rho_{kl}^{(X)}{\mathcal S}^{(X)}_{kk}{\mathcal S}^{(X)}_{ll}&=&
\Delta t^2 \left(\mathcal{O}\left(\frac{1}{(k+l)^2}\right)+
\mathcal{O}\left(\frac{1}{(k-l)^2}\right) \right).
\end{eqnarray}
We combine the foregoing calculations with~\eqref{e:cov_per}
\[\var\left\{\widehat{\mathcal S}^{(Y)}_{kk} \right\}=
\left({\mathcal S}^{(X)}_{kk}+{\mathcal S}^{(\varepsilon)}_{kk} \right)^2.\]
{\small \begin{eqnarray}
\label{varscore1}
\var\left(\dot{\ell}_\tau(\widehat{\bm{\sigma}}) \right)&=&\sum_{l=1}^{N/2-1}
\sum_{k=1}^{N/2-1}\frac{\Delta t^2\cov\left(
\widehat{\mathcal S}^{(Y)}_{kk},\widehat{\mathcal S}^{(Y)}_{ll}\right)}
{\left(\overline{\sigma}^2_{X}+ \sigma^2_{\varepsilon}\left|2 \sin(\pi f_k\Delta t)\right|^2\right)^2
\left(\overline{\sigma}^2_{X}+ \sigma^2_{\varepsilon}\left|2 \sin(\pi f_l\Delta t)\right|^2\right)^2}.
\end{eqnarray}}
We note that 
\[\cov\left(
\widehat{\mathcal S}^{(Y)}_{kk},\widehat{\mathcal S}^{(Y)}_{ll}\right)=
\rho_{kl}^{(X)}{\mathcal S}^{(X)}_{kk}{\mathcal S}^{(X)}_{ll}+\delta_{kl}\left[
{\mathcal S}^{(\varepsilon)}_{ll}\right]^2+2\delta_{kl}{\mathcal S}^{(X)}_{kk}
{\mathcal S}^{(\varepsilon)}_{kk}.\]
Thus it follows that:
\begin{eqnarray}
\label{varscore2}
\var\left(\dot{\ell}_\tau(\widehat{\bm{\sigma}}) \right)&=&\sum_{l=1}^{N/2-1}\frac{\Delta t^2}
{
\left(\overline{\sigma}^2_{X}+ \sigma^2_{\varepsilon}\left|2 \sin(\pi f_l\Delta t)\right|^2\right)^2}+C+\mathcal{O}(\log(\Delta t)\Delta t^{-1/4})\\
&=&\mathcal{O}(\Delta t^{-1/2})+C+\mathcal{O}(\log(\Delta t)\Delta t^{-1/4}).
\nonumber
\end{eqnarray}
The extra order terms acknowledge potential effects from the drift.
We need to establish the size of $C$. Using \eqref{e:cov_per} we find that:
\begin{eqnarray*}
|C|
&\le&\sum_{l\neq k}^{N/2-1}\frac{\Delta t^4C_2 ( (k+l)^{-2}+(k-l)^{-2})}
{\left(\overline{\sigma}^2_{X}+ \sigma^2_{\varepsilon}\left|2 \sin(\pi f_k\Delta t)\right|^2\right)^2
\left(\overline{\sigma}^2_{X}+ \sigma^2_{\varepsilon}\left|2 \sin(\pi f_l\Delta t)\right|^2\right)^2}\\
&\le &2\sum_{k=1}^{N/2-1}\sum_{\tau=1}^k\frac{\Delta t^4C_2 ( (2k-\tau)^{-2}+\tau^{-2})}
{\left(\overline{\sigma}^2_{X}+ \sigma^2_{\varepsilon}\left|2 \sin(\pi f_{k-\tau}\Delta t)\right|^2\right)^2
\left(\overline{\sigma}^2_{X}+ \sigma^2_{\varepsilon}\left|2 \sin(\pi f_k\Delta t)\right|^2\right)^2}\\
&\sim &2\sum_{k=1}^{N/2-1}\sum_{\tau=1}^k\frac{C_2 ( (2k-\tau)^{-2}+\tau^{-2})}
{\left({\tau}^2_{X}+ \sigma^2_{\varepsilon}\left|2 \sin(\pi f_{k-\tau}\Delta t)\right|^2/\Delta t\right)^2
\left({\tau}^2_{X}+ \sigma^2_{\varepsilon}\left|2 \sin(\pi f_k\Delta t)\right|^2/\Delta t\right)^2}\\
&=&{\mathcal{O}}(\log(\Delta t)).
\end{eqnarray*}
This is negligible in size compared to $\Delta t^{-1/2}$. 
Similar calculations can bound contributions from the off diagonals in the other two calculations. 
Also as $\overline\sigma^2_X=\tau_X \Delta t$
\begin{alignat}{1}
\label{eqn:fisherinfo2}
-\E\left\{\ddot{\ell}_{\tau \tau}(\bm{\sigma})\right\}&
=
\sum_{k=1}^{N/2-1}\frac{ \Delta t^2}{\left(\overline{\sigma}^2_{X}+ \sigma^2_{\varepsilon}\left|2 \sin(\pi f_k\Delta t)\right|^2\right)^2}+\mathcal{O}(\log(\Delta t))={\mathcal{O}}(\Delta t^{-1/2})\\
\nonumber
-\E\left\{\ddot{\ell}_{\varepsilon \varepsilon}(\bm{\sigma})\right\}
&=\sum_{k=1}^{N/2-1}\frac{\left|2 \sin(\pi f_k )\right|^4}
{\left(\overline{\sigma}^2_{X}+ \sigma^2_{\varepsilon}\left|2 \sin(\pi f_k \Delta t)\right|^2\right)^2}+\mathcal{O}(\log(\Delta t))={\mathcal{O}}(\Delta t^{-1})\\
\nonumber
-\E\left\{\ddot{\ell}_{\tau \varepsilon}(\bm{\sigma})\right\}&=\sum_{k=1}^{N/2-1}\frac{\Delta t\left|2 \sin(\pi f_k )\right|^2
}
{\left(\overline{\sigma}^2_{X}+ \sigma^2_{\varepsilon}\left|2 \sin(\pi f_k \Delta t)\right|^2\right)^2}+{\mathcal{O}}(\log(\Delta t))={\mathcal{O}}(\Delta t^{-1/2}).
\end{alignat}
The order terms follow from usual spectral theory on the white noise process, as well as bounds on $J_k^{(\mu)}$.
We can also by considering the variance of the observed Fisher information deduce that renormalized versions of the entries of the observed Fisher information converge in probability to a constant, or 
\[{\mathrm{diag}}(\Delta t^{1/4},\Delta t^{1/2}){\bm{F}}{\mathrm{diag}}(\Delta t^{1/4},\Delta t^{1/2})\longrightarrow  {\bm{\mathcal F}},\]
and thus using Slutsky's theorem we can deduce that:
\[{\mathrm{diag}}(\Delta t^{-1/4},\Delta t^{-1/2})\left[\begin{pmatrix}
\widehat{\sigma}^2_{X}/\Delta t\\
\widehat{\sigma}^2_{\varepsilon}
\end{pmatrix}-\begin{pmatrix}
\overline\sigma_{X}^{\ast 2}/\Delta t\\ \sigma^{\ast 2}_{\varepsilon}
\end{pmatrix}\right]{\mathrm{diag}}(\Delta t^{-1/4},\Delta t^{-1/2})\overset{L}{\longrightarrow}
N\left(\bm{0},\bm{\mathcal F}^{-1}\right),\]
where the entries of $\bm{\mathcal F}$ can be found from \eqref{eqn:fisherinfo2}, \eqref{varscore1} and \eqref{varscore2}, and
\begin{eqnarray*}
&&{\mathrm{diag}}(\Delta t^{-1/4},\Delta t^{-1/2})\var\left\{\left[\begin{pmatrix}
\widehat{\sigma}^2_{X}/\Delta t\\
\widehat{\sigma}^2_{\varepsilon}
\end{pmatrix}-\begin{pmatrix}
\overline\sigma_{X}^{\ast 2}/\Delta t\\ \sigma^{\ast 2}_{\varepsilon}
\end{pmatrix}\right]\right\}{\mathrm{diag}}(\Delta t^{-1/4},\Delta t^{-1/2})\\
&=&{\mathrm{diag}}(\Delta t^{-1/4},\Delta t^{-1/2})\bm{F}^{-1}\bm{F}\bm{F}^{-1}{\mathrm{diag}}(\Delta t^{-1/4},\Delta t^{-1/2})\\
&=& {\mathrm{diag}}(\Delta t^{-1/4},\Delta t^{-1/2})\bm{F}^{-1}{\mathrm{diag}}(\Delta t^{-1/4},\Delta t^{-1/2})\longrightarrow  {\bm{\mathcal F}}^{-1}.
\end{eqnarray*}
We have
\begin{eqnarray}
{\bm{\mathcal F}}&=&\begin{pmatrix}
\frac{T}{\sigma_{\varepsilon}16}
\frac{1}{\tau_X^{3/2}} & 0\\
0 & \frac{2T}{\sigma^4_{\varepsilon}}
\end{pmatrix}=\begin{pmatrix}
{\cal I}_{\tau \tau}  & 0\\
0 & {\cal I}_{\varepsilon \varepsilon}
\end{pmatrix}.
\label{eq:Fmatrix}
\end{eqnarray}
This expression follows by direct calculation. Asymptotic normality of both $\widehat{\tau}_x$ and $\widehat{\sigma}_{\varepsilon}^2$ follows by the usual arguments. We can determine the asymptotic variance of $\widehat{\langle X,X \rangle}^{(w)}$ via
\begin{eqnarray}
\nonumber
\var\left\{\widehat{\langle X,X \rangle}^{(w)}\right\}&=&T^2\var\left\{\widehat{\tau}_x\right\}\\
&=&T \frac{\sigma_{\varepsilon}}{\tau_X^{1/2}}16  \tau_X^{2}\sqrt{\Delta t}.
\label{varwhittle}
\end{eqnarray}
We see that the variance depends on the length of the time course, the inverse of the signal to noise ratio, the square root of the sampling period and the fourth power of the ``average standard deviation'' of the $X_t$ process.



\section{Proof of Theorem \ref{thm:thmB} \label{sec:apB}}
We now wish to use these results to deduce properties of $\widehat{\sigma}$. Firstly using the well known {\em invariance} of maximum likelihood estimators to transfer the estimators of $\overline{\sigma}_X^2$ and $\sigma^2_{\varepsilon}$ to estimators of $\langle X, X \rangle_T$. We therefore take
\begin{eqnarray}
\nonumber
\widehat{\langle X, X\rangle}_T^{(m_1)}=
\sum_{k=0}^{N-1}\widehat{\mathcal S}^{(X)}_{kk}(\widehat{L}_k)=\sum_{k=0}^{N-1}
\widehat{L}_k\widehat{\mathcal S}^{(Y)}_{kk}\end{eqnarray}
It therefore follows that with $\widehat{\tau}_X=\tau_X+\delta\tau_X$
and $\widehat{\sigma}_\varepsilon^2={\sigma}_\varepsilon^2+\delta\sigma^2_\varepsilon$
{\small \begin{eqnarray*}
\E\left\{\widehat{\langle X, X\rangle}_T^{(m_1)} \right\}&=&
\sum_{k=0}^{N-1}\E\left\{
\left(\frac{\overline{\sigma}^2_{X}+\delta \sigma^2_{X}}{
\overline{\sigma}^2_{X}+\delta \sigma^2_{X}+(
\sigma^{ 2}_{\varepsilon}+\delta \sigma^2_{\varepsilon})\left|2 \sin(\pi f_k\Delta t)\right|^2}\right)\widehat{\mathcal S}^{(Y)}_{kk}\right\}\\
&=&
\sum_{k=0}^{N-1}\E\left\{
\left(\frac{\overline{\sigma}^2_{X}+\delta \sigma^2_{X}}{
1+\frac{\left[\delta \sigma^2_{X}+
\delta \sigma^2_{\varepsilon}\left|2 \sin(\pi f_k\Delta t)\right|^2\right]}{\overline{\sigma}^2_{X}+\sigma^{ 2}_{\varepsilon}\left|2 \sin(\pi f_k\Delta t)\right|^2}}\right)\frac{\widehat{\mathcal S}^{(Y)}_{kk}}{\overline{\sigma}^2_{X}+\sigma^{ 2}_{\varepsilon}\left|2 \sin(\pi f_k\Delta t)\right|^2}
\right\}\\
&=&
\sum_{k=0}^{N-1}\E\left\{
\left(\overline{\sigma}^2_{X}+\delta \sigma^2_{X}\right)\left(
1-\frac{\left[\delta \sigma^2_{X}+
\delta \sigma^2_{\varepsilon}\left|2 \sin(\pi f_k\Delta t)\right|^2\right]}{(\overline{\sigma}^2_{X}+\sigma^{ 2}_{\varepsilon}\left|2 \sin(\pi f_k\Delta t)\right|^2)}\right)\frac{\widehat{\mathcal S}^{(Y)}_{kk}}{\overline{\sigma}^2_{X}+\sigma^{ 2}_{\varepsilon}\left|2 \sin(\pi f_k\Delta t)\right|^2}
\right\}\\
&=&\sum_{k=0}^{N-1}[\overline{\sigma}^2_{X}+O\left( \Delta t^{5/4}\right)]+{\mathcal{O}}\left(\sqrt{\Delta t}\log(\Delta t)\right)=\E\left\{\langle X, X\rangle_T\right\}+{\mathcal{O}}\left(\sqrt{\Delta t}\log(\Delta t)\right)\\
&&+{\mathcal{O}}\left(\sqrt[4]{\Delta t}\right).
\end{eqnarray*}}
This implies that the estimator is asymptotically unbiased.
We can also note that the variance of the new estimator is given by: 
\begin{eqnarray*}
\var\left\{ \widehat{\langle X, X\rangle}_T^{(m_1)}\right\}
&=&\sum_j \sum_k \cov\{\widehat{L}_j\widehat{\mathcal{S}}_{jj}^{(Y)}, \widehat{L}_k\widehat{\mathcal{S}}_{ll}^{(Y)}\}\\
&=&\sum_j \sum_k \cov\{\frac{\widehat{L}_j}{L_j} L_j\widehat{\mathcal{S}}_{jj}^{(Y)}, \frac{\widehat{L}_k}{L_k} L_k\widehat{\mathcal{S}}_{kk}^{(Y)}\}\\
&=&\sum_j \sum_k \cov\left\{\left(1+\frac{\delta \tau_X}{\tau_X}-\frac{\delta \tau_X\Delta t+\delta{\sigma}^2_{\varepsilon}
\left|2\sin(\pi f_j \Delta t)\right|^2}{\tau_X\Delta t+{\sigma}^2_{\varepsilon}
\left|2\sin(\pi f_j \Delta t)\right|^2} \right) L_j\widehat{\mathcal{S}}_{jj}^{(Y)},\right.\\
\nonumber
&&\left. \left(1+\frac{\delta \tau_X}{\tau_X}-\frac{\delta \tau_X\Delta t+\delta{\sigma}^2_{\varepsilon}
\left|2\sin(\pi f_k \Delta t)\right|^2}{\tau_X\Delta t+{\sigma}^2_{\varepsilon}
\left|2\sin(\pi f_k \Delta t)\right|^2}  \right) L_k\widehat{\mathcal{S}}_{kk}^{(Y)}\right\}.
\end{eqnarray*}
Then
\begin{eqnarray*}
\var\left\{ \widehat{\langle X, X\rangle}_T^{(m_1)}\right\}
&=&\sum_j \sum_k \left\{\cov\{ L_j\widehat{\mathcal{S}}_{jj}^{(Y)}, L_k\widehat{\mathcal{S}}_{kk}^{(Y)}\}+
\cov\{\frac{\delta \tau_X}{\tau_X} L_j\widehat{\mathcal{S}}_{jj}^{(Y)}, L_k\widehat{\mathcal{S}}_{kk}^{(Y)}\}
 \right.\\
\nonumber
&&+
\cov\{ L_j\widehat{\mathcal{S}}_{jj}^{(Y)}, \frac{\delta \tau_X}{\tau_X}L_k\widehat{\mathcal{S}}_{kk}^{(Y)}\}\\
\label{terms3}
&&-\cov\{\frac{\delta \tau_X\Delta t+\delta{\sigma}^2_{\varepsilon}
\left|2\sin(\pi f_j \Delta t)\right|^2}{\tau_X\Delta t+{\sigma}^2_{\varepsilon}
\left|2\sin(\pi f_j \Delta t)\right|^2} L_j\widehat{\mathcal{S}}_{jj}^{(Y)}, L_k\widehat{\mathcal{S}}_{kk}^{(Y)}\}\\
\label{terms4}
&&-\cov\{ L_j\widehat{\mathcal{S}}_{jj}^{(Y)},\frac{\delta \tau_X\Delta t+\delta{\sigma}^2_{\varepsilon}
\left|2\sin(\pi f_k \Delta t)\right|^2}{\tau_X\Delta t+{\sigma}^2_{\varepsilon}
\left|2\sin(\pi f_k \Delta t)\right|^2} L_k\widehat{\mathcal{S}}_{kk}^{(Y)}\}\\
\label{terms5}
&&\left.+\cov\{\frac{\delta \tau_X}{\tau_X} L_j\widehat{\mathcal{S}}_{jj}^{(Y)},\frac{\delta \tau_X}{\tau_X} L_k\widehat{\mathcal{S}}_{kk}^{(Y)}\}+\dots\right\}\\
&=&\sum_j \sum_k \left\{\delta_{jk}\sigma^4_X+L_j L_k\cov\{\frac{\delta \tau_X}{\tau_X} \widehat{\mathcal{S}}_{jj}^{(Y)}, \widehat{\mathcal{S}}_{kk}^{(Y)}\}+\dots\right\}
\end{eqnarray*}
By looking at the individual terms of this expression, and noting that the estimated renormalized variance $\widehat{\tau}_X=\tau_X+\delta\tau_X$
and $\widehat{\sigma}_\varepsilon^2={\sigma}^\varepsilon_X+\delta\sigma^2_\varepsilon$
are linear combinations of $\widehat{\mathcal{S}}_{kk}^{(Y)},$ we can deduce the stated order terms, by again noting the $\sqrt{\Delta t}$ order of the important terms. However to leading order, this estimator will perform identically to $\widehat{\langle X,X \rangle}^{(w)}$ in terms of variance.

\section{Proof of Time Domain Form}\label{sec:time}
The integral can be calculated from first principles using complex-variables with $z=e^{2i\pi f}$. Thus $dz/df=2i\pi z$ or $df= dz/(2i\pi z)$. \eqref{eqnsmooth}
takes the form
\begin{eqnarray}
\ell_{\tau}&=&\frac{1}{2i \pi}\oint_{|z|=1}
\frac{\overline{\sigma}_X^2}{\overline{\sigma}^2_X z-\sigma^2_{\varepsilon}[z-1]^2}\; z^{\tau}\;dz.
\end{eqnarray}
We need the poles, or:
\begin{eqnarray*}&&\overline{\sigma}^2_X z-\sigma^2_{\varepsilon}[z-1]^2=0\;\Longleftrightarrow
\;
z=1+\frac{\overline{\sigma}^2_X}{2 \sigma^2_{\varepsilon}}\pm \sqrt{\frac{\overline{\sigma}^2_X}{ \sigma^2_{\varepsilon}}+\frac{\overline{\sigma}^4_X}{4\sigma^4_{\varepsilon}}}=z^{\pm}\end{eqnarray*}
If $\left|\frac{\sigma^2_\varepsilon}{\overline{\sigma}^2_{X}}
\right|<1$ we have
\begin{eqnarray*}
z^-&=&1+\frac{\overline{\sigma}^2_X}{2 \sigma^2_{\varepsilon}}- \frac{\overline{\sigma}_X^2}{2 \sigma^2_{\varepsilon}}\sqrt{1+\frac{4 \sigma^2_\varepsilon}{\sigma^2_{X}}}\\
&&=  1+\frac{\overline{\sigma}_X^2}{2 \sigma^2_{\varepsilon}}- \frac{\overline{\sigma}_X^2}{2 \sigma^2_{\varepsilon}}\left(1+\frac{1}{2}
\frac{4 \sigma^2_\varepsilon}{\sigma^2_{X}}+\frac{1}{4}\frac{(-1)}{2}\left[\frac{4 \sigma^2_\varepsilon}{\sigma^2_{X}}\right]^2+{\mathcal{O}}\left(\frac{ \sigma^6_\varepsilon}{\sigma^6_{X}} \right)\right)\\
&=&\frac{ \sigma^2_\varepsilon}{\sigma^2_{X}}+{\mathcal{O}}\left(\frac{ \sigma^4_\varepsilon}{\sigma^4_{X}} \right)\\
z^+&=&\frac{\sigma^2_{X}}{\sigma^2_\varepsilon}+\dots\end{eqnarray*}
We then note that:
\begin{eqnarray*}
\ell_{\tau}&=&-\frac{1}{2i \pi}\oint_{|z|=1}
\frac{\overline{\sigma}_X^2/(\sigma^2_{\varepsilon})}{-\overline{\sigma}_X^2/(\sigma^2_{\varepsilon}) z+[z-1]^2}\; z^{\tau}\;dz=-\frac{1}{2i \pi}\oint_{|z|=1}
\frac{\overline{\sigma}_X^2/(\sigma^2_{\varepsilon})}{(z-z^-)(z-z^+)}\; z^{\tau}\;dz\\
&=&\frac{2i \pi}{2i \pi}\overline{\sigma}_X^2/(\sigma^2_{\varepsilon})\frac{\left(\frac{ \sigma^2_\varepsilon}{\sigma^2_{X}}\right)^{\tau}}{z^+-\frac{ \sigma^2_\varepsilon}{\sigma^2_{X}}+{\mathcal{O}}\left(\frac{ \sigma^4_\varepsilon}{\sigma^4_{X}} \right)}
=\left(\frac{ \sigma^2_\varepsilon}{\sigma^2_{X}}\right)^{\tau}+{\mathcal{O}}\left(\frac{ \sigma^{2\tau+2}_\varepsilon}{\overline{\sigma}^{2\tau+2}_{X}} \right)
\end{eqnarray*}
If on the other hand you consider $\left|\frac{ \sigma^2_\varepsilon}{\sigma^2_{X}}
\right|>1$ which in many scenarios is more realistic then we find that:
\begin{eqnarray*}
z^-&=&1+\frac{\overline{\sigma}_X^2}{2 \sigma^2_{\varepsilon}}- \frac{\overline{\sigma}_X}{ \sigma_{\varepsilon}}\sqrt{1+\frac{\sigma^2_{X}}{4\Delta t \sigma^2_\varepsilon}}=  1+\frac{\overline{\sigma}_X^2}{2 \sigma^2_{\varepsilon}}- \frac{\overline{\sigma}_X}{ \sigma_{\varepsilon}}\left(1+\frac{1}{2}\frac{\sigma^2_{X}}{4 \sigma^2_\varepsilon}\right)\\
&=&1-\frac{\overline{\sigma}_X}{ \sigma_{\varepsilon}}
+{\mathcal{O}}\left(\frac{\overline{\sigma}_X^2}{\sigma^2_{\varepsilon}}\right)\\
z^+&=&1+\frac{\overline{\sigma}_X}{ \sigma_{\varepsilon}}
\end{eqnarray*}
In this case we find that
\begin{eqnarray*}
\ell_{\tau}&=&\overline{\sigma}_X^2/(\sigma^2_{\varepsilon})\frac{[1-\frac{\overline{\sigma}_X}{ \sigma_{\varepsilon}}]^{\tau}}{2\frac{\overline{\sigma}_X}{ \sigma_{\varepsilon}}
+{\mathcal{O}}\left(\frac{\overline{\sigma}_X^2}{\sigma^2_{\varepsilon}}\right)}=
\frac{\overline{\sigma}_X}{2\sigma_{\varepsilon}}\left(1-\frac{\overline{\sigma}_X}{ \sigma_{\varepsilon}}\right)^{\tau}+{\mathcal{O}}\left(\frac{\overline{\sigma}_X^{2}}{2\sigma_{\varepsilon}^2}
\left(1-\frac{\overline{\sigma}_X}{ \sigma_{\varepsilon}}\right)^{\tau}\right) .
\end{eqnarray*}

In both cases the decay of the filter is geometric. We note that in most practical examples $L_k$ decays very rapidly in $k$. Therefore, we do not need to integrate between $-1/2$ to $1/2$, and only need to integrate over $-1/\pi$ to $1/\pi$. In this range of $f$ we find that for smallish remainder term $R_3$ we have:
$\sin^2(\pi f )= \pi^2 f^2 +R_3(f \pi)$.
Then we note 
\begin{eqnarray*}
\ell_{\tau}&=&\int_{-\frac{1}{\pi}}^{\frac{1}{\pi}}
\frac{\overline{\sigma}_X^2}{\overline{\sigma}_X^2+4\sigma^2_{\varepsilon}\pi^2 f^2 +R_3(f\pi)}\;e^{2i\pi f \tau }\;df+C\\
&=&\frac{\overline{\sigma}_X}{2\sigma_{\varepsilon}}\int_{-\infty}^{\infty}
\left[2\frac{\overline{\sigma}_X/(\sigma_{\varepsilon})}{\frac{\sigma^2_X}{\sigma^2_{\varepsilon}}+ 4\pi^2 f^2 }+R_4(f)\right]\;e^{2i\pi f \tau }\;df+C\\
&=&\frac{\overline{\sigma}_X}{2\sigma_{\varepsilon}}e^{-\frac{\overline{\sigma}_X|\tau|}{
\sigma_{\varepsilon}}}+C.
\end{eqnarray*}
Thus we are smoothing the autocovariance sequence with a smoothing window that becomes a delta function as $\overline{\sigma}_X/\sigma_{\varepsilon}\rightarrow \infty$. It is reasonable that this non-dimensional quantity arises as an important factor.
\end{document}

%% file: texheader.tex
\newcommand{\E}{\ensuremath{\mathrm{E}}}

\newcommand{\var}{\ensuremath{\mathrm{var}}}

